\documentclass[lettersize,journal]{IEEEtran}
\usepackage{amsmath,amsfonts}

\usepackage{array}
\usepackage[caption=false,font=normalsize,labelfont=sf,textfont=sf]{subfig}
\usepackage{textcomp}
\usepackage{stfloats}
\usepackage{url}
\usepackage{verbatim}
\usepackage{graphicx}
\usepackage{cite}
\usepackage{amsfonts}
\usepackage{amsmath,amsthm}
\newtheorem{theorem}{Theorem}
\usepackage[caption=false,font=normalsize,labelfont=sf,textfont=sf]{subfig}

\usepackage{ulem}
\usepackage{graphicx}
\usepackage{tcolorbox}
\usepackage{threeparttable}
\usepackage{booktabs}
\usepackage{adjustbox}
\usepackage{hyperref}
\usepackage{cleveref}
\usepackage{enumitem}
\usepackage{textcomp}
\usepackage{diagbox}
\usepackage{textcomp}
\usepackage{textcomp} 
\usepackage{algorithm}  
\usepackage{algpseudocode}
\usepackage{tikz}
\usepackage[T1]{fontenc}
\usepackage{amssymb}

\newtheorem{definition}{Definition}
\newtheorem{proposition}{Proposition}

\usepackage{amsmath}
\usepackage{color}
\definecolor{green}{rgb}{0, 0.5, 0}
\definecolor{orange}{rgb}{0.8, 0.6, 0.2}
\definecolor{orange2}{rgb}{1.0, 0.6, 0.2}
\definecolor{red}{rgb}{1.0, 0.0, 0.0}
\definecolor{teal}{rgb}{0.0, 0.4, 0.4}
\definecolor{purple}{rgb}{0.65,0,0.65}
\definecolor{saffron}{rgb}{0.95,0.75,0.2}
\definecolor{turquoise}{rgb}{0.0,0.5,0.5}
\definecolor{black}{rgb}{0.0, 0.0, 0.0}
\definecolor{gray}{rgb}{0.5, 0.5, 0.5}

\crefname{section}{§}{§§}
\Crefname{section}{§}{§§}

\newenvironment{packeditemize}{
\begin{list}{$\bullet$}{
\setlength{\labelwidth}{8pt}
\setlength{\itemsep}{0pt}
\setlength{\leftmargin}{\labelwidth}
\addtolength{\leftmargin}{\labelsep}
\setlength{\parindent}{0pt}
\setlength{\listparindent}{\parindent}
\setlength{\parsep}{0pt}
\setlength{\topsep}{3pt}}}{\end{list}}

\newcommand{\Tref}[1]{Table~\ref{#1}}
\newcommand{\Eref}[1]{Eq.~(\ref{#1})}
\newcommand{\Fref}[1]{Fig.~\ref{#1}}

\usepackage{caption}

\usepackage{multirow}
\usepackage{pifont}
\usepackage{booktabs}
\usepackage{cleveref}
\usepackage{bbding}
\crefname{section}{§}{§§}
\usepackage{soul}
\usepackage{color}
\usepackage{xcolor}
\usepackage{colortbl}
\usepackage{setspace}
\usepackage{enumitem}
\usepackage{tcolorbox}
\usepackage{hhline}
\usepackage{changepage}
\usepackage{tabu}
\usepackage{subfig}
\usepackage{blindtext}
\usepackage{tcolorbox}
\usepackage{graphicx}
\usepackage{wrapfig}
\usepackage{amsthm}

\definecolor{orange2}{rgb}{1.0, 0.6, 0.2}

\algnewcommand\algorithmicinitialization{\textbf{Initialization:}}
\algnewcommand\Initialization{\item[\algorithmicinitialization]}

\hypersetup{
  colorlinks   = true,    % Colours links instead of ugly boxes
  urlcolor     = black,    % Colour for external hyperlinks
  linkcolor    = black,    % Colour of internal links
  citecolor    = red      % Colour of citations
}

\newcommand{\Name}{\texttt{SwitchPatch}\xspace}

\begin{document}

\title{SwitchPatch: Physical Adversarial Attack Strategy with Switchable Adversarial Objectives}

\author{Hanrui~Jiang$^{*}$,~
Yutong~Wu$^{*}$,~
Shiyi~Yao,~
Chen~Ling,~
Xingshuo~Han$^{\dagger}$,~
Hangcheng~Liu,~
Xinyi~Huang$^{\dagger}$,~\IEEEmembership{Member,~IEEE}
Tianwei~Zhang,~\IEEEmembership{Member,~IEEE}
  \thanks{$^{*}$ Hanrui Jiang and Yutong Wu contributed equally to this work and are co-first authors.}
  \thanks{$^{\dagger}$ Xingshuo Han and Xinyi Huang are the corresponding authors.}
  \IEEEcompsocitemizethanks{\IEEEcompsocthanksitem Hanrui Jiang is with Artificial Intelligence Thrust, Information Hub, Hong Kong University of Science and Technology (Guangzhou), Guangzhou 511455, China.
  \IEEEcompsocthanksitem Yutong Wu, Hangcheng Liu and Tianwei Zhang are with the School of Computer Science and Engineering, Nanyang Technological University, Singapore 639815.
  \IEEEcompsocthanksitem Shiyi Yao and Chen Ling are with the School of Cyber Science and Engineering, Wuhan University, Wuhan, China.
  \IEEEcompsocthanksitem Xingshuo Han and Xinyi Huang are with the College of Computer Science and Technology/College of Software, Nanjing University of Aeronautics and Astronautics, 211106, Nanjing, China.(e-mail: xingshuo.han@nuaa.edu.cn, xyhuang81@gmail.com)

  }
}

\markboth{}{}


\maketitle

\begin{abstract}

Physical adversarial patch (PAP) attacks attach carefully crafted patches to physical objects to manipulate a deployed model. 
However, existing PAP attacks suffer from several limitations. First, existing patches remain continuously active, which prevents selective targeting of specific attack objectives and compromises stealth. Second, these approaches require target device access or hardware configuration knowledge, and often rely on costly external equipment.

To address these limitations, this paper introduces \Name, a novel physical adversarial attack strategy that employs a physically static adversarial patch yet can be triggered to produce dynamic and controllable attack effects. 
Unlike existing approaches, \Name can transition between states through predefined triggers, enabling adaptation to dynamic environments.
Moreover, to improve stealth, we design two trigger patterns: one overlapping with the patch and another spatially separated from it.
These triggers can be implemented at low cost without target device access or hardware configuration knowledge.

We make three contributions. First, we provide theoretical and empirical analysis to establish the feasibility of \Name and characterize the number of attack objectives it can support. Second, we develop a gradient-based framework for static yet switchable attacks through diverse trigger patterns. Third, we conduct extensive Unmanned Ground Vehicle (UGV) experiments to validate the effectiveness, transferability, and robustness of \Name.

\end{abstract}


% \vspace{-10pt}
\section{Introduction}

\IEEEPARstart{A}{s} machine-learning systems are increasingly deployed in real-world applications, physical attacks against these systems pose growing security risks. In particular, physical adversarial attacks especially in the form of physical adversarial patches (PAPs) ~\cite{eykholt2018robust} have emerged as an important security concern. These attacks can manipulate model predictions at inference time without access to the training process. Specifically, these patches are carefully crafted perturbations embedded into physical objects, capable of consistently deceiving the target models across diverse environments.  For example, in autonomous driving scenarios, such patches can induce traffic sign misclassifications and erroneous depth estimation, potentially leading to severe consequences.

\label{sec:intro}

\begin{figure}[t]
	\centering
	    % \hspace{2mm} % For alignment purpose
		\begin{minipage}[t]{0.32\linewidth}
			\centering
   \subfloat[Benign]{
			\includegraphics[width=\linewidth,trim={0 0 0 0mm},clip]{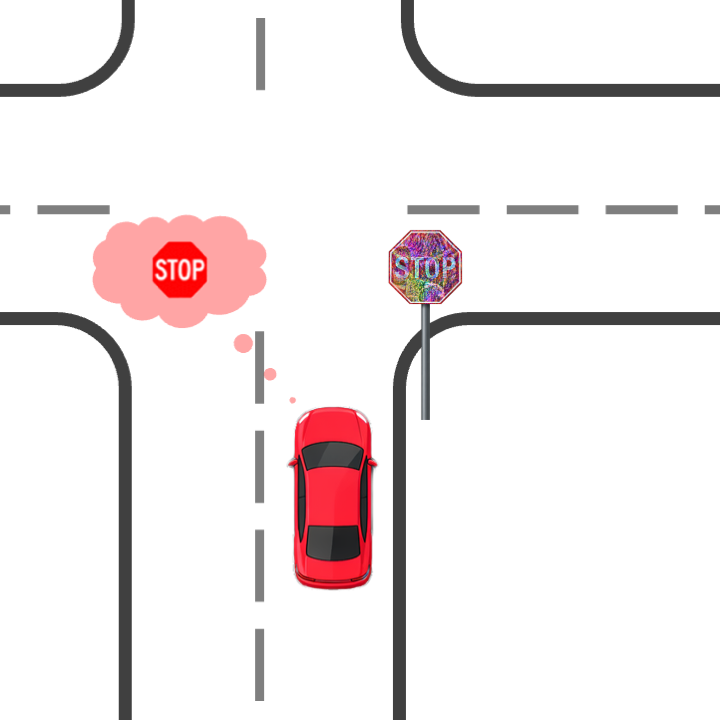}}
		\end{minipage}
  % \hspace{-5mm}
		\begin{minipage}[t]{0.32\linewidth}
			\centering
   \subfloat[Overlapping]{
			\includegraphics[width=\linewidth,trim={0 0 0 0mm},clip]{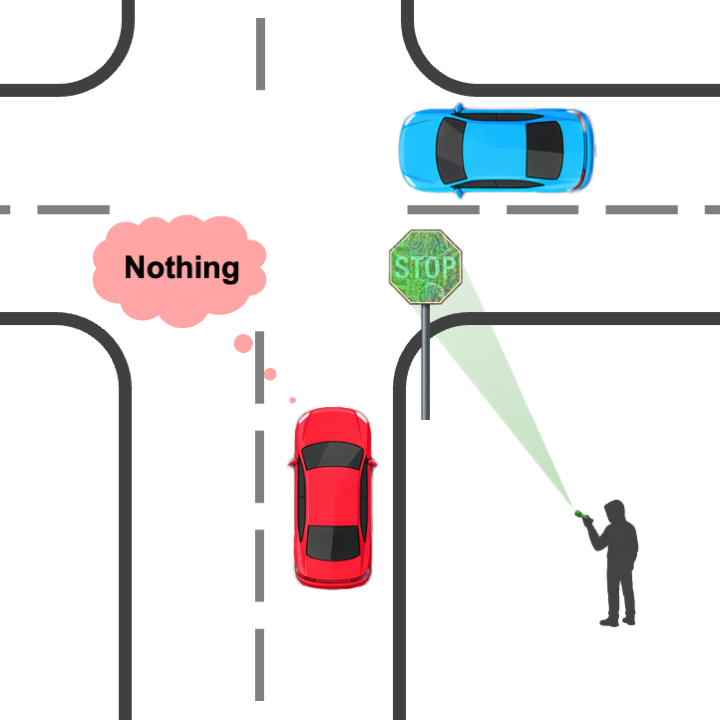}}
		\end{minipage}
  % \hspace{-5mm}
		\begin{minipage}[t]{0.32\linewidth}
			\centering
   \subfloat[Separate]{
			\includegraphics[width=\linewidth,trim={0 0 0 0mm},clip]{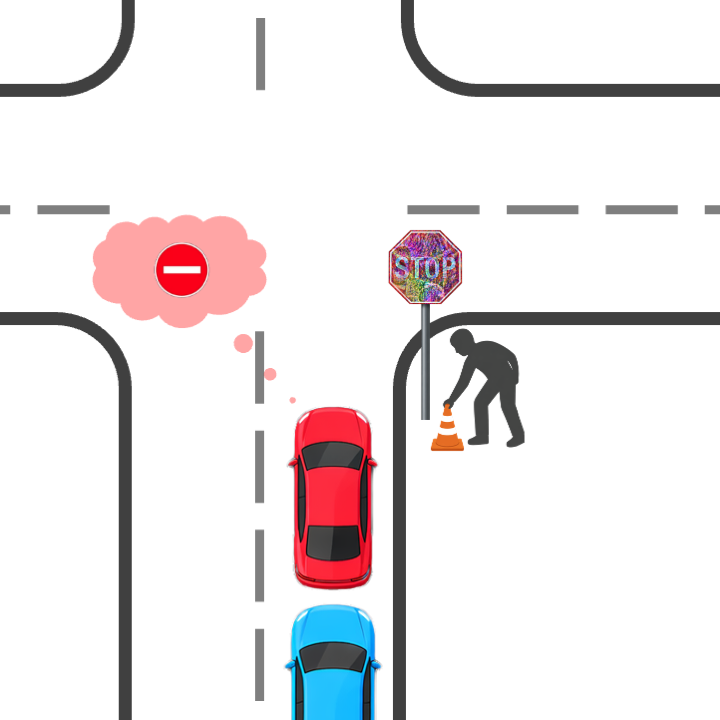}}
		\end{minipage}
\vspace{-5pt}
\caption{(a) \Name is benign to vehicles under normal conditions; (b) \Name causes the hiding attack (HA) when the green light is projected on; (c) \Name causes the misclassification attack (MA). The stop sign is detected as No Passing when a traffic cone is deployed}
\vspace{-20pt}
\label{fig:intro}
% \vspace{-20pt}
\end{figure}

Despite their effectiveness, most existing PAPs suffer from limitations. Once deployed, a patch remains continuously active. In practical settings such as autonomous driving, it may therefore affect any autonomous vehicle (AV) that observes it. This rigid attack model cannot adapt to dynamic real-world scenarios. Moreover, such indiscriminate behavior increases the risk of exposure and detection~\cite{chou2020sentinet, xiang2021patchguard}. To address the issues, Zhu et al. proposed TPatch~\cite{zhu2023tpatch}.
TPatch differs from conventional PAPs by using an acoustic trigger: it remains benign until the trigger disrupts the camera’s stabilization pipeline, causing motion blur and model misclassification. This selective activation enables attacks only under desired conditions, such as when a target AV is present, improving adaptability while reducing detectability.
However, TPatch has practical limitations. First, the adversary must obtain a camera of the same model as the one used by the target vehicle. The adversary also needs detailed information about the camera and sensor. Second, TPatch relies on costly hardware, such as an arbitrary waveform generator, an audio amplifier, and ultrasonic transducers, to generate the acoustic trigger. Third, this trigger generation hardware must be positioned in close proximity to the target.

To address these limitations, we propose \Name, a novel physical adversarial attack strategy. Here we present an example of autonomous driving in Fig. \ref{fig:intro}. By adopting different visual triggers, \Name achieves diverse attack goals under varying traffic conditions.
\Name addresses shortcomings of prior work. 
Unlike prior work, \Name does not require any hardware information about the victim vehicle and supports both white-box and black-box attack settings. Because the triggers are purely visual, \Name does not require expensive equipment. Furthermore, unlike TPatch, \Name does not require a trigger generator in close proximity to the target.

Additionally, \Name also extends the functionality of PAPs in two ways. First, it can adapt to different attack goals in real time. 
\Name enables dynamic behavior using a static patch. 
The patch can transition between states under predefined visual triggers, such as lighting changes, nearby objects, or occlusions. 
In addition to enabling transitions from a benign state to an adversarial state and back to the benign state, it further supports switching among attack goals through different triggers.
Second, it also supports flexible trigger placement through two trigger patterns, which broadens applicable scenarios and can improve stealth.
In the first pattern, the patch and the trigger overlap. For example, a specific lighting condition illuminates the patch and activates the attack. In the second pattern, the patch and the trigger are separated. For example, an object such as a traffic cone is placed near the patch and serves as the trigger.

There are a couple of non-trivial challenges to realizing this PAP framework. Motivated by these challenges, our work delivers the following contributions:

\textbf{(1) Theory: feasibility and capability of dynamic behavior from a static patch.} 
\ul{\textit{Determining whether a single static adversarial patch can reliably support multiple dynamic attack objectives under predefined environmental conditions without degrading attack effectiveness.}}
Existing studies have not investigated the feasibility of inducing dynamic attacks from a static adversarial patch by manipulating environmental conditions. This gap raises two questions: first, whether a single static patch is capable of achieving multiple attack goals under different predefined conditions; and second, how many goals it can support before its effectiveness deteriorates.

We answer these questions with an in-depth theoretical analysis. To begin with, we prove the feasibility of our method. More specifically, by leveraging 
Cooperative Game Theory~\cite{brandenburger2007cooperative}, we prove that the context information can affect the model prediction.
In parallel, by leveraging the Weierstrass Extreme Value Theorem~\cite{martinez2014weierstrass}, we prove that an adversarial patch can achieve different goals by applying different conditions (e.g., colored light). Finally, our analysis reveals a fundamental trade-off: supporting more attack goals makes it harder to find effective perturbations.

\textbf{(2) Method: a static but switchable and visually stealthy PAP.}  
\ul{\textit{Designing a static yet switchable physical adversarial patch that can achieve different attack goals under predefined conditions while remaining practical and stealthy.}}
This challenge involves three key requirements. First, the patch must be physically static, yet able to switch between different attack goals when predefined conditions are met. Second, the method must be cost-effective so that it can be deployed in practice. Third, both the patch and the trigger conditions should appear visually natural and should not attract human attention.

To address these challenges, we develop a gradient-based optimization framework. First, to enable a static yet switchable attack, we design a condition-oriented loss that drives the patch to exhibit different behaviors under predefined conditions while remaining benign otherwise. Second, to support practical deployment, we design multiple visual trigger patterns that can be flexibly instantiated in different real-world scenarios at low cost. Third, to improve visual naturalness, we introduce a joint loss that encourages the patch to resemble the target object, thereby reducing its visual salience.

\textbf{(3) Evaluation: reliable activation and physical-world robustness.}  
\ul{\textit{Ensuring that the patch reliably achieves the intended attack objective under predefined conditions in real-world environments.}}
This challenge involves two key requirements. First, the predefined conditions must consistently activate the intended objective without triggering incorrect behavior or introducing errors. Second, the patch must remain effective under physical-world noise, such as changes in illumination and weather conditions.

To address the challenge, we conduct extensive evaluation and robustness-oriented training. We evaluate \Name on four tasks: 
classification, object detection, monocular depth estimation, and semantic segmentation.
To improve robustness in the physical world, we incorporate color shifting and intensity adjustments during optimization via Expectation over Transformation (EoT)~\cite{eykholt2018robust}. We further validate \Name in both simulation and real-world experiments using an unmanned ground vehicle (UGV). The evaluation includes 3 object detectors, 5 image classifiers for traffic sign recognition, 3 semantic segmentation models, and 4 CNN-based or transformer-based models for monocular depth estimation. We test both white-box and black-box settings to demonstrate robustness and adaptability across attack scenarios.

In summary, this paper advances physical adversarial attacks by presenting, to our knowledge, the first static but switchable PAP that supports multiple attack goals under predefined conditions. We provide a theoretical foundation for \Name and validate its practicality through extensive simulation and real-world experiments.
Notably, \Name represents more than just a patch: it embodies an innovative and versatile attack strategy, which is not restricted by any specific conditions or tasks.

% \vspace{-10pt}
\section{Background}
\label{sec:bg}

\subsection{Physical Adversarial Patch}

Eykholt et al. \cite{eykholt2018robust} proposed $RP_2$ for robust physical adversarial attacks on traffic sign classification. Through graffiti-mimicking patches, their attack causes targeted misclassifications of traffic signs. Similarly, Duan et al. \cite{duan2021adversarial} proposed AdvLB, an adversarial attack method that leverages laser beams as physically realizable adversarial perturbations. Zhao et al. \cite{zhao2019seeing} developed FIR, ERG and nested-AE to enable robust physical adversarial attacks on real-world object detectors, achieving stable attack effectiveness across long distances and wide angles. Ji et al. \cite{ji2021poltergeist} presented Poltergeist acoustic attacks that manipulate inertial sensors to induce image blurs and lead to object misclassifications, while Lovisotto et al. \cite{lovisotto2021slap} proposed SLAP, a novel projector-based adversarial perturbation technique, enabling robust physical attacks in the autonomous driving scenario. Nassi et al. \cite{nassi2020phantom} reveal split-second phantom attacks on commercial ADASs, validating real-world attacks on Tesla and Mobileye 630, and propose GhostBusters, a robust vision-based countermeasure. Additionally, Sato et al. \cite{sato2024invisible} propose an invisible infrared laser reflection attack to mislead CAV traffic sign recognition, and Zhu et al. \cite{zhu2023tpatch} proposed TPatch, an acoustic-triggered physical adversarial patch to attack autonomous driving vision models, verifying its effectiveness on mainstream detectors through simulation and real-world experiments. Beyond classification and detection, Cheng et al. \cite{cheng2022physical} develop a robust physical-object-oriented adversarial patch attack against monocular depth estimation, Zheng et al. \cite{zheng2024pi} further propose $\pi$-Jack, a black-box physical adversarial attack on monocular depth estimation via perspective hijacking, and Liu et al. \cite{liubeware} introduce AdvRM, an adversarial patch attack that physically decouples patches from obstacles via road deployment to mislead autonomous driving MDE models.

% \vspace{-5mm}
\subsection{Threat Model}
\label{sec:threatmodel}

% We describe our threat model from the attack scenarios, goals,
% % requirements, 
% and adversary's capabilities. 

\subsubsection{Attack scenario}

Following prior work on PAP \cite{eykholt2018robust,zhao2019seeing}, we consider an attacker who generates \Name and places it in a digital image or on a physical object. \Name remains benign without the trigger, but activates malicious behavior or switches targets when triggered. We consider two trigger patterns.

% \smallskip
\noindent\textbf{Separate Trigger Architecture.}
This pattern uses a contextual cue spatially separated from \Name. We refer to this pattern as a separate trigger because the trigger and the patch usually appear in different image regions. When a predefined context appears, \Name can switch between benign and malicious behaviors or change the attack target.

For example, an attacker may place \Name on a stop sign. Under normal conditions, the patched sign appears benign to passing vehicles. When a separate trigger, such as a nearby traffic cone, appears, the victim model may fail to detect the sign, or label it as another object.
% in semantic segmentation.

% \smallskip
\noindent\textbf{Overlapping Trigger Architecture.}
This pattern uses a physical condition applied directly to \Name, so the trigger and the patch occupy the same image region. Once activated, \Name can switch between benign and malicious behavior or change targets.

For example, an attacker may place \Name on a stop sign and activate it by projecting light of a specific color onto the patch.
This can cause the victim model to miss or misclassify the sign. In depth estimation, \Name attached to an obstacle can use the same trigger to alter perceived depth, making the obstacle appear closer or farther than it actually is.

\subsubsection{Attack goal}
In classification, object detection, and depth estimation, the adversary may pursue five attack goals: (1) misclassification attack (MA), causing a traffic sign classifier or detector to predict the wrong sign; (2) hiding attack (HA), causing a detector to miss the target; (3) appearing attack (AA), causing a classifier to recognize a nonexistent object; (4) far attack (FA), increasing the estimated depth of an obstacle;
and (5) near attack (NA), decreasing the estimated depth of an obstacle.

\subsubsection{Adversary capability}

We assume the adversary has access to the real attack environment and can predefine multiple attack goals that \Name is designed to induce. After generating \Name, the adversary can place it at suitable locations, such as traffic signs, objects, or the road. The adversary then can leverage predefined conditions 
and has access to the triggering context that he is going to use to decide when to activate or switch the attacks according to the actual physical scenarios. 

We consider both white-box and black-box settings. In the white-box setting, the adversary knows the target model architecture, parameters, and training hyperparameters. In the black-box setting, the adversary has no knowledge of the victim model. We further assume the adversary has no information about the camera used by the victim system, including its brand, resolution, and lens.

% \vspace{-3mm}
\section{Theoretical Analysis}

\label{sec:methodology}

In our framework, we design two trigger patterns, the overlapping trigger and the separate trigger architectures.
% \vspace{-12pt}

\subsection{Separate Trigger Architecture}
\subsubsection{Problem Formulation}

We consider a normal victim object $x$. The attacker aims to generate the corresponding \Name and blend it with $x$. We denote $\mathcal{X}$ as the space of input images, $\mathcal{Y}$ as the output space of the target model $f: \mathcal{X} \rightarrow \mathcal{Y}$. An adversarial object $x^{'}$ is defined as: 
% \vspace{-5pt}
\begin{equation}
    x' := (1-r) \cdot x + r\cdot \delta,
% \vspace{-5pt}
\end{equation}

where $\delta$ is the perturbation applied to the victim object $x$, and $r\in(0,1]$ is the blend ratio. Formally, an effective \Name is defined as follows:
% \vspace{-3pt}
\begin{equation}
\left\{
\begin{array}{l}
    f(x') = y_1 \\
    f(x' \oplus (c, P)) = y_2
\end{array}
\right.
\quad \text{s.t. } y_1 \neq y_2
\label{eq:ctae_formulation}
% \vspace{-3pt}
\end{equation}
$\oplus$ stands for the operation of placing the triggering context $c$ at the position $P$, which can be either physical or digital. $y_1$, $y_2$ represent two distinct model outputs. 
The goal is to find the optimal perturbation $\delta$ to make the model output them with and without the triggering context, respectively.

Inspired by~\cite{wangunified}, which exploits the cooperative game theory to analyze the in-group interaction of the conventional AE, we propose modeling \Name as a cooperative game consisting of two distinct player groups: the adversarial unit group $\mathcal{A}=\{a_1,...,a_n\}$ and the context feature group $\mathcal{C}=\{c_1,...,c_m\}$. 
The elements in each group comprise combinations of pixels in \Name and the context object respectively. 
The reward of the game $v(\cdot)$, on the other hand, can be regarded as the loss of the model output w.r.t. the target output $y_t$. 
In this case, the interaction index between the adversarial group and context feature group can be formally written as:
% \vspace{-3pt}
\[
\begin{array}{cl}
    I_{cross} & = \mathbb{E}_{ij}[I_{a_i c_j}]\\
    & = \frac{1}{mn}\sum_{i=0}^n\sum_{j=0}^m\big\{\phi(S_{ij}|\mathcal{P}) \\
    & -[\phi(a_i|\mathcal{P}\backslash\{c_j\})+ \phi(c_j|\mathcal{P}\backslash\{a_i\})]\big\},
    \label{eq:Icross}
\end{array}
% \vspace{-3pt}
\]
where $S_{ij}=\{a_i\}\cup\{c_j\}$, $\mathcal{P}=\mathcal{A}\cup\mathcal{C}$. 
Based on the equations in Appendix~\ref{app:sta},
% \ref{app:sta}
the cross interaction index can be further simplified into the following form:
% \vspace{-5pt}
\begin{equation}
I_{cross}'= \mathbb{E}_{ij}[\phi_{a_i, w/ c_j}-\phi_{a_i,w/ o.c_j}]
% \vspace{-5pt}
\end{equation}
The Shapley Interaction Index is therefore differentiable and can be directly maximized by gradient descent. 

\subsubsection{Feasibility Proof}
\label{sec:proof2}

In this work, we mainly focus on attacking end-to-end deep learning models~\cite{glasmachers2017limits, silver2017predictron}. These models are designed to process raw input data and produce a final output prediction by applying gradient-based learning as a whole, without requiring any intermediate steps, hand-crafted features, or separate processing modules. Their formal definition is given below: 
\begin{definition}
\label{def:1}
\textit{ (End-to-end model) Let $f: \mathcal{X} \rightarrow \mathcal{Y}$ be a continuous projection. For $\forall{x}\subseteq \mathcal{X}$, $y=f(x)$. $f$ is an end-to-end model if and only if the Jacobian matrix $\frac{\partial{y}}{\partial{x}}$ exists}.
\end{definition}

Definition~\ref{def:1} illustrates the end-to-end learning from the aspect of the gradient. The existence of the Jacobian matrix $\frac{\partial{y}}{\partial{x}}$ to any part of the input ensures the entire model can be trained by utilizing the gradient-based learning. It also guarantees some special characteristics of end-to-end models, especially that the prediction to one part of the input may be influenced by the other:

% \vspace{-5pt}
\begin{proposition}  \textit{Let $f: \mathcal{X} \rightarrow \mathcal{Y}$ be an end-to-end nonlinear model. $x$ is a given object in the input. Let $\mathcal{S}_{x} = \{ i | x_i \neq 0 \}$ be the player set of $x$,  $\hat{x}$ is the unit from different parts of the image such that $\mathcal{S}_{x} \cap\ \mathcal{S}_{\hat{x}}= \emptyset$. Then $\exists{\hat{x}} $ s.t. $ \exists{p} \in \mathcal{S}_{x}$ and $ \exists{q} \in \mathcal{S}_{\hat{x}}$, we have $|I_{pq}| > 0$.} 
\label{prop:1}
\end{proposition}

% \vspace{-10pt}
\renewcommand{\qedsymbol}{}
\begin{proof}
We assume the contrary of this proposition and derive a contradiction.
Assume that $\forall{\hat{x}} \in \mathcal{X}$, and  $ \forall{p} \in \mathcal{S}_{x}$ and $ \forall{q} \in \mathcal{S}_{\hat{x}}$,  $|I_{pq}|=0$. 
Without loss of generality, we consider a single player unit $\hat{x} = \delta_q$ that $\mathcal{S}_{\hat{x}} = \{q\}$. The Shapley Interaction Index is:
$$I_{pq}=\mathbb{E}_{pq}[\phi_{x_q, w/ x_p}-\phi_{x_q,w/ o.x_p}]$$

Let $k=\frac{|S|!(n-|S|-2)!}{(n-1)!}$, where,

\begin{align*}
    \phi_{x_q, w/ x_p} &=\sum_{S\subseteq\Omega\backslash\{x_p,x_q\}}{{k\cdot(v(S\cup\{x_p,x_q\})-v(S\cup\{x_p\}))} }\\
                                &= \sum_{S\subseteq\Omega\backslash\{x_p,x_q\}}{{k\cdot(\nabla_{x_q}{f(S\cup{x_p})})}\delta_q}
\end{align*}

Similarly,
% \vspace{-5pt}
\begin{align*}
    \phi_{x_q, w/ o. x_p}  &=\sum_{S\subseteq\Omega\backslash\{x_q\}}{{k\cdot(v(S\cup\{x_q\})-v(S))} }\\
                                    &= \sum_{S\subseteq\Omega\backslash\{x_q\}}{k\cdot(\nabla_{x_q}{f(S)})\delta_q}  \\
                                    &= \sum_{S\subseteq\Omega\backslash\{x_p,x_q\}}{k\cdot(\nabla_{x_q}{f(S)})\delta_q}+ \\ 
                                    & \sum_{S\subseteq\Omega\backslash\{x_p,x_q\}}{k\cdot(\nabla_{x_q}{f(S\cup x_p)})\delta_q} \;(as\:x_p\cup x_q=\emptyset)\\
% \vspace{-5pt}
\end{align*}
As we assume $|I_{pq}|=0$, we have:
% \vspace{-5pt}
\begin{equation}
    I_{pq} = \sum_{S\subseteq\{\Omega\}\backslash\{x_p, x_q\}}{k\cdot(\nabla_{x_q}{f(S)})\delta_q} = 0
\label{eq:proof1}
% \vspace{-5pt}
\end{equation}
Given the model $f$ is nonlinear, $\exists{\nabla_{x_q}{f(S)}}\neq\mathbf{0}$ ($\nabla_{x_q}{f(S)}$ varies with different values of $S$ according to the nonlinear conditions). On the other hand,~\Eref{eq:proof1} holds for $\forall{\delta_q}\in\mathbb{R}$.  Therefore, $\nabla_{x_q}{f(S)}=\mathbf{0}, \forall{S}$, which contradicts the definition of the nonlinear model.
\end{proof}
% \vspace{-5pt}

\noindent Proposition~\ref{prop:1} indicates that \textbf{for any end-to-end model, the context information can also contribute to the final prediction of the model}. This property renders the feasibility of \Name, as demonstrated in the following proposition. For easy presentation, we assume $f$ is a multi-label classification model. The proposition can be extended to other types of models as well. Based on Proposition~\ref{prop:1}, the solution of \Name can be found given some special constraints about the model, as narrated in Proposition~\ref{prop:2}.

\begin{proposition} (See proof in Appendix~\ref{app:prop2})
% ~\ref{app:prop2})
\textit{Let $f: \mathcal{X} \rightarrow \mathcal{Y}$ be a well-trained end-to-end model that $f_y(x)>\beta$ and $f_y(x+c)>\beta$,  $\beta\in\mathbb{R}^+$ stands for the classification threshold. $\delta_c$ represents the adversarial pixels of \Name using the single-step gradient information, given as $\delta_c = \alpha \nabla_{x}(\mathcal{L}(x+c) - \mathcal{L}(x))$. $c$ is the triggering context. The 
perturbation budget is $\alpha$. If $\|\nabla_{x}{\mathcal{L}(x+c)}\|^2_2 \ge\frac{-\log{\beta}}{\alpha}$ and $\nabla_{x}^T{\mathcal{L}(x+c)}\nabla_{x}{\mathcal{L}(x)}\le0$ then $\exists{\delta_c}$ that satisfies~\Eref{eq:ctae_formulation}}.
\label{prop:2}
\end{proposition}  
% \vspace{-5pt}

Proposition~\ref{prop:2} indicates that, although the end-to-end nature of the model ensures that predictions can be influenced by pixels from different regions of the input, achieving a \Name attack requires additional considerations. The prerequisite given also varies the performance of \Name on different models.
% \sout{, as demonstrated in~\Cref{sec: main result}.}

% \vspace{-10pt}
\subsection{Overlapping Trigger Architecture}

% (brief description, figure)
\subsubsection{Problem Formulation}

Let $X$ denote the input image space, and $Y$ denote the output of the target model $f: X \rightarrow Y$. An adversarial patch $x'$ is defined as $x + \delta$, where $\delta$ is the perturbation applied to $x$. The patch $x'$ should satisfy the following objectives: 

\begin{packeditemize}

\item Under normal conditions, the patch does not exhibit an adversarial effect:
% \vspace{-5pt}
$$f(x + \delta) = f(x) = y $$

\item The patch is activated to be malicious under certain pre-defined conditions. Specifically, the attacker establishes some pairs of $(cl_k, y_k)$, where $cl_k$ is a special condition and $y_k$ is the goal the attacker wants to induce. This is formulated as:
% \vspace{-5pt}
$$f(x + \delta + cl_k) = y_k $$

\item Stealthiness: $x'$ should be close enough to $x$ to evade human inspections, i.e.,
% \vspace{-5pt}
$$|| x' - x ||_{p} \le \epsilon $$

\end{packeditemize}

\subsubsection{Feasibility Proof}
\label{sec:proof}

We denote the solution space for each attack goal $y_k$ as $S_k$. To find a perturbation $\delta^*$ that satisfies all the attack goals, the following condition must be met:
% \vspace{-5pt}
\[
\delta \in S = X_\epsilon \cap S_1 \cap S_2 \cap \dots \cap S_N
% \vspace{-5pt}
\]
where $X_\epsilon = \{ \delta \mid \|\delta\|_p \leq \epsilon \}$ represents the set of perturbations constrained by $\epsilon$.
This problem can also be described as a constrained optimization problem:
% \vspace{-5pt}
\begin{align}
\label{eq:general_for}
\delta^* = \arg\min_{\delta \in \Delta} \left( \sum_{k=1}^{N} \mathcal{L}(f(x + \delta + cl_k), y_k) \right) + \mathcal{L}(f(x + \delta), y)
% \vspace{-5pt}
\end{align}
where $\Delta = \{ \delta \mid \|\delta\|_p \leq \epsilon \}$ is the set of perturbations that satisfy the $L_p$-norm constraint, and $\mathcal{L}$ is a loss function.
% (e.g., cross-entropy).

%\noindent\textbf{Optimal Solution.}
Based on the  Weierstrass Extreme Value Theorem~\cite{martinez2014weierstrass}, any continuous function on a compact set must attain a maximum and minimum value. To apply this theorem, we assume that the loss functions $\mathcal{L}(f(x + \delta), y)$ and $\mathcal{L}(f(x + \delta + cl_k), y_k)$ are continuous. Additionally, the constraint set $\Delta = \{ \delta \mid \|\delta\|_p \leq \epsilon \}$ is compact, as it is both bounded and closed. Therefore, there must exist an optimal solution $\delta^*$ within the set $\Delta$ that satisfies the objective function.

%\noindent\textbf{Local optimality.}
Due to the non-convex nature of the neural network loss function $\mathcal{L}$, directly solving the problem to obtain a global optimum may not always be possible. However, finding a local optimum is still meaningful in the context of adversarial attacks, where a satisfactory solution is often sufficient. We analyze the local optima using the Karush-Kuhn-Tucker (KKT) condition. 
We define the Lagrange function as:
% \vspace{-5pt}
\[
\mathcal{L}(\delta, \lambda) = \sum_{k=1}^{N} \mathcal{L}(f(x + \delta + cl_k), y_k) + \mathcal{L}(f(x + \delta), y) + \lambda (\|\delta\|_p - \epsilon)
% \vspace{-5pt}
\]
According to the KKT conditions, there exists a multiplier $\lambda \geq 0$ such that:
% \vspace{-5pt}
\[
\nabla_{\delta} \mathcal{L}(\delta, \lambda) = 0, \quad \|\delta\|_p \leq \epsilon, \quad \lambda (\|\delta\|_p - \epsilon) = 0
% \vspace{-5pt}
\]
These conditions ensure that even if the global solution is not achievable, a local optimum $\delta^*$ that meets the KKT conditions can still be found, providing a practical solution.

% as \Name is not limited to these conditions

\begin{figure}[t]
    \centering
    \includegraphics[width=\linewidth]{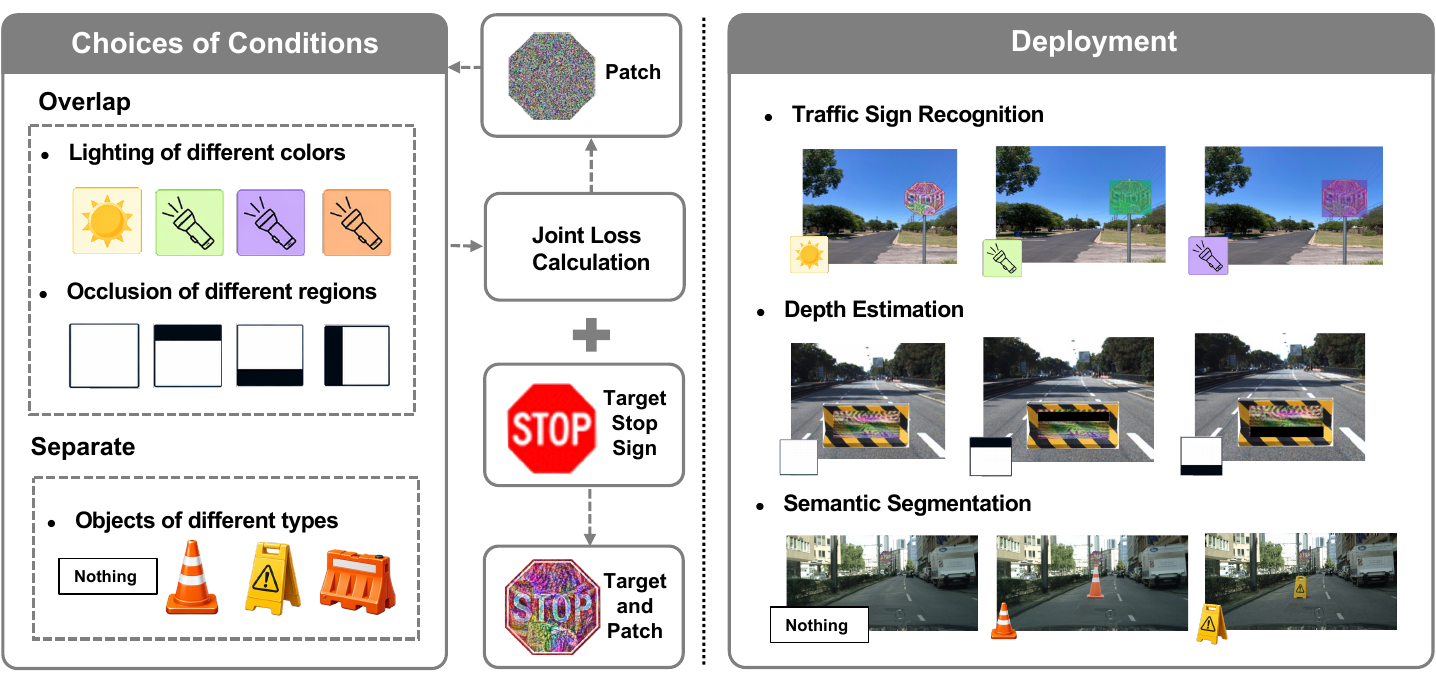}
    \caption{Overview of \Name. It presents a novel attack strategy as it can be flexibly extended to more pre-defined conditions and applied to more tasks. }
    \label{fig:overview}
    \vspace{-18pt}
\end{figure}

% \vspace{-10pt}
\subsection{Capacity of \Name}
%\noindent\textbf{How many attack goals can be achieved for \Name.} 
Intuitively, \Name is able to achieve an arbitrary number of attack goals under distinct pre-defined conditions, and its solution space is smaller than that of conventional PAPs for one fixed attack goal. We demonstrate that it is difficult for a stationary PAP to achieve unlimited attack goals: as the number of attack goals increases, the generation of \Name becomes increasingly difficult. The details are described in Appendix~\ref{app:cp}.
% ~\ref{app:cp}
%Then we conduct experiments to verify this argument.  

% \vspace{-5pt}
\section{Methodology}
We present the concrete methodology of generating and deploying \Name against different machine learning systems. 
Fig.~\ref{fig:overview} illustrates the overview of \Name. 

\subsection{General Formulation}
Across all tasks and trigger types, the generation of the adversarial perturbation $x_{adv}$ follows a similar optimization framework, as described in ~\Eref{eq:general_for}. Task-specific formulations are detailed in Appendix~\ref{app:fdt}.
% app:fdt

\subsection{Loss Design}
% \subsubsection{Loss Design}
% \vspace{-5pt}
\subsubsection{Separate Trigger }
\label{Sec:Meth}

Different tasks may require different formats of the adversarial 
loss function, as described below.

\begin{packeditemize}

\item \textbf{Semantic Segmentation.}
A semantic segmentation model performs per-pixel classification 
for an input image $x$, outputting a pixel-wise category map 
$Y$ of size $H \times W$. Therefore, for each pixel, we have:
% \vspace{-5pt}
\begin{equation}
\begin{array}{rr}
\mathcal{L}^{(i,j)} = \sum_{i,j}^{H,W}\big(\mathcal{L}_1^{(i,j)}(f(x+\delta), y_1) \\ + \lambda \cdot \mathcal{L}_2^{(i,j)}(f(x+\delta \oplus c), y_2)\big) \cdot \mathcal{W}^{(i,j)}.
\end{array}
\end{equation}
% \vspace{-5pt}

\noindent where $\mathcal{W}^{(i,j)}$ is the pixel-wise weight mask 
(0.9 for the target region, 0.2 for the rest).

\item \textbf{Multi-Label Learning.}
A multi-label learning model outputs a prediction vector 
$Y \in \{0,1\}^C$ for $C$ candidate categories, supporting 
AA, HA, and MA objectives. The loss is defined as:
% \vspace{-5pt}
\begin{equation}
\mathcal{L} = \sum_{i=1}^{C} w_{i}^{(k)} \cdot 
Y_{i}^{(k)} \log \left( \sigma \big( f(x' \oplus cl_k)_i 
\big) \right)
\label{eq:mll_loss}
\end{equation}
% \vspace{-5pt}

\noindent where $w_{i}^{(k)}$ is the class weight, $f(\cdot)_i$ is the 
logit of the $i$-th category, and $\sigma(\cdot)$ is the 
sigmoid activation. For HA, we set $Y_i = 0$ for the 
target category.

\item \textbf{Object Detection.}
An object detection model outputs $\hat{y} = \{y_{loc}, y_{size}, C\}$. The attacker can choose HA or MA, and the loss $\mathcal{L}$ is represented as:
% \vspace{-5pt}
\begin{equation}
\mathcal{L} =
\left\{
\begin{matrix}
-l_{obj}, & \text{if HA} \\[3pt]
l_{obj} + l_{cls} + l_{box}, & \text{otherwise}
\end{matrix}
\right.
\label{eq:det_loss}
% \vspace{-5pt}
\end{equation}

\noindent where $l_{obj}$, $l_{cls}$, $l_{box}$ are the objectness, 
classification, and bounding box regression losses, and 
$h$ denotes the index of the HA goal.

\item \textbf{Self-Balance Optimization.}
To address unbalanced optimization between dual attack 
objectives, we dynamically adjust the balance weight per 
iteration. The loss at iteration $t$ is:
% \vspace{-5pt}
\begin{equation}
\mathcal{L}^{t} = \mathcal{L}_{1}^{t}\big(f(x_{t}' \oplus c), y_{1}\big)
+ \left(\frac{\mathcal{L}_{2}^{t-1}}
{\mathcal{L}_{1}^{t-1}}\right)^{\!\gamma / \sqrt{t}} 
\cdot \mathcal{L}_{2}^{t}\big(f(x_{t}'), y_{2}\big)
\label{eq:self_balance}
% \vspace{-5pt}
\end{equation}

\noindent where $\mathcal{L}_{1}^{t}$, $\mathcal{L}_{2}^{t}$ are loss 
terms at iteration $t$, $\mathcal{L}_{1}^{t-1}$, 
$\mathcal{L}_{2}^{t-1}$ are values from the previous 
iteration, $\gamma$ is the scaling hyper-parameter, 
$x_{t}' = x + \delta_{t}$, and $\sqrt{t}$ is the decay 
term to prevent optimization vibration.

\end{packeditemize}

\subsubsection{Overlapping Trigger }

To enable attack goal switching in real time, we design a new loss function for optimizing  \Name, as shown below:
% \vspace{-5pt}
\begin{equation}
\begin{split}
    \underset{\Name}{\mathrm{argmin}}  \mathbb{E}_{x\sim X}  \mathcal{L}_{no} + \sum_{ k=1 }^{N} w_k * \mathcal{L}_{cl}^{k} + \mathcal{L}_{en}
\end{split} 
\label{eq:objective_design}
% \vspace{-5pt}
\end{equation}
where $\mathcal{L}_{no}$ is the normal loss that makes \Name achieve the benign effects without projections and $\mathcal{L}_{cl}^{k}$ is the $k_{th}$ adversarial loss in $N$ conditions. $\mathcal{L}_{en}$ is the enhancement loss for improving the stealthiness and robustness, which will be detailed in the following sections.
The detailed optimization process is described in Algorithm 3 in Appendix~\ref{algorithm23}. 
% ~\ref{alg:alg_1}.

Different tasks may require different formats of the adversarial loss function $\mathcal{L}_{cl}^{k}$ in Eq~\ref{eq:objective_design}, as described below. 

\begin{packeditemize}

\item\textbf{Object Classification.}
 The attacker only considers MA for the classification task. Hence, the goals are set as $G_s = {MA_1; ...; MA_N}$, where they all can be optimized with the cross-entropy loss:
 % \vspace{-5pt}
\begin{equation}
\mathcal{L}_{cl}^{k} = CELoss(f(x^{'} + cl_k), y_k)
\label{eq:classification_objective_design}
% \vspace{-5pt}
\end{equation}

\item\textbf{Object Detection.}
The attacker can choose HA or MA to target the object detection task. There could be two strategies to establish the goal set: (1) the attacker can adopt the same set as classification: $G_{s1} = \{MA_1; ...; MA_N\}$. (2) The attacker can combine HA and MA in the goal set: $G_{s2} = \{HA; MA_1; ...; MA_{N-1}\}$. Correspondingly, the loss term $\mathcal{L}_{cl}^{k}$ can be represented as:
\begin{equation}
\mathcal{L}_{cl}^{k}=\left\{\begin{matrix} \mathcal{L}_{HA} + CELoss(f(x^{'} + cl_k), y_k), & \text{if $G_{s2}$}\\
CELoss(f(x^{'} + cl_k), y_k), & \text{otherwise}
\end{matrix}\right.
\label{eq:detection_objective_design}
\end{equation}
where $\mathcal{L}_{HA}$ is the HA loss, used in prior works~\cite{song2018physical, zhu2023tpatch}.

\item\textbf{Depth Estimation.}
The attacker can choose two attack goals: $G_{s} = \{FA;NA\}$.
The loss term $\mathcal{L}_{cl}^{k}$ can be represented as:
% \vspace{-5pt}
\begin{equation}
\mathcal{L}_{cl}^{k} = CELoss(f(x^{'} + cl_k), D_k)
\label{eq:de_objective_design}
% \vspace{-5pt}
\end{equation}

\end{packeditemize}

Robustness is enhanced by applying Expectation over Transformation (EoT). Stealthiness is improved through PGD with an $L_{\infty}$ constraint and three losses (content, smoothness, and photorealism regularization). Physical-world adaptation is achieved via context augmentation, EoT on both the perturbation and background, and a total variation loss that smooths the patch. Further details are provided in Appendix~\ref{app:dresepwa}.
% \ref{app:dresepwa}.

Algorithm 1
% \ref{alg:unified_main} 
in Appendix~\ref{algorithm23}
% \ref{algorithm23} 
describes the overall process of generating \Name. The details of the algorithms of overlapping trigger architecture and separate trigger architecture are also in Appendix~\ref{algorithm23}.
% \ref{algorithm23}.

% \vspace{-10pt}
\section{Experimental Evaluation}

\subsection{Experimental Setup}

\noindent \textbf{Evaluation Metrics.} We consider the following metrics.
% for evaluation. 

(1) \textit{PIoU} (\textit{Partial Intersection over Union}). IoU (Intersection over Union) is widely used to evaluate object detection. However, IoU is not able to accurately reflect the attack effectiveness, as it focuses on the prediction accuracy of the entire
image. We therefore use PIoU, which computes IoU restricted to the attacker-specified region $S_a$: 
% \sout{(where labels are intended to change)}:
% \vspace{-5pt}
\begin{equation}
\mathrm{PIoU}=\frac{\left(S_t\cap S_p\right)\cap S_a}{\left(S_t\cup S_p\right)\cap S_a},
% \vspace{-5pt}
\end{equation}
where $S_t$ and $S_p$ denote the ground-truth and predicted object regions, respectively.

(2) \textit{Benign Accuracy (BA).} This represents the performance of \Name on the patched validation set under normal conditions. We use Accuracy and Mean Average Precision (mAP) for traffic sign classification and detection, respectively. For depth estimation, we define the correct prediction as within 5\% of the ground-truth distance of the target obstacle. 
% \sout{This value should be as high as possible.}

(3) \textit{Attack Success Rate (ASR).} We measure attack performance using ASR:
% \vspace{-5pt}
\begin{equation}
    \mathrm{ASR}=\frac{1}{N}\sum_{i=1}^{N}\mathbb{I}\!\left(f(x'_i\oplus c_i)=y_2 \ \wedge\  f(x'_i)=y_1\right),
\label{eq: all ASR}
% \vspace{-5pt}
\end{equation}
where $N$ is the number of tested samples and $\mathbb{I}(\cdot)$ is the indicator function. \Eref{eq: all ASR} counts an attack as successful if and only if \Name simultaneously (i) induces the target effect under the predefined triggering context and (ii) preserves benign behavior in the absence of the context. 
For semantic segmentation, we further require a PIoU change of at least $0.5$ in the target region, and the predicted labels under both settings (with and without the triggering context) must match the preset goals.
% }

(4) \textit{w. ASR and w/o. ASR.} This metric is used to analyze how \Name performs towards the particular attack goals. `w. ASR' and `w/o. ASR' stand for the attack success rate when the image is with and without the triggering context, respectively. Formally,
% \vspace{-5pt}
\begin{equation}
\left\{
\begin{array}{l}
    \mathrm{w.} \space \mathrm{ASR} = \frac{2\cdot\mathbb{I}(f(x'\oplus c)=y_2)}{N}, \\
    \mathrm{w/o.} \space \mathrm{ASR} = \frac{2\cdot\mathbb{I}(f(x')=y_1)}{N}.
\end{array}
\right.
% \vspace{-5pt}
\end{equation}

(5) \textit{Goal-i ASR ($G_i$-ASR).} This means that an attack goal can be achieved with its corresponding condition while maintaining benign results in the absence of conditions. Such a metric is helpful for us to analyze the performance.

\noindent\textbf{Thresholds.} For the detection model, the IoU threshold in mAP is set to 0.5. For depth estimation, the threshold of prediction error is set to 14\%, which is the default setting in~\cite{liubeware}. For a ground-truth depth of 10 meters, this threshold value means that a prediction error of at least 1.4 ($10\times14\%$) meters is considered a successful attack. We also test other threshold values for depth estimation in the ablation study.

\noindent\textbf{Models and Datasets.} We evaluated the following models: YOLOv3~\cite{redmon2018yolov3}, YOLOv5~\cite{yolov5}, EfficientDet~\cite{tan2020efficientdet}, DeepLabV3~\cite{chen2017rethinking}, SegFormer~\cite{xie2021segformer}, SETR~\cite{zheng2021rethinking}, GCN decoders, VOC2007~\cite{pascal-voc-2007}, Faster R-CNN, VGG-13/16, ResNet-50/101, MobileNetV2, Mono2~\cite{mono2}, ManyDepth~\cite{manydepth}, MiDaS~\cite{midas}, DepthAnything~\cite{depthanything}.
We employed the following datasets: MSCOCO17~\cite{lin2014microsoft}, BDD-100K~\cite{yu2020bdd100k}, Cityscapes~\cite{cordts2016cityscapes}, MS COCO~\cite{COCO}, GTSRB~\cite{stallkamp2012man}, KITTI~\cite{Geiger2012CVPR}.
Detailed experimental configurations are provided in Appendix~\ref{app:es}.
% ~\ref{app:es}.

% \vspace{-15pt}
\subsection{Dataset Evaluation Results}
\subsubsection{Effectiveness}\hfill

\textbf{Separate Trigger}

We consider three attack scenarios in separate trigger experiments. Scenario 1: the attack is dormant normally and activates only when a trigger appears.
Scenario 2: the attack is always active, with the trigger switching objectives.
Scenario 3: the attack is normally active but deactivated by the trigger.

% \vspace{-15pt}

\begin{figure*}[t]
    \centering
    \includegraphics[width=\linewidth]{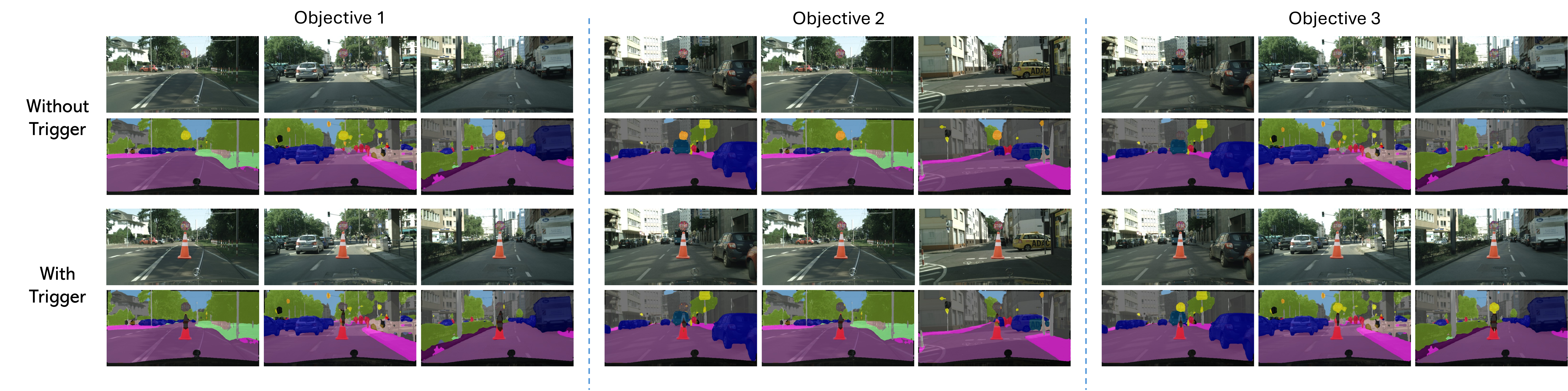}
    \vspace{-20pt}
    \caption{\textbf{Visualization of separate trigger on semantic segmentation.} The upper row shows the segmentation results without triggering context. The lower row stands for the opposite situation. The \textcolor{yellow}{yellow} zone in the segmentation output stands for the label  `\textcolor{yellow}{Traffic Sign}', which is the ground truth label of the victim object. The \textcolor{gray}{gray} zone represents `\textcolor{gray}{building}', while the \textcolor{orange2}{orange} region means `\textcolor{orange2}{traffic light}'. The labels shift perfectly in these non-cherry-picked examples.}
    \label{fig:seg_visual}
% \vspace{-10pt}
\end{figure*}

\textit{Semantic Segmentation.}
% \noindent \textbf{Effectiveness.} 
Semantic segmentation models are more vulnerable to separate triggers than other tasks. As shown in~\Tref{tab:T_SEG}, the ASRs in all three attack scenarios exceed 80\%, even in Scenario 2, where performance drops substantially in object detection and multi-label learning. This is because segmentation predicts the label of each pixel based on surrounding pixels, so the model tends to extract and exploit more contextual information. 
\Fref{fig:seg_visual} shows visualization results.
We also evaluate the effectiveness of the label weight mask $\mathcal{W}$. Results are in Appendix~\ref{app：rst}.
% ~\ref{app:separate_semantic_segmentation}.

\begin{table}[h]
% \vspace{-5pt}
    \centering
    \caption{Attack results on semantic segmentation.}
    \vspace{-5pt}
    \resizebox{\linewidth}{!}{\begin{tabular}{ccccccccc}
    \toprule
       \multirow{2}{*}{Model} & \multirow{2}{*}{Dataset} & \multirow{2}{*}{Condition} & \multicolumn{2}{c}{Scenario 1} & \multicolumn{2}{c}{Scenario 2} & \multicolumn{2}{c}{Scenario 3} \\ \cmidrule{4-9}
        & & & ASR & PIoU & ASR & PIoU & ASR & PIoU \\ \midrule
        \multirow{6}{*}{DeepLabV3} & \multirow{3}{*}{Cityscapes} & w. & 93.2 & 0.862 & 91.7 & 0.821 & 94.5 & 0.874 \\
         & & w/o. & 94.1 & 0.867 & 93.9 & 0.862 & 93.5 & 0.869 \\
         & & w. $\wedge$ w/o. & 92.1 & 0.847 & 90.7 & 0.817 & 93.0 & 0.856 \\ \cmidrule{2-9}
         & \multirow{3}{*}{BDD-100k} & w. & 97.9 & 0.871 & 95.2 & 0.862 & 98.1 & 0.877 \\
         & & w/o. & 98.7 & 0.884 & 96.1 & 0.872 & 96.5 & 0.863 \\
         & & w. $\wedge$ w/o. & 95.2 & 0.848 & 95.9 & 0.847 & 96.1 & 0.871 \\ \midrule
        \multirow{6}{*}{SegFormer} & \multirow{3}{*}{Cityscapes} & w. & 95.2 & 0.843 & 90.7 & 0.812 & 94.6 & 0.839 \\ 
         & & w/o. & 93.4 & 0.828 & 89.6 & 0.797 & 94.3 & 0.833 \\ 
         & & w. $\wedge$ w/o. & 90.2 & 0.791 & 88.7 & 0.750 & 92.8 & 0.801 \\ \cmidrule{2-9}
         & \multirow{3}{*}{BDD-100k} & w. & 92.2 & 0.833 & 93.7 & 0.874 & 95.9 & 0.839 \\
         & & w/o. & 94.8 & 0.878 & 95.5 & 0.841 & 97.2 & 0.874\\
         & & w. $\wedge$ w/o. & 92.0 & 0.806 & 92.4 & 0.809 & 95.1 & 0.844 \\ \midrule
        \multirow{6}{*}{SETR} & \multirow{3}{*}{Cityscapes} & w. & 92.0 & 0.803 & 88.4 & 0.773 & 90.6 & 0.791 \\ 
         & & w/o. & 91.7 & 0.798 & 89.7 & 0.788 & 94.9 & 0.827\\
         & & w. $\wedge$ w/o. & 91.0 & 0.791 & 89.4 & 0.781 & 92.7 & 0.811 \\ \cmidrule{2-9}
         & \multirow{3}{*}{BDD-100k} & w. & 94.3 & 0.821 & 90.5 & 0.791 & 92.3 & 0.804\\
         & & w/o. & 94.4 & 0.813 & 89.4 & 0.782 & 93.2 & 0.813 \\
          & & w. $\wedge$ w/o. & 92.7 & 0.818 & 88.2 & 0.741 & 94.4 & 0.820\\
         \bottomrule
    \end{tabular}}
    \label{tab:T_SEG}
% \vspace{-10pt}
\end{table}

% \vspace{-20pt}

\textit{Object Detection.}
% \noindent \textbf{Effectiveness}. 
The results are shown in~\Tref{tab:T_OD}. The ASRs of Scenario 1 and Scenario 3 are much higher, while in Scenario 2, it is harder to make the separate trigger have multiple attack goals in most cases. 
Another finding is that the attack performance of the separate trigger on BDD-100k is generally better than that on MSCOCO. There are two possible reasons: (1) BDD-100k is a less challenging dataset for object detection.
(2) The data distribution of BDD-100k is much simpler than MSCOCO.

\begin{table}[h]
% \vspace{-8pt}
    \centering
    \caption{Attack results on object detection.}
    \vspace{-5pt}
    \resizebox{\linewidth}{!}{\begin{tabular}{ccccccccc}
    \toprule
       \multirow{2}{*}{Model} & \multirow{2}{*}{Dataset} & \multirow{2}{*}{Condition} & \multicolumn{2}{c}{Scenario 1} & \multicolumn{2}{c}{Scenario 2} & \multicolumn{2}{c}{Scenario 3} \\ \cmidrule{4-9}
        & & & ASR & PIoU & ASR & PIoU & ASR & PIoU \\ \midrule
        \multirow{6}{*}{YOLOv3} & \multirow{3}{*}{MSCOCO} & w. & 87.8 & 0.142 & 43.6 & 0.374 & 94.5 & 0.911 \\ 
         & & w/o.  & 92.5 & 0.898 & 85.8 & 0.101 & 88.7 & 0.131 \\
         & & w. $\wedge$ w/o. & 85.4 & - & 41.6 & - & 84.7 & - \\ \cmidrule{2-9}
         & \multirow{3}{*}{BDD-100k} & w. & 91.9 & 0.127 & 53.4 & 0.452 & 94.3 & 0.877 \\
         & & w/o. & 92.1 & 0.794  & 91.5 & 0.073 &  87.8 & 0.152 \\
         & & w. $\wedge$ w/o. & 88.0 & - & 48.9 & - & 86.1 & - \\ \midrule
        \multirow{6}{*}{YOLOv5} & \multirow{3}{*}{MSCOCO} & w. & 90.5 & 0.073 & 56.4 & 0.451 & 93.2 & 0.898\\  
         & & w/o.& 91.3 & 0.875 & 91.4 & 0.097 & 93.8 & 0.103 \\ 
         & & w. $\wedge$ w/o. & 87.9 & - & 52.1 & - & 90.5 & -\\ \cmidrule{2-9} 
         & \multirow{3}{*}{BDD-100k} & w. & 89.1 & 0.105 & 53.1 & 0.471 & 91.2 & 0.828  \\
         & & w/o.  & 93.5 & 0.880 & 91.4 & 0.085 & 88.7 & 0.101 \\
         & & w. $\wedge$ w/o. & 87.3 & - & 47.7 & - & 87.2 & -\\ \midrule
        \multirow{6}{*}{EfficientDet} & \multirow{3}{*}{MSCOCO} & w. & 70.7 & 0.253 & 55.7 & 0.500 & 89.0 & 0.867 \\
         & & w/o. & 91.1 & 0.889 & 79.4 & 0.173
 & 73.3 & 0.225 \\ 
         & & w. $\wedge$ w/o. & 69.9 & - & 54.0 & - & 72.0 & - \\ \cmidrule{2-9}
         & \multirow{3}{*}{BDD-100k} & w. & 87.4 & 0.105 & 59.1 & 0.511  & 89.4 & 0.804 \\
         & & w/o.  & 94.4 & 0.914 & 92.7 & 0.103 & 91.3 & 0.087 \\ 
         & & w. $\wedge$ w/o. & 84.3 & - & 54.9 & - & 86.7 & -\\
         \midrule
    \end{tabular}}
    \label{tab:T_OD}
% \vspace{-10pt}
\end{table}

\textit{Multi-label Classification.}
\label{subsec: multilabel_perf}
% \noindent\textbf{Effectiveness.} 
\Tref{tab:T_MC} reports results on multi-label classification. In Scenario 1, ASRs in both conditions exceed 80\%, since switching from benign to a hiding attack is relatively easy. In Scenario 2, misclassification ASR is relatively low, likely because harder attacks require more budget.
% \sout{: misclassification must suppress logits of the correct class while increasing those of the target class, inducing a larger output change.}
Combining multiple attack goals in one perturbation further degrades performance. Notably, the separate trigger depends strongly on background/context.
% \sout{and is therefore hard to make universal across diverse backgrounds (see disscusion)}

\begin{table}[h]
    \centering
    \caption{Attack results on multi-label classification.}
    \vspace{-5pt}
    \resizebox{0.85\linewidth}{!}{
    \begin{tabular}{cccccc}
    \toprule
       Model & Dataset & Condition & Scenario 1 & Scenario 2 & Scenario 3 \\ \midrule
        \multirow{6}{*}{ML-GCN} & \multirow{3}{*}{MSCOCO} & w. & 92.5 & 56.4 & 93.1  \\
         & & w/o. & 94.4 & 89.9 & 94.1  \\
         & & w. $\wedge$ w/o. & 90.7 & 52.7 & 90.5 \\ \cmidrule{2-6}
         & \multirow{3}{*}{VOC-07} & w. & 92.1 & 50.5 & 92.3  \\
         & & w/o. & 90.2 & 81.1 &  93.4 \\
         & & w. $\wedge$ w/o. & 88.7 & 46.4 & 90.4 \\ \midrule
        \multirow{6}{*}{MLDecoder} & \multirow{3}{*}{MSCOCO} & w. & 93.7 & 48.5 & 95.1 \\ 
         & & w/o. & 93.9 & 86.0 & 94.1  \\ 
         & & w. $\wedge$ w/o. & 92.3 & 43.8 & 93.0 \\ \cmidrule{2-6}
         & \multirow{3}{*}{VOC-07} & w. & 94.2 & 32.2 &  97.6 \\
         & & w/o. & 93.8 & 80.1 & 87.1  \\ 
         & & w. $\wedge$ w/o. & 93.8 & 31.0 & 84.9 \\ \bottomrule
        \end{tabular}}
    \label{tab:T_MC}
% \vspace{-5pt}
\end{table}
% 1-3

% \vspace{15pt}

\textbf{Overlapping Trigger}\hfill

\textit{Object Detection and Classification.}
% (Traffic Sign Recognition)
\label{sec:case_study_1}
\Name is validated on object detectors and image classifiers in a dataset simulation, where \Name is directly attached to digital images.
Tables~\ref{tab:classification_dataset_results} and \ref{tab:detection_dataset_results} report the overall results for traffic sign classification and detection, respectively. All models achieve high ASRs, demonstrating the effectiveness of \Name. We also observe that the ASRs for all models are lower than $G_i$-ASR. Obviously, \Name needs to meet the attack targets under multiple lighting conditions. This demonstrates that attack performance may gradually decrease with more attack goals, which we have proved through theoretical analysis in Section~\ref{sec:methodology}.
We further conduct ablation studies on the impact of colored light intensity, patch region size, and increasing number of attack goals. Detailed results and analysis are provided in Appendix~\ref{app:sderot}.
% ~\ref{app:Traffic_Sign_Recognition}. 
The effects of color-goal combinations and attack-goal weights $W_k$ are discussed in Appendix ~\ref{sec:appendix}.
% ~\ref{sec:appendix}.

\begin{table}[h]
\caption{ASR(\%) of \Name on traffic sign classification.}
\vspace{-5pt}
\setlength{\tabcolsep}{2mm}{\resizebox{\linewidth}{!}{
\begin{tabular}{cccccc}
\hline
Classification & VGG-13 & VGG-16  & ResNet-50 & ResNet-101 & MobileNetV2 \\ \hline
BA &95.2   &100.0    &74.9  &70.1   & 78.2\\ \hline
$G_1$-ASR (MA1) &82.1    &96.7  &87.3   &80.0   & 81.8 \\ \hline
$G_2$-ASR (MA2)&80.9     &97.3  &99.7   &100.0  & 83.6\\ \hline
ASR &70.9   &95.9   &65.4   &62.3   & 64.1\\ \hline
\end{tabular}}}
\label{tab:classification_dataset_results}
% \vspace{-15pt}
\end{table}

\begin{table}[h]
% \vspace{-10pt}
\caption{ASR(\%) of \Name on traffic sign detection.}
\vspace{-5pt}
\setlength{\tabcolsep}{7mm}{\resizebox{\linewidth}{!}{
\begin{tabular}{cccc}
\hline
Detection & YOLOv3 & YOLOv5 & Faster R-CNN  \\ \hline
BA &100.0  &95.4  &100.0   \\ \hline
$G_1$-ASR (MA) &91.8  &75.4  &85.1  \\ \hline
$G_2$-ASR (HA) &95.6  &91.2  &90.6 \\ \hline
ASR &85.9  &71.9  &80.3  \\ \hline
\end{tabular}}}
\label{tab:detection_dataset_results}
% \vspace{-8pt}
\end{table}

\begin{table*}[h]
\centering
\caption{Transferability across depth estimation models in simulation.}
\vspace{-5pt}
% \vspace{-0.3cm}
\resizebox*{\linewidth}{!}{

\begin{tabular}{c|cccc|cccc|cccc|cccc}
\hline
\multirow{2}{*}{\diagbox[dir=NW]{Generation}{Test}} & \multicolumn{4}{c|}{Mono2} & \multicolumn{4}{c|}{Mande} & \multicolumn{4}{c|}{MiDaS} & \multicolumn{4}{c}{DeAny} \\ \cline{2-17} 
                                                                           & BA      & $G_1$-ASR      & $G_2$-ASR     & ASR& BA      & $G_1$-ASR     & $G_2$-ASR    & ASR  & BA      & $G_1$-ASR      & $G_2$-ASR   & ASR  & BA      & $G_1$-ASR     & $G_2$-ASR   & ASR  \\ \hline
Mono2                                                                         & 97.3   &  96.7 & 85.4 &84.8 & 100 &  43.5 & 18.7  & 18.5 & 81.2  & 0  & 25.3 & 0 & 100   & 12.5  & 6.2 & 0 \\
Mande                                                                         &100   &43.6    &11.36   &6.25  &100   &80.0   &93.7    &73.21 &100  &0   &12.05  &0  &87.5    &18.7  &10.9  &5.8 \\
MiDaS                                                                         & 99.1   & 11.7   & 7.53  & 6.0 & 100 & 12.1  &0 &0   & 85.7   & 82.6   & 59.9  & 56.2  & 86.4 & 25.6 & 5.2  & 5.12 \\
DeAny                                                                               &86.6   &2.5    &23.3  &0.1  &99.2   &8.3  &20.0    &1.5   &94.1    &5.6  &16.3  &2.4 &99.2  &72.0   &83.96 &65.3   \\ \hline 
\end{tabular}

}
\label{tab:de}
% \vspace{-8pt}
\end{table*}

\textit{Depth Estimation.}
We validate the effectiveness of \Name against CNN-based and transformer-based depth estimation models. Similar to traffic sign recognition (Section~\ref{sec:case_study_1}), we attach \Name directly to digital images for simulation.
Table~\ref{tab:de} reports the results for different depth estimation models. We make four observations. 
(1) All models retain high benign performance when \Name is attached without predefined conditions, indicating that \Name has little effect under normal operating conditions. 
(2) In the white-box setting, all models achieve high ASRs (56.29\%--84.88\%), as the attacks are generated with access to the target model. 
(3) Model vulnerability varies: Mono2 is highly susceptible to attacks generated for itself and Mande, while MiDaS is the most robust overall. 
(4) In the black-box setting, transferability is limited, particularly between CNN-based models (Mono2, Mande) and transformer-based models (MiDaS, DeAny). This result suggests that architectural differences produce distinct feature representations and decision boundaries, which reduce attack effectiveness. Fig.~\ref{fig:strong_weak_depth} in Appendix ~\ref{app:aekdm} further shows that, on KITTI with Mono2, \Name can successfully induce FA and NA under red and green light projections, respectively, while remaining benign without projection.
We further conduct ablation studies on the impact of color intensity. Results are provided in Appendix ~\ref{app:sderot}.
% ~\ref{app:Depth_Estimation}. 

% \subsection{Attack Transferability.} 

\begin{figure*}[t]
\centering

\includegraphics[width=0.32\linewidth]{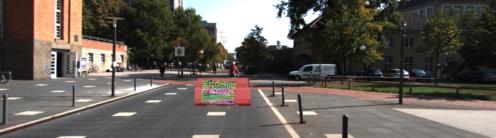}
\includegraphics[width=0.32\linewidth]{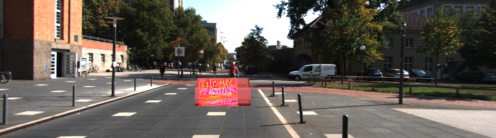}
\includegraphics[width=0.32\linewidth]{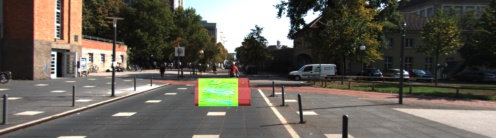}
\includegraphics[width=0.32\linewidth]{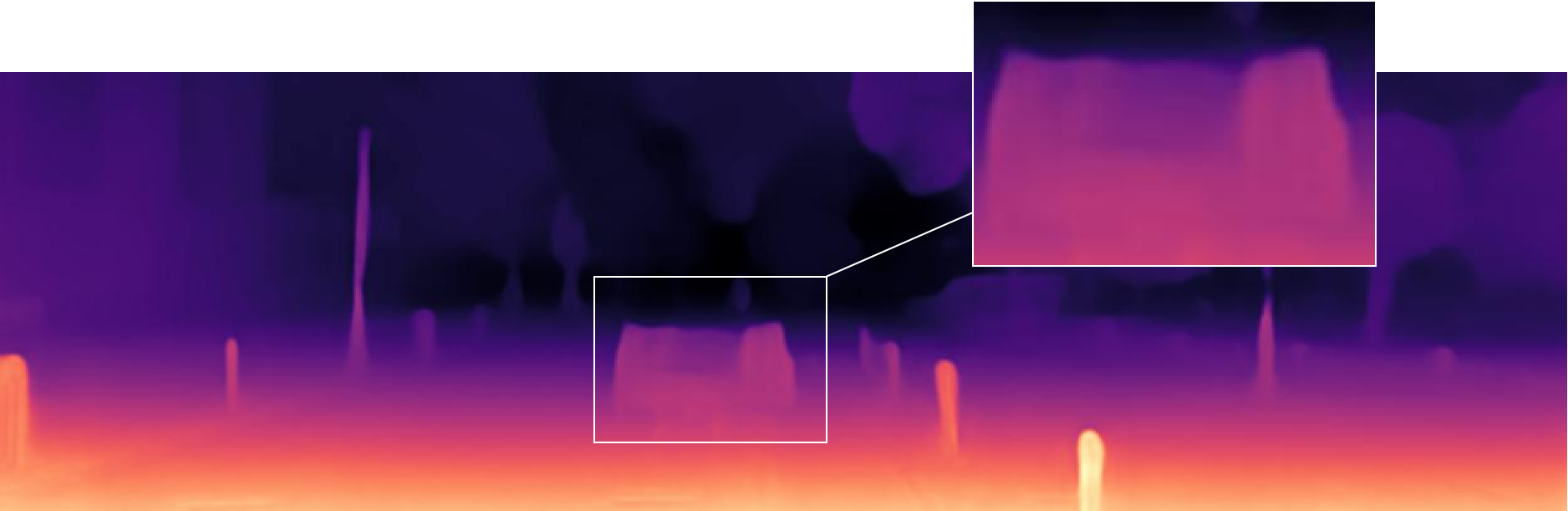}
\includegraphics[width=0.32\linewidth]{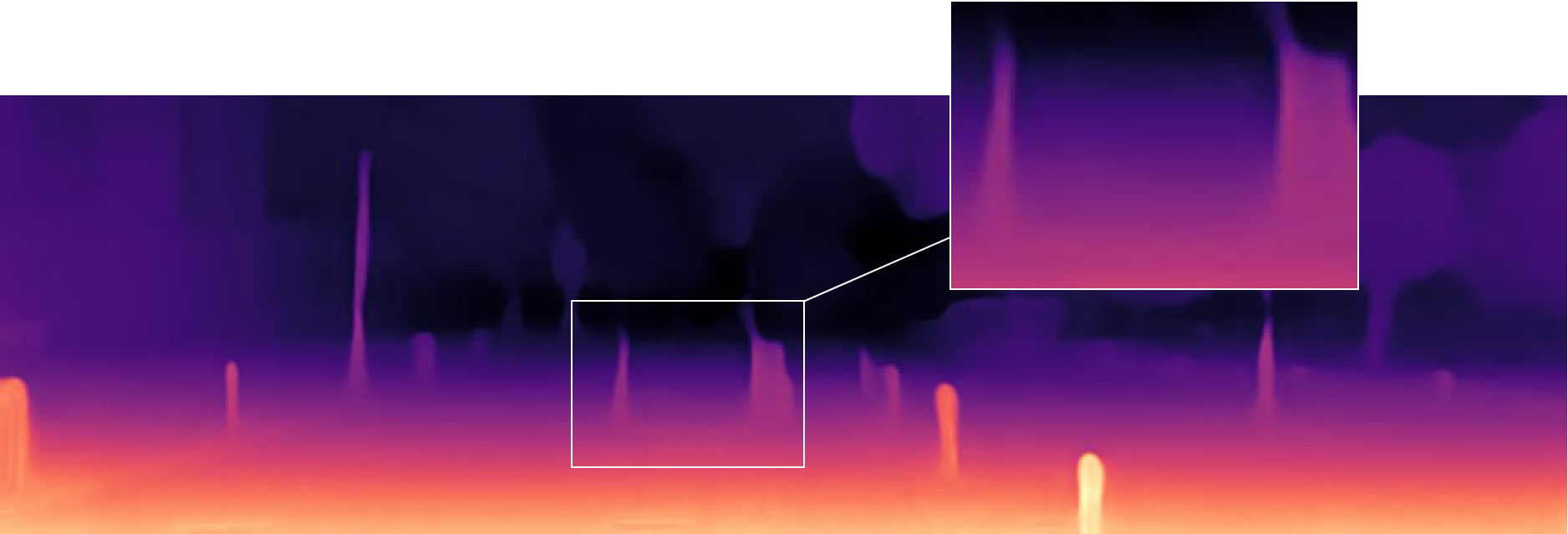}
\includegraphics[width=0.32\linewidth]{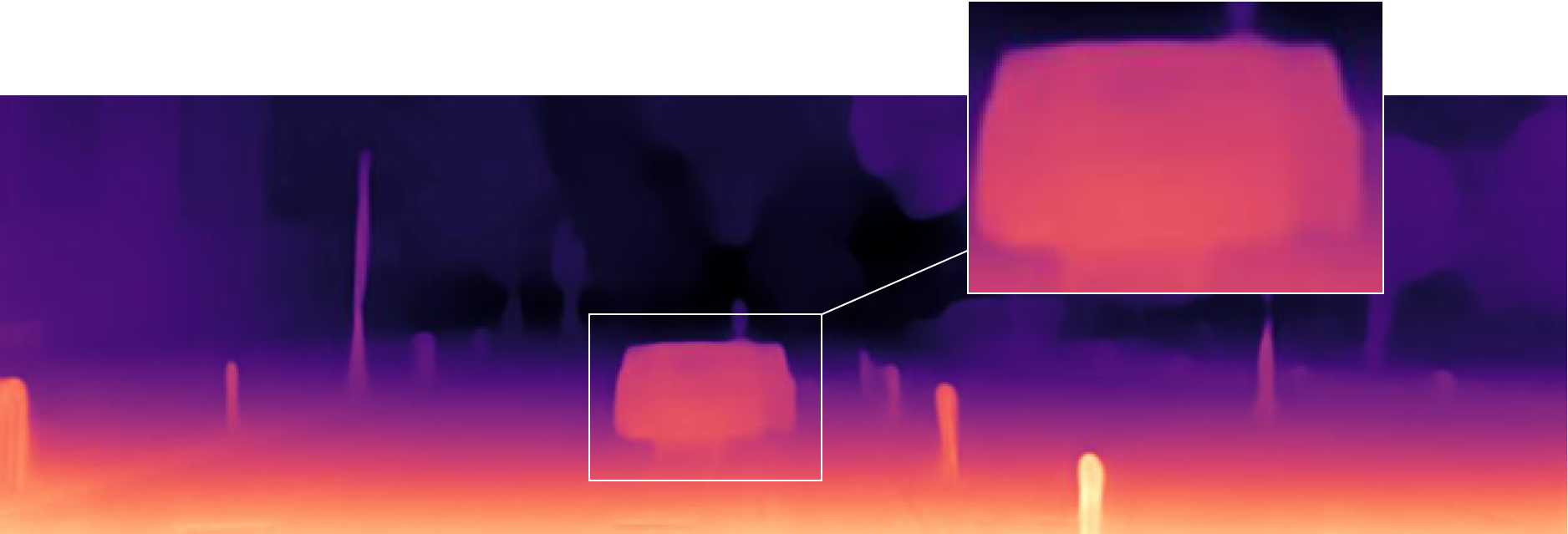}
\vspace{-5pt}
\caption{Visualizations on KITTI dataset. From left to right: \Name does not affect depth estimation without light projection; \Name is predicted up to more than 20\% farther away than its actual distance when red light is projected; \Name is predicted up to less than 20\% closer than its actual distance when green light is projected.}
\label{fig:strong_weak_depth}
% \vspace{-15pt}
\end{figure*}

\subsubsection{Transferability}
Details are provided in Appendix ~\ref{app:transferability}.
% ~\ref{app:transferability}.

% \vspace{-5pt}
\subsection{Physical World Evaluation}

% \vspace{-8pt}

\subsubsection{Experimental Setup}\hfill
 
\textbf{Separate Trigger.}

Without loss of generality, we choose the object detection task. 
Following previous experiments, we use a traffic cone as the trigger and a stop sign as the victim object. Photos taken at the attack location minimize background influence, with the camera 5 meters away, $\lambda$ = 0.5, and YOLOv3/YOLOv5 (MSCOCO2017) as victim models.
We focus on Scenario 1, where \Name is benign normally but launches a hiding attack when the trigger is present. We place and remove the traffic cone to observe label shifting, calculating ASR separately for frames with and without the trigger.

\textbf{Overlapping Trigger.}

\textit{Object Detection and Classification.}
% \noindent\textbf{Setup.} 
Experiments are conducted on a closed campus road using a UGV, our \Name, and a color flashlight (Fig. 7).
% ~\ref{fig:setup_physical}
(1) \textit{UGV and cameras.} The UGV has an Intel RealSense D435i front-facing camera. A DJI Action 3 and iPhone 11 Pro Max are also mounted at the same position for comparison. All cameras run at 30 fps at 1.5 meters height (resolutions in Table~\ref{tab:cameras} in Appendix~\ref{physical_overlapping}). 
(2) \textit{\Name} is color-printed at 50cm $\times$ 50cm and placed at 1.7 meters height. 
(3) \textit{Color flashlight.} The flashlight (3000LM, \$16.75 from Amazon) offers green and orange filters, three dimming levels and high-frequency flash mode. It is placed 2.3 meters in front of \Name, focused to cover it. The UGV starts 15 meters away, moving forward at 5m/s (18km/h). We repeat the experiment for 5 runs.
We conduct both dynamic evaluations (testing attack performance across varying distances) and static evaluations (testing multiple light conditions on the same frame). Detailed setups are provided in Appendix~\ref{physical_overlapping}. 

\textit{Monocular Depth Estimation.}
We conduct the experiments on a closed campus road. The target camera is an Intel RealSense D435i. We mount \Name on the rear of a BMW X1, which serves as the target object. The vehicle is 4.95 m long, 1.97 m wide, and 1.905 m high. We use Mono2 as the monocular depth estimation model. The victim vehicle drives toward the target vehicle from 10 m away while recording the adversarial scene.

\subsubsection{Evaluation Results}\hfill

\textbf{Separate Trigger.}

 \Tref{tab:tri_distance} and~\Tref{tab:distance_s} show how trigger and camera distance affect \Name's performance. From~\Tref{tab:tri_distance}, YOLOv5's w. ASR drops from 95.5\% to 51.5\% as the trigger moves farther from the stop sign. This is likely due to limited positional augmentation during optimization and CNN receptive field constraints. As shown in~\Fref{fig:OD_dist_cons}, the trigger's influence diminishes gradually and becomes negligible beyond 30 centimeters. From~\Tref{tab:distance_s}, both w. ASRs and w/o. ASRs decrease with greater camera distance, indicating that \Name's overall effectiveness degrades as the camera moves farther away. In Appendix~\ref{app:physical_separate}, we also provide interpretation of the attack.

\begin{table}[h]
% \vspace{-10pt}
\centering
\caption{ASR under different triggering distances.}
\vspace{-5pt}
\setlength{\tabcolsep}{1mm}{\resizebox{0.8\linewidth}{!}{
\begin{tabular}{c|cccc|cccc}
\toprule
\multicolumn{1}{c|}{Model} & \multicolumn{4}{c|}{YOLOv3} & \multicolumn{4}{c}{YOLOv5} \\ \midrule
\multicolumn{1}{c|}{Distance (cm)}&\multicolumn{1}{l|}{0-5} & \multicolumn{1}{l|}{5-10} & \multicolumn{1}{l|}{10-15} &  \multicolumn{1}{l|}{15-20} & \multicolumn{1}{l|}{0-5} & \multicolumn{1}{l|}{5-10} & \multicolumn{1}{l|}{10-15} &  \multicolumn{1}{l}{15-20} \\ \midrule
\multicolumn{1}{c|}{w. ASR(\%)} & 90.4 & 84.5 & 81.5 & 44.1 & 95.5 & 91.2 & 86.7 & 51.5 \\ \bottomrule
\end{tabular}}}
\label{tab:tri_distance}
\vspace{-15pt}
\end{table}

% \vspace{-10pt}

\begin{table}[h]
\centering
\caption{ASR under different camera distances.}
\vspace{-5pt}
\setlength{\tabcolsep}{1mm}{\resizebox{0.8\linewidth}{!}{
\begin{tabular}{c|cccc|cccc}
\toprule
\multicolumn{1}{c|}{Model} & \multicolumn{4}{c|}{YOLOv3} & \multicolumn{4}{c}{YOLOv5} \\ \midrule
\multicolumn{1}{c|}{Distance (m)}&\multicolumn{1}{c|}{2-4} & \multicolumn{1}{c|}{4-6} & \multicolumn{1}{c|}{6-8} & \multicolumn{1}{c|}{8-10} & \multicolumn{1}{c|}{2-4} & \multicolumn{1}{c|}{4-6} & \multicolumn{1}{c|}{6-8} & \multicolumn{1}{c}{8-10} \\ \midrule
\multicolumn{1}{c|}{w. ASR(\%)} & 90.4 & 85.5 & 71.1 & 54.3 & 95.1 & 93.2 & 76.6 & 50.5 \\
\multicolumn{1}{c|}{w/o. ASR(\%)} & 94.0 & 94.1 & 85.9 & 70.4 & 97.3 & 95.5 & 88.3 & 74.6 \\ \bottomrule
\end{tabular}}}
\label{tab:distance_s}
\vspace{-5pt}
\end{table}

\begin{figure*}[t]
    \centering
    \includegraphics[width=0.90\linewidth]{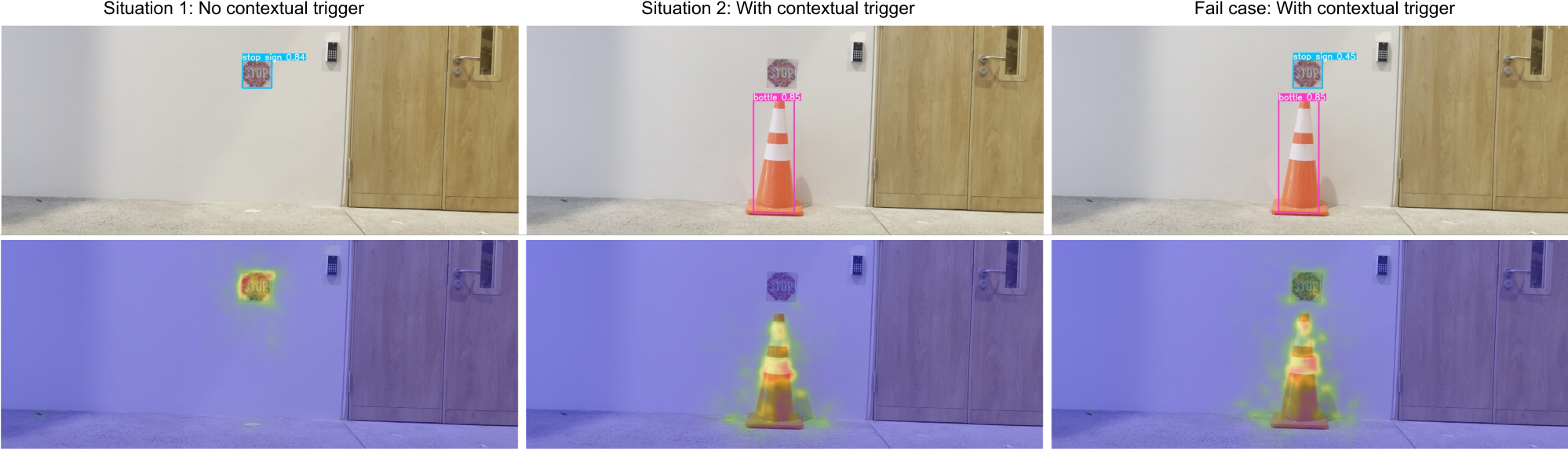}
     \caption{Attack interpretation with saliency map. First row: detection results. Second row: corresponding saliency map. First column: without triggering context. Second column: with triggering context. Third row, a failure case with triggering context.}
    \vspace{-10pt}
    \label{fig:Dist_OD}
\end{figure*}

\begin{figure}
    \centering
    \includegraphics[width=0.9\linewidth]{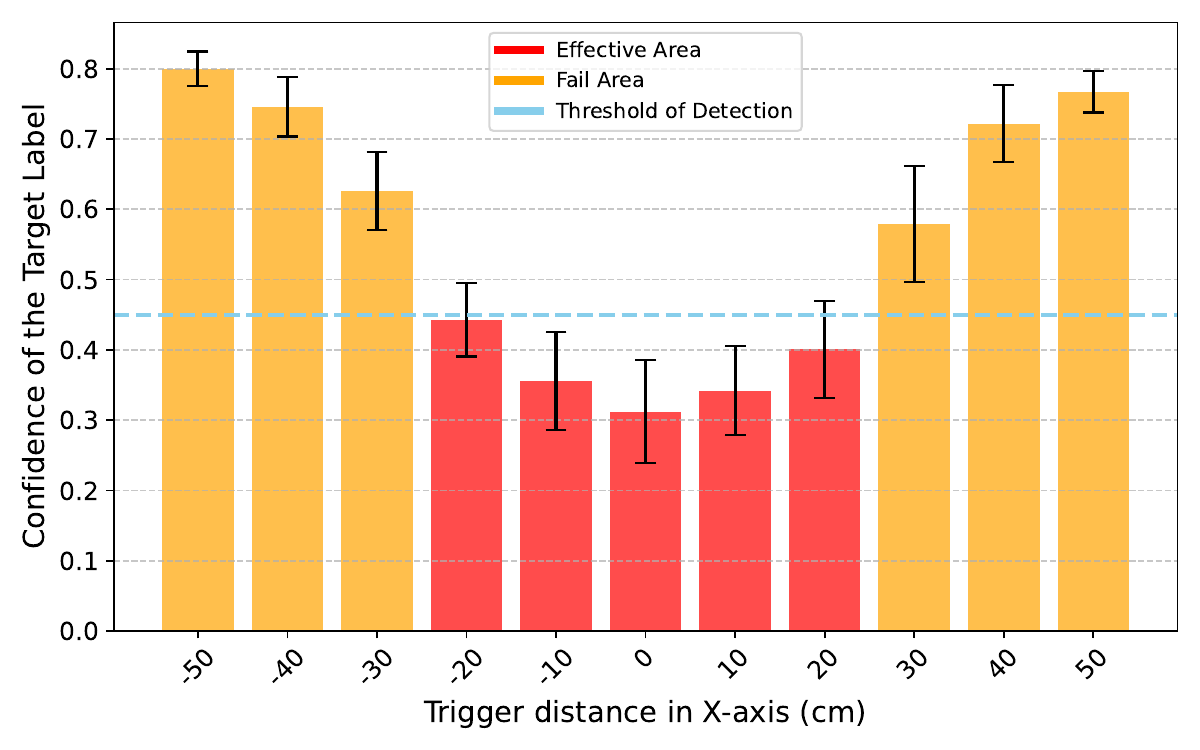}
    \vspace{-5pt}
    \caption{Prediction confidence of YOLOv5 with different distances between the triggering context and victim object.}
    \label{fig:OD_dist_cons}
\vspace{-20pt}
\end{figure}

\begin{figure*}[t]
% \vspace{-15pt}
\centering
\includegraphics[width=0.95\linewidth]{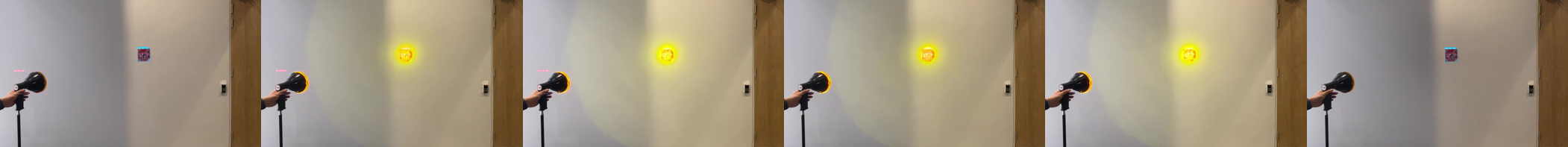}
\includegraphics[width=0.95\linewidth]{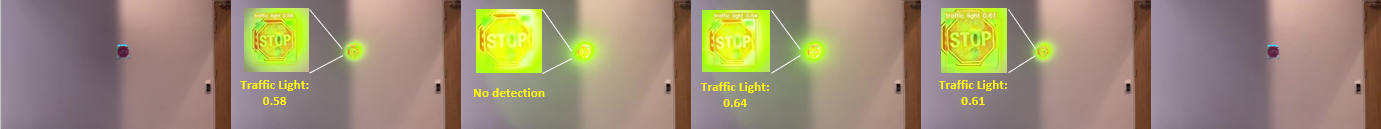}
\caption{Consecutive frames inference by YOLOv5 under high-frequency flashlight of \Name in the physical world. \textbf{Upper row}: the light color is orange. The intensity of colored lights from left to right: no projection, weak, strong, weak, strong, no projection. HA can be activated successfully under both weak and strong light intensities while remaining benign without light projections. \textbf{Lower row}: the light color is green. The intensity of colored lights from left to right: no projection, weak, strong, strong, weak, no projection. MA only fails when the light intensity is too strong.}
\label{fig:strong_weak}
\vspace{-10pt}
\end{figure*}

\textbf{Overlapping Trigger.}

\textit{Results for Object Detection and Classification.}
 
\paragraph{Static Evaluation Results}
\label{sec:physical_robustness}

We use three object detectors: YOLOv3, YOLOv5, and Faster R-CNN. For each model, we generate one \Name. We use the green color for attack goal\_1 (HA) and orange color for attack goal\_2 (MA, which recognize a stop sign as a ``Traffic Light'' sign).
Fig.~\ref{fig:strong_weak} and Table~\ref{tab:effectiveness_physical} show the results. \Name attacks all models effectively. However, the ASR is much lower than on the dataset in Table~\ref{tab:detection_dataset_results}. This may be due to the color gap between the light and the mask, as well as limited control of light intensity.

\begin{table}[h]
\vspace{-5pt}
\caption{Effectiveness of \Name for traffic sign recognition in the physical world.}
\vspace{-5pt}
\setlength{\tabcolsep}{6mm}{\resizebox{\linewidth}{!}{
\begin{tabular}{cccc}
\hline
&YOLOv3&YOLOv5& Faster R-CNN \\ \hline
BA & 84.4   & 83.8 & 93.7 \\
$G_1$-ASR (MA) & 60.2 & 58.8   & 69.2 \\
$G_2$-ASR (HA)& 98.9 & 96.7 & 96.1   \\ 
ASR & 45.5 & 48.7 & 36.7 \\ \hline
\end{tabular}}}
\label{tab:effectiveness_physical}
\vspace{-15pt}
\end{table}

\paragraph{Dynamic Evaluation Results}
% \noindent\textbf{Impact on distance.}

A vehicle is driving on a road equipped with high-resolution cameras and advanced image processing algorithms for object recognition. As it approaches a traffic sign, typically, the sign should appear larger in the vehicle's camera feed as the distance between the vehicle and the sign decreases.
As we stated in Appendix~\ref{app:dresepwa},
% ~\ref{sec:robustness_enhancement}
in the EoT setting, we do not use the traditional assumption that the distribution of pixels is uniform, but instead set a larger weight on a smaller pixel size. Table~\ref{tab:distance} shows the results. 
We observe that the results do not exhibit too many differences across these models.
However,  $G_1$-ASR (MA) increases slightly as the distance decreases.

\begin{table}[h]
\centering
\caption{ASR(\%) under different distances (m) for traffic sign recognition.}
\vspace{-5pt}
\setlength{\tabcolsep}{1.5mm}{\resizebox{\linewidth}{!}{
\begin{tabular}{cccccccccc}
\hline
\multicolumn{1}{l|}{Model} & \multicolumn{3}{c}{YOLOv3} & \multicolumn{3}{c}{YOLOv5} & \multicolumn{3}{c}{Faster R-CNN} \\ \hline
\multicolumn{1}{l|}{Distance}&\multicolumn{1}{l|}{3-6 } & \multicolumn{1}{l|}{6-9} & \multicolumn{1}{l|}{9-15} & \multicolumn{1}{l|}{3-6} & \multicolumn{1}{l|}{6-9} & \multicolumn{1}{l|}{9-15} & \multicolumn{1}{l|}{3-6} & \multicolumn{1}{l|}{6-9} & \multicolumn{1}{l}{9-15}  \\ \hline
\multicolumn{1}{l|}{BA} & 86.8 & 85.8 & 85.2 & 93.4 & 93.4 & 93.0 &91.5  &89.5 &   88.7 \\
\multicolumn{1}{l|}{$G_1$-ASR (MA)} &63.0   & 57.3 & 46.6 &  71.5& 65.0  &66.3 & 60.7 &68.4&  46.3  \\ 
\multicolumn{1}{l|}{$G_2$-ASR (HA)} & 82.8 & 75.4 & 78.1 & 79.1 & 88.0 &86.8 & 78.7 &  82.2 & 81.0 \\\hline
\end{tabular}}}
\label{tab:distance}
\vspace{-3pt}
\end{table}

We further evaluate the robustness of \Name under varying sunlight intensities, its transferability across detection models, and its consistency across different cameras. Detailed results are provided in Appendix~\ref{physical_overlapping}.
% ~\ref{physical_overlapping}.

\textit{Results for Monocular Depth Estimation.}
\label{sec:case_study_2}

This attack is generic and can be applied to different objects on public roads. We focus on a static vehicle, as shown in Fig.~\ref{fig:setup_physical}. We choose vehicles because they are common in driving scenes, widely used in prior work~\cite{liubeware, zheng2024pi,cheng2022physical}, and key targets of autonomous driving perception systems.

\begin{figure}[h]
\centering
\includegraphics[width=0.8\linewidth]{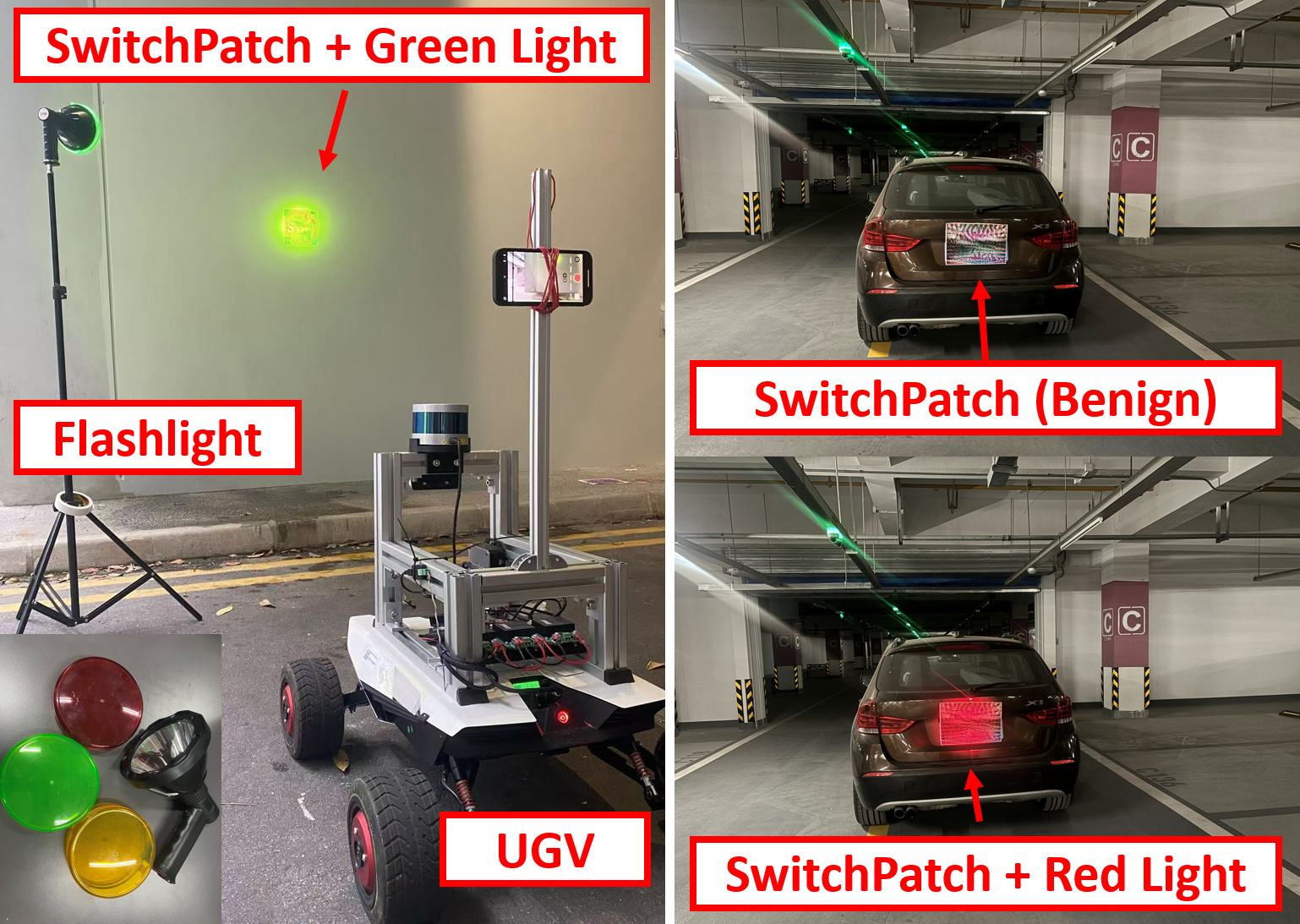}
\vspace{-5pt}
\caption{Experimental setup in the real-world. \textit{Left:} \Name is attached to a stop sign for traffic sign recognition; \textit{Right:} \Name is attached to the back of the vehicle for depth estimation.}
\label{fig:setup_physical}
\vspace{-5pt}
\end{figure}

\paragraph{Evaluation methodology}
We drive each route four times: benign, with \Name, with \Name under green light, and with \Name under blue light. We estimate ground-truth depth by $z = fH/s$, where $f$ is the focal length, $H$ is the real vehicle height, and $s$ is its height in the image, which is then used to compute $E_d$. We record frames at constant speed from 3 m to 15 m, collecting 200 frames total.

\paragraph{Impact of distance}
We first investigate how the distance can affect \Name for depth estimation.
For each distance interval (i.e., 3 m), we have 50 frames of images for evaluation.
Table~\ref{tab:depth_distance} provides the results, showing that \Name can achieve high ASRs at varying distances for all the models. Fig.~\ref{fig:strong_weak_depth} in Appendix~\ref{sec:appendix}
% ~\ref{sec:appendix} 
visualizes the \Name attack in the physical world.

\begin{table}[t]
\centering
\caption{ASR(\%) under different distances (m) for depth estimation.}
\vspace{-5pt}
\setlength{\tabcolsep}{1.5mm}{\resizebox{\linewidth}{!}{
\begin{tabular}{cccccccccc}
\hline
\multicolumn{1}{l|}{Model} & \multicolumn{3}{c}{Mono2} & \multicolumn{3}{c}{Mande} & \multicolumn{3}{c}{MiDaS} \\ \hline
\multicolumn{1}{l|}{Distance}&\multicolumn{1}{l|}{3-6 } & \multicolumn{1}{l|}{6-9} & \multicolumn{1}{l|}{9-15} & \multicolumn{1}{l|}{3-6} & \multicolumn{1}{l|}{6-9} & \multicolumn{1}{l|}{9-15} & \multicolumn{1}{l|}{3-6} & \multicolumn{1}{l|}{6-9} & \multicolumn{1}{l}{9-15}  \\ \hline
\multicolumn{1}{l|}{BA} & 87.5 &  85.4& 83.3 & 97.9 & 85.4 & 81.3 &91.6 &81.6 &  79.2  \\
\multicolumn{1}{l|}{$G_1$-ASR(NA)} & 66.7  & 58.3 & 52.5  &70.8& 64.5 & 60.4 & 69.3 & 52.1 & 31.3   \\ 
\multicolumn{1}{l|}{$G_2$-ASR(FA)} & 45.8 & 41.7 & 35.4 & 41.7 & 39.6 & 33.3 & 43.7 & 37.5 &27.1\\\hline
\end{tabular}}}
\label{tab:depth_distance}
\vspace{-15pt}
\end{table}

We also evaluate the impact of different depth thresholds on attack success rate. Details are provided in Appendix~\ref{physical_overlapping}.
% ~\ref{physical_overlapping}.

\section{Discussion}
\label{sec:occluding_method}

\noindent\textbf{More switching conditions.}
\Name is the first to demonstrate that an attacker can dynamically switch attack targets using a PAP in the physical world, making it a novel approach without an existing baseline. Beyond the triggers used in our experiments, attackers can also employ other techniques, such as occluding parts of the adversarial example to achieve switchable attack goals. Specifically, an attacker can optimize a global perturbation for traffic signs and tailor the optimization to different occlusion positions, each corresponding to a specific attack target. Fig. 13
% ~\ref{fig:occluding} 
in Appendix~\ref{sec:ATF}
% \ref{sec:ATF} 
illustrates demos using occlusion. The attacker can use a cube to occlude the upper, left, or other areas of the patch.

To demonstrate the feasibility of such attacks, we use a cube to occlude the left/up part as attack goal 1 and the right/down part as attack goal 2, respectively. Experiments are conducted using YOLOv3. The results are presented in Table XXVII
% ~\ref{tab:occluding} 
in Appendix~\ref{sec:ATF};
% \ref{sec:ATF}; 
the validation set includes 100 images that are randomly selected from the KITTI dataset.  
We observe that the other trigger techniques provide more consistent results than occlusion techniques, which although highly effective in the MA (especially for Left and Right) show significant drops in the HA scenario. We encourage researchers to explore further techniques to enhance the strength and stealthiness of \Name in the future.

\noindent\textbf{Countermeasures.} We have used different defense methods, e.g., input preprocessing, including image smoothing~\cite{cohen2019certified}, feature compression~\cite{jia2019comdefend} and input
randomization~\cite{xie2017mitigating}; defensive dropout~\cite{wang2018defensive} and adversarial training~\cite{madry2017towards}, to defend against \Name with simple experiments. These defense methods can slightly mitigate ASR in the range of 0 to 23\% on YOLOv5, which means these methods cannot fundamentally defend \Name. Since \Name is a general attack strategy, designing effective defense techniques that can be applied to various tasks will be our future work.

\noindent\textbf{Colored light projection improvement.}
In \Name, the adversary activates the attack subtly and briefly, leaving the vehicle little time to react; thus, light strength has limited impact on stealthiness. Compared with \cite{lovisotto2021slap,nassi2020phantom,zhu2023tpatch,chen2019shapeshifter}, \Name uniquely preserves the original traffic sign texture, and although it is less stealthy than RP2~\cite{eykholt2018robust}, it supports multiple attack goals. Its daytime performance is weaker because the colored light projections are less visible under strong ambient light. This limitation could be mitigated by using brighter colored lights or other triggers.

% \vspace{-10pt}
\section{Conclusion}
We introduce \Name, a novel and versatile PAP designed for dynamic and strategic manipulation in real-world environments. \Name is unique in its ability to leverage a diverse set of predefined physical conditions, allowing it to seamlessly adapt its attack objectives based on real-time situational awareness. We demonstrate the effectiveness, adaptability, and robustness of \Name through extensive evaluations across both simulation and real-world scenarios. Our results consistently highlight high attack success rates across a wide range of operational conditions.

% \section{Conclusion}
% The conclusion goes here.

\vspace{-5pt}
\section*{Acknowledgments}
ChatGPT (OpenAI) is used to assist in generating icons for figures and in polishing the language of the manuscript.

% \clearpage
% \balance
\bibliographystyle{IEEEtran}
\bibliography{references}

@inproceedings{cordts2016cityscapes,
  title={The cityscapes dataset for semantic urban scene understanding},
  author={Cordts, Marius and Omran, Mohamed and Ramos, Sebastian and Rehfeld, Timo and Enzweiler, Markus and Benenson, Rodrigo and Franke, Uwe and Roth, Stefan and Schiele, Bernt},
  booktitle={CVPR},
  pages={3213--3223},
  year={2016}
}

@misc{pascal-voc-2007,
	author = "Everingham, M. and Van~Gool, L. and Williams, C. K. I. and Winn, J. and Zisserman, A.",
	title = "The {PASCAL} {V}isual {O}bject {C}lasses {C}hallenge 2007 {(VOC2007)} {R}esults",
	howpublished = "http://www.pascal-network.org/challenges/VOC/voc2007/workshop/index.html"}

@inproceedings{zheng2021rethinking,
  title={Rethinking semantic segmentation from a sequence-to-sequence perspective with transformers},
  author={Zheng, Sixiao and Lu, Jiachen and Zhao, Hengshuang and Zhu, Xiatian and Luo, Zekun and Wang, Yabiao and Fu, Yanwei and Feng, Jianfeng and Xiang, Tao and Torr, Philip HS and others},
  booktitle={CVPR},
  pages={6881--6890},
  year={2021}
}

@article{xie2021segformer,
  title={SegFormer: Simple and efficient design for semantic segmentation with transformers},
  author={Xie, Enze and Wang, Wenhai and Yu, Zhiding and Anandkumar, Anima and Alvarez, Jose M and Luo, Ping},
  journal={NeurIPS},
  volume={34},
  pages={12077--12090},
  year={2021}
}

@article{chen2017rethinking,
  title={Rethinking atrous convolution for semantic image segmentation},
  author={Chen, Liang-Chieh and Papandreou, George and Schroff, Florian and Adam, Hartwig},
  journal={CoRR},
  year={2017}
}

@inproceedings{yu2020bdd100k,
  title={Bdd100k: A diverse driving dataset for heterogeneous multitask learning},
  author={Yu, Fisher and Chen, Haofeng and Wang, Xin and Xian, Wenqi and Chen, Yingying and Liu, Fangchen and Madhavan, Vashisht and Darrell, Trevor},
  booktitle={CVPR},
  pages={2636--2645},
  year={2020}
}

@inproceedings{tan2020efficientdet,
  title={Efficientdet: Scalable and efficient object detection},
  author={Tan, Mingxing and Pang, Ruoming and Le, Quoc V},
  booktitle={CVPR},
  pages={10781--10790},
  year={2020}
}

@misc{yolov5,
  author = {Glenn Jocher},
  title = {YOLOv5 by Ultralytics},
  year = {2020},
  publisher = {GitHub},
  journal = {GitHub repository},
  howpublished = {\url{https://github.com/ultralytics/yolov5}},
}

@inproceedings{eykholt2018robust,
  title={Robust physical-world attacks on deep learning visual classification},
  author={Eykholt, Kevin and Evtimov, Ivan and Fernandes, Earlence and Li, Bo and Rahmati, Amir and Xiao, Chaowei and Prakash, Atul and Kohno, Tadayoshi and Song, Dawn},
  booktitle={CVPR},
  year={2018}
}

@inproceedings{chen2019shapeshifter,
  title={Shapeshifter: Robust physical adversarial attack on faster r-cnn object detector},
  author={Chen, Shang-Tse and Cornelius, Cory and Martin, Jason and Chau, Duen Horng Polo},
  booktitle={ECML PKDD},
  year={2019},
}

@inproceedings{zhao2019seeing,
  title={Seeing isn't believing: Towards more robust adversarial attack against real world object detectors},
  author={Zhao, Yue and Zhu, Hong and Liang, Ruigang and Shen, Qintao and Zhang, Shengzhi and Chen, Kai},
  booktitle={ACM CCS},
  year={2019}
}

@inproceedings{lovisotto2021slap,
  title={$\{$SLAP$\}$: Improving Physical Adversarial Examples with $\{$Short-Lived$\}$ Adversarial Perturbations},
  author={Lovisotto, Giulio and Turner, Henry and Sluganovic, Ivo and Strohmeier, Martin and Martinovic, Ivan},
  booktitle={USENIX Security},
  year={2021}
}

@inproceedings{athalye2018synthesizing,
  title={Synthesizing robust adversarial examples},
  author={Athalye, Anish and Engstrom, Logan and Ilyas, Andrew and Kwok, Kevin},
  booktitle={ICML},
  pages={284--293},
  year={2018},
  organization={PMLR}
}

@inproceedings{ji2021poltergeist,
  title={Poltergeist: Acoustic adversarial machine learning against cameras and computer vision},
  author={Ji, Xiaoyu and Cheng, Yushi and Zhang, Yuepeng and Wang, Kai and Yan, Chen and Xu, Wenyuan and Fu, Kevin},
  booktitle={S\&P},
  year={2021}
}

@inproceedings{duan2021adversarial,
  title={Adversarial laser beam: Effective physical-world attack to dnns in a blink},
  author={Duan, Ranjie and Mao, Xiaofeng and Qin, A Kai and Chen, Yuefeng and Ye, Shaokai and He, Yuan and Yang, Yun},
  booktitle={CVPR},
  year={2021}
}

@article{madry2017towards,
  title={Towards deep learning models resistant to adversarial attacks},
  author={Madry, Aleksander and Makelov, Aleksandar and Schmidt, Ludwig and Tsipras, Dimitris and Vladu, Adrian},
  journal={CoRR},
  year={2017}
}

@article{xie2017mitigating,
  title={Mitigating adversarial effects through randomization},
  author={Xie, Cihang and Wang, Jianyu and Zhang, Zhishuai and Ren, Zhou and Yuille, Alan},
  journal={CoRR},
  year={2017}
}

@inproceedings{lin2014microsoft,
  title={Microsoft coco: Common objects in context},
  author={Lin, Tsung-Yi and Maire, Michael and Belongie, Serge and Hays, James and Perona, Pietro and Ramanan, Deva and Doll{\'a}r, Piotr and Zitnick, C Lawrence},
  booktitle={ECCV},
  pages={740--755},
  year={2014},
  organization={Springer}
}

@INPROCEEDINGS{Geiger2012CVPR,
  author = {Andreas Geiger and Philip Lenz and Raquel Urtasun},
  title = {Are we ready for Autonomous Driving? The KITTI Vision Benchmark Suite},
  booktitle = {CVPR},
  year = {2012}
}

@inproceedings{xiang2021patchguard,
  title={PatchGuard: A Provably Robust Defense against Adversarial Patches via Small Receptive Fields and Masking.},
  author={Xiang, Chong and Bhagoji, Arjun Nitin and Sehwag, Vikash and Mittal, Prateek},
  booktitle={USENIX Security},
  pages={2237--2254},
  year={2021}
}

@inproceedings{nassi2020phantom,
  title={Phantom of the adas: Securing advanced driver-assistance systems from split-second phantom attacks},
  author={Nassi, Ben and Mirsky, Yisroel and Nassi, Dudi and Ben-Netanel, Raz and Drokin, Oleg and Elovici, Yuval},
  booktitle={ACM CCS},
  year={2020}
}

@inproceedings{song2018physical,
  title={Physical adversarial examples for object detectors},
  author={Song, Dawn and Eykholt, Kevin and Evtimov, Ivan and Fernandes, Earlence and Li, Bo and Rahmati, Amir and Tramer, Florian and Prakash, Atul and Kohno, Tadayoshi},
  booktitle={USENIX workshop},
  year={2018}
}

@inproceedings{zhu2023tpatch,
  title={Tpatch: A triggered physical adversarial patch},
  author={Zhu, Wenjun and Ji, Xiaoyu and Cheng, Yushi and Zhang, Shibo and Xu, Wenyuan},
  booktitle={USENIX Security},
  year={2023}
}

@inproceedings{sato2024invisible,

  title={{Invisible Reflections: Leveraging Infrared Laser Reflections to Target Traffic Sign Perception}},

  author={Sato, Takami and Bhupathiraju, S Hrushikesh and Clifford, Michael and Sugawara, Takeshi and Chen, Qi Alfred and Rampazzi, Sara},

  booktitle={NDSS},

  year={2024}

}

@article{redmon2018yolov3,
  title={Yolov3: An incremental improvement},
  author={Redmon, Joseph and Farhadi, Ali},
  journal={CoRR},
  year={2018}
}

@inproceedings{gatys2016image,
  title={Image style transfer using convolutional neural networks},
  author={Gatys, Leon A and Ecker, Alexander S and Bethge, Matthias},
  booktitle={CVPR},
  year={2016}
}

@article{stallkamp2012man,
  title={Man vs. computer: Benchmarking machine learning algorithms for traffic sign recognition},
  author={Stallkamp, Johannes and Schlipsing, Marc and Salmen, Jan and Igel, Christian},
  journal={Neural networks},
  year={2012},
}

@Misc{COCO,
  author = {Microsoft},
  year = {2018}, 
  note = {\url{https://cocodataset.org/}},
  title = {Common Objects in Context (COCO) dataset}
}

@conference{
  kitti,
  author = "Andreas Geiger and Philip Lenz and Raquel Urtasun",
  title = "Are we ready for autonomous driving? The KITTI vision benchmark suite",
  booktitle = "CVPR",
  year = "2012"
}

@inproceedings{cohen2019certified,
  title={Certified adversarial robustness via randomized smoothing},
  author={Cohen, Jeremy and Rosenfeld, Elan and Kolter, Zico},
  booktitle={ICML},
  year={2019}
}

@inproceedings{jia2019comdefend,
  title={Comdefend: An efficient image compression model to defend adversarial examples},
  author={Jia, Xiaojun and Wei, Xingxing and Cao, Xiaochun and Foroosh, Hassan},
  booktitle={CVPR},
  year={2019}
}

@inproceedings{wang2018defensive,
  title={Defensive dropout for hardening deep neural networks under adversarial attacks},
  author={Wang, Siyue and Wang, Xiao and Zhao, Pu and Wen, Wujie and Kaeli, David and Chin, Peter and Lin, Xue},
  booktitle={ICCAD},
  year={2018},
}

@inproceedings{sharif2016accessorize,
  title={Accessorize to a crime: Real and stealthy attacks on state-of-the-art face recognition},
  author={Sharif, Mahmood and Bhagavatula, Sruti and Bauer, Lujo and Reiter, Michael K},
  booktitle={CCS},
  year={2016}
}

@inproceedings{luan2017deep,
  title={Deep photo style transfer},
  author={Luan, Fujun and Paris, Sylvain and Shechtman, Eli and Bala, Kavita},
  booktitle={CVPR},
  year={2017}
}

@article{levin2007closed,
  title={A closed-form solution to natural image matting},
  author={Levin, Anat and Lischinski, Dani and Weiss, Yair},
  journal={TPAMI},
  year={2007}
}

@article{tashiro2020diversity,
  title={Diversity can be transferred: Output diversification for white-and black-box attacks},
  author={Tashiro, Yusuke and Song, Yang and Ermon, Stefano},
  journal={NeurIPS},
  year={2020}
}

@inproceedings{cheng2022physical,
  title={Physical attack on monocular depth estimation with optimal adversarial patches},
  author={Cheng, Zhiyuan and Liang, James and Choi, Hongjun and Tao, Guanhong and Cao, Zhiwen and Liu, Dongfang and Zhang, Xiangyu},
  booktitle={ECCV},
  year={2022},
}

@article{shapley1953value,
  title={A value for n-person games},
  author={Shapley, Lloyd S},
  journal={Contribution to the Theory of Games},
  volume={2},
  year={1953}
}

@article{Index,
  title={An axiomatic approach to the concept of interaction among players in cooperative games},
  author={Grabisch, Michel and Roubens, Marc},
  journal={International Journal of game theory},
  volume={28},
  pages={547--565},
  year={1999},
  publisher={Springer}
}

@inproceedings{zheng2024pi,
  title={$\{$$\pi$-Jack$\}$:$\{$Physical-World$\}$ Adversarial Attack on Monocular Depth Estimation with Perspective Hijacking},
  author={Zheng, Tianyue and Hu, Jingzhi and Tan, Rui and Zhang, Yinqian and He, Ying and Luo, Jun},
  booktitle={USENIX Security},
  year={2024}
}

@inproceedings{liubeware,
  title={Beware of Road Markings: A New Adversarial Patch Attack to Monocular Depth Estimation},
  author={Liu, Hangcheng and Wu, Zhenhu and Wang, Hao and Han, Xingshuo and Guo, Shangwei and Xiang, Tao and Zhang, Tianwei},
  booktitle={NeurIPS},
  year={2024}
}

@inproceedings{mono2,
  title={Digging into self-supervised monocular depth estimation},
  author={Godard, Cl{\'e}ment and Mac Aodha, Oisin and Firman, Michael and Brostow, Gabriel J},
  booktitle={ICCV},
  year={2019}
}

@inproceedings{manydepth,
  title={The temporal opportunist: Self-supervised multi-frame monocular depth},
  author={Watson, Jamie and Mac Aodha, Oisin and Prisacariu, Victor and Brostow, Gabriel and Firman, Michael},
  booktitle={CVPR},
  year={2021}
}

@article{midas,
  title={Towards robust monocular depth estimation: Mixing datasets for zero-shot cross-dataset transfer},
  author={Ranftl, Ren{\'e} and Lasinger, Katrin and Hafner, David and Schindler, Konrad and Koltun, Vladlen},
  journal={TPAMI},
  year={2022},
  publisher={IEEE}
}

@article{depthanything,
  title={Depth anything: Unleashing the power of large-scale unlabeled data},
  author={Yang, Lihe and Kang, Bingyi and Huang, Zilong and Xu, Xiaogang and Feng, Jiashi and Zhao, Hengshuang},
  journal={CoRR},
  year={2024}
}

@article{martinez2014weierstrass,
  title={On Weierstrass extreme value theorem},
  author={Mart{\'\i}nez-Legaz, Juan Enrique},
  journal={Optimization letters},
  year={2014}
}

@inproceedings{chou2020sentinet,
  title={Sentinet: Detecting localized universal attacks against deep learning systems},
  author={Chou, Edward and Tramer, Florian and Pellegrino, Giancarlo},
  booktitle={SPW},
  pages={48--54},
  year={2020},
  organization={IEEE}
}

@article{brandenburger2007cooperative,
  title={Cooperative game theory: Characteristic functions, allocations, marginal contribution},
  author={Brandenburger, Adam},
  journal={Stern School of Business. New York University},
  volume={1},
  pages={1--6},
  year={2007}
}

@inproceedings{wangunified,
  title={A Unified Approach to Interpreting and Boosting Adversarial Transferability},
  author={Wang, Xin and Ren, Jie and Lin, Shuyun and Zhu, Xiangming and Wang, Yisen and Zhang, Quanshi},
  booktitle={ICLR},
  year={2021}
}

@inproceedings{glasmachers2017limits,
  title={Limits of end-to-end learning},
  author={Glasmachers, Tobias},
  booktitle={ACML},
  pages={17--32},
  year={2017},
  organization={PMLR}
}

@inproceedings{silver2017predictron,
  title={The predictron: End-to-end learning and planning},
  author={Silver, David and Hasselt, Hado and Hessel, Matteo and Schaul, Tom and Guez, Arthur and Harley, Tim and Dulac-Arnold, Gabriel and Reichert, David and Rabinowitz, Neil and Barreto, Andre and others},
  booktitle={ICML},
  pages={3191--3199},
  year={2017},
  organization={PMLR}
}

@article{wang2020unified,
  title={A unified approach to interpreting and boosting adversarial transferability},
  author={Wang, Xin and Ren, Jie and Lin, Shuyun and Zhu, Xiangming and Wang, Yisen and Zhang, Quanshi},
  journal={arXiv preprint arXiv:2010.04055},
  year={2020}
}

% \clearpage

\clearpage
\appendix
\setcounter{equation}{16}

\subsection{Supplementary to the Theoretical Analysis}
\label{app:sta}

\noindent\textbf{The Shapley Value}~\cite{shapley1953value}. Consider a cooperative game with a set of $N$ distinct players $\Omega=\{1, 2,...,N\}$ participating in it, wishing to win a high reward. The Shapley Value $\phi(i|\Omega)$ unbiasedly measures the contribution of the $i$-th player to the total reward gained by all players in $\Omega$, defined as follows:
\begin{equation}
 \phi(i|\Omega)=\sum_{S\subseteq\Omega\backslash\{i\}}{\frac{|S|!(n-|S|-1)!}{n!}(v(S\cup\{i\})-v(S))}, 
\end{equation}
where $v(\cdot)$ stands for the reward function, and $v(S)$ is the reward obtained by players from a subset $S\subseteq\Omega$.

\noindent\textbf{Shapley Interaction Index}~\cite{Index}. This index $I_{ij}$ between players $i,j$ is defined as follows:
\begin{equation}
I_{ij}=\phi(S_{ij}|\Omega)-[\phi(i|\Omega\backslash\{j\})+\phi(j|\Omega\backslash\{i\})],   
\label{eq: idx def}
\end{equation}
where $\phi(i|\Omega\backslash\{j\})$ and $\phi(j|\Omega\backslash\{i\})$ represent the individual contributions of units $i$ and $j$, respectively. This metric shows the gain in terms of the reward when both $i$ and $j$ are present in the game over the circumstances where they are not. In this case, Wang et al.~\cite{wang2020unified} showed an equivalent form of the index as follows:
\begin{equation}
I_{ij}'=\phi_{i,w\backslash j}-\phi_{i,w\backslash o j},
\label{eq: idx simp}
\end{equation}
where $\phi_{i,w\backslash j}$ is the gain of member $i$ when $j$ is always present, $\phi_{i,w\backslash o j}$ is that when $j$ is always absent. Specifically,
% \begin{equation}
% \left\{
% \begin{array}{l}
%     \phi_{i,w\backslash j}=\sum_{S\subseteq\Omega\backslash\{i,j\}}{\frac{|S|!(n-|S|-2)!}{(n-1)!}(v(S\cup\{i,j\})-v(S\cup\{j\}))} \\
%     \phi_{i,w\backslash o. j}=\sum_{S\subseteq\Omega\backslash\{i\}}{\frac{|S|!(n-|S|-2)!}{(n-1)!}(v(S\cup\{i\})-v(S))}
% \end{array}
% \right.
% \end{equation}
\begin{equation}
\left\{
\begin{array}{l}
    \phi_{i,w\backslash j}=\sum_{S\subseteq\Omega\backslash\{i,j\}}\frac{|S|!(n-|S|-2)!}{(n-1)!}(v(S\cup\{i,j\})-{}\\
    \qquad\qquad\qquad\qquad\qquad\qquad\qquad v(S\cup\{j\})) \\
    \phi_{i,w\backslash o. j}=\sum_{S\subseteq\Omega\backslash\{i\}}{\frac{|S|!(n-|S|-2)!}{(n-1)!}(v(S\cup\{i\})-v(S))}
\end{array}
\right.
\end{equation}

In comparison with directly calculating~\Eref{eq: idx def} which is NP hard, obtaining the index from~\Eref{eq: idx simp} is much more efficient. Moreover,~\Eref{eq: idx simp} is of a simpler form for further analysis.

\subsection{Proof of Proposition 2}
\label{app:prop2}
\begin{proof}

    We start by analyzing the conditions. 
    With $\|\nabla_{x}{\mathcal{L}(x+c)}\|^2_2 \ge \frac{-\log{\beta}}{\alpha}$ and $\nabla^{T}_{x}{\mathcal{L}(x+c)}\nabla_{x}{\mathcal{L}(x)}\le0$, it gives that:
    % \begin{align*}
    % & \nabla^T{\mathcal{L}(x+c)}(\nabla{\mathcal{L}(x+c)} - \nabla{\mathcal{L}(x)}) \\
    % = \quad & \nabla^{T}{\mathcal{L}(x+c)}\nabla{\mathcal{L}(x+c)} - \nabla^{T}{\mathcal{L}(x+c)}\nabla{\mathcal{L}(x)}  \\
    % \geq \quad & \|\nabla{\mathcal{L}(x+c)}\|^2_2\\
    % \geq \quad &  \frac{-\log \beta}{\alpha}.
    % \end{align*}
    \begin{equation}
    \begin{aligned}
        \|\nabla_{x}{\mathcal{L}(x+c)}\|^2_2 &= \nabla^T_{x}{\mathcal{L}(x+c)}\nabla_{x}{\mathcal{L}(x+c)} \\
        & \le \nabla^T_{x}{\mathcal{L}(x+c)}\nabla_{x}{\mathcal{L}(x+c)} \\     &\quad -\nabla_{x}^T{\mathcal{L}(x+c)}\nabla_{x}{\mathcal{L}(x)} \\
        & = \nabla^T_{x}{\mathcal{L}(x+c)}[\nabla_{x}{\mathcal{L}(x+c)}-\nabla_{x}{\mathcal{L}(x)}]
    \end{aligned}
    \label{eq:9}
  \end{equation}
    On the other hand, notice that $\|\nabla{\mathcal{L}(x)}\|^2_2 \geq 0$, we have:
    \begin{equation}
    \begin{aligned}
        & \nabla^T_{x}{\mathcal{L}(x)}(\nabla_{x}{\mathcal{L}(x+c)} - \nabla_{x}{\mathcal{L}(x)}) \\
        = \quad & \nabla^T_{x}{\mathcal{L}(x)}\nabla_{x}{\mathcal{L}(x+c)} - \nabla^T_{x}{\mathcal{L}(x)}\nabla_{x}{\mathcal{L}(x)} \\
        = \quad & \nabla^T_{x}{\mathcal{L}(x)}\nabla_{x}{\mathcal{L}(x+c)} - \|\nabla_{x}{\mathcal{L}(x)}\|^2_2 \\
        \leq \quad & 0. 
    \end{aligned}
    \label{eq:10}
    \end{equation}

    According to the Taylor's expansion, consider the loss function:
    \begin{align*}
        \mathcal{L}(x + \delta_c) = \mathcal{L}(x) + \delta_{c}^T\nabla_{x}\mathcal{L}(x) + R_1(\delta)
    \end{align*}
    The remainder term is sufficiently small $R_1(\delta_c) \rightarrow 0$.
    Insert the \Name adversarial perturbation $\delta_c = \alpha \nabla_{x}(\mathcal{L}(x+c) - \mathcal{L}(x))$, with Eq-\eqref{eq:9},
    \begin{align*}
        \mathcal{L}(x + \delta_c) &= \mathcal{L}(x) + \alpha \nabla_{x}^T(\mathcal{L}(x+c) - \mathcal{L}(x))\nabla_{x}\mathcal{L}(x) \\ &\quad + R_1(\delta_c) \\
        &= \mathcal{L}(x) + \alpha \nabla^T_{x}\mathcal{L}(x)(\nabla_{x}\mathcal{L}(x+c) - \nabla_{x}\mathcal{L}(x)) \\
        &\quad + R_1(\delta_c) \\
        &\leq \mathcal{L}(x),
    \end{align*}
    and similarly for $\mathcal{L}(x + \delta_c)$ with Eq-\eqref{eq:10}, we have,
    \begin{align*}
        \mathcal{L}(x +c+ \delta_c) & = \mathcal{L}(x+c) \\
        &\quad + \alpha \nabla_{x}^T\mathcal{L}(x+c)(\nabla_{x}\mathcal{L}(x+c) - \nabla_{x}\mathcal{L}(x)) \\
                                    & \quad + R_1(\delta_c) \\
                                    & \geq \mathcal{L}(x+c) -\log \beta.
    \end{align*}

    Notice that the loss function is the negative logarithm of the prediction:
    $
    \mathcal{L}(x) = - \log f(x).
    $
    
    Therefore, given condition on $f(x)$, we take the logarithm of both sides of the inequality, yielding: 
    \begin{align*}
        \mathcal{L}(x + \delta_c) &\leq \mathcal{L}(x)  \\
        % -\mathcal{L}(x + \delta_c) &\geq  -\log \mathcal{L}(x)\\
        \exp\left(-\mathcal{L}(x + \delta_c) \right)&\geq  \exp\left(-\log \mathcal{L}(x)\right)\\
        f(x + \delta_c) &\ge \beta .
    \end{align*}
    Symmetrically, we have:
        \begin{align*}
        \mathcal{L}(x + c + \delta_c) &\geq \mathcal{L}(x+c)  \\
        \exp\left(-\mathcal{L}(x + c+ \delta_c) \right)&\geq  \exp\left(-\log \mathcal{L}(x + c)\right)\\
        f(x + c + \delta_c) &\le \beta .
    \end{align*}
    Thus, the Eq. 2 holds.
\end{proof}
% -\eqref{eq:ctae_formulation}

\subsection{Capacity of PAP}
\label{app:cp}

\begin{theorem}\label{tm:Theorem1}
    The solution space of \Name decreases as the number $N$ of attack goals $y_{k}$ increases. 
\end{theorem}

\begin{proof}
Based on the above definitions,
%We denote the solution space of \Name for the attack goal $\mathcal{L}_{cl}^k$ as $\mathcal{S}^k$. Then, given the perturbation constraint of $\Name\in X_{\epsilon}$, 
the \Name solution needs to satisfy $\delta \in \mathcal{S}= X_{\epsilon} \cap \mathcal{S}_1 \cap\ldots \cap \mathcal{S}_N$. Since the solution space $\mathcal{S}_k$ is fixed when the types of conditions $cl_k$ are pre-defined, the size of solution space $\mathcal{S}$ decreases with the increase of $N$.
\end{proof}

Theorem~\ref{tm:Theorem1} indicates that optimizing $\delta$ becomes more difficult as the number of attack goals $\mathcal{L}_{cl}$ increases, due to the reduced solution space. This is also supported by experiments in Fig.~\ref{fig:increasing_goals_iter}. 

%Besides, we further show that the attack success rate reduces with the increase of attack goals.

\begin{theorem}\label{tm:Theorem2}
  The rate of successful and simultaneous attacks on all objectives decreases as the number $N$ of attack goals $y_{k}$ increases. 
\end{theorem}
\begin{proof}
    By denoting $\tilde{\mathcal{L}}^{N}$ as the optima of simultaneously attacking all $N$ objectives and $\tilde{x}_{adv}^{N_s}$ as the optima of simultaneously attacking $N_s$ objectives where $N_s \cap [N]$, we have
    \begin{align}
        \tilde{\mathcal{L}}^N &= 
        \max\{\mathcal{L}_{cl}^{k} |k=1,\ldots,N\} \\
        &\geq \max\{\mathcal{L}_{cl}^{k} |k=1,\ldots,N, \text{and}\ k\neq l_1,\ldots,l_{N_s}\} 
        = \tilde{\mathcal{L}}^{N_s}, \nonumber
    \end{align}
    which completes the proof.
\end{proof}

Theorem~\ref{tm:Theorem2} indicates that attack success rate monotonically decreases as the number of attack targets increases. Empirical evidence supporting this theorem can be found in Fig.~\ref{fig:increasing_goals} across different models like YOLOv3 and Faster R-CNN.

\begin{figure}[h]
\centering
\includegraphics[width=0.7\linewidth]{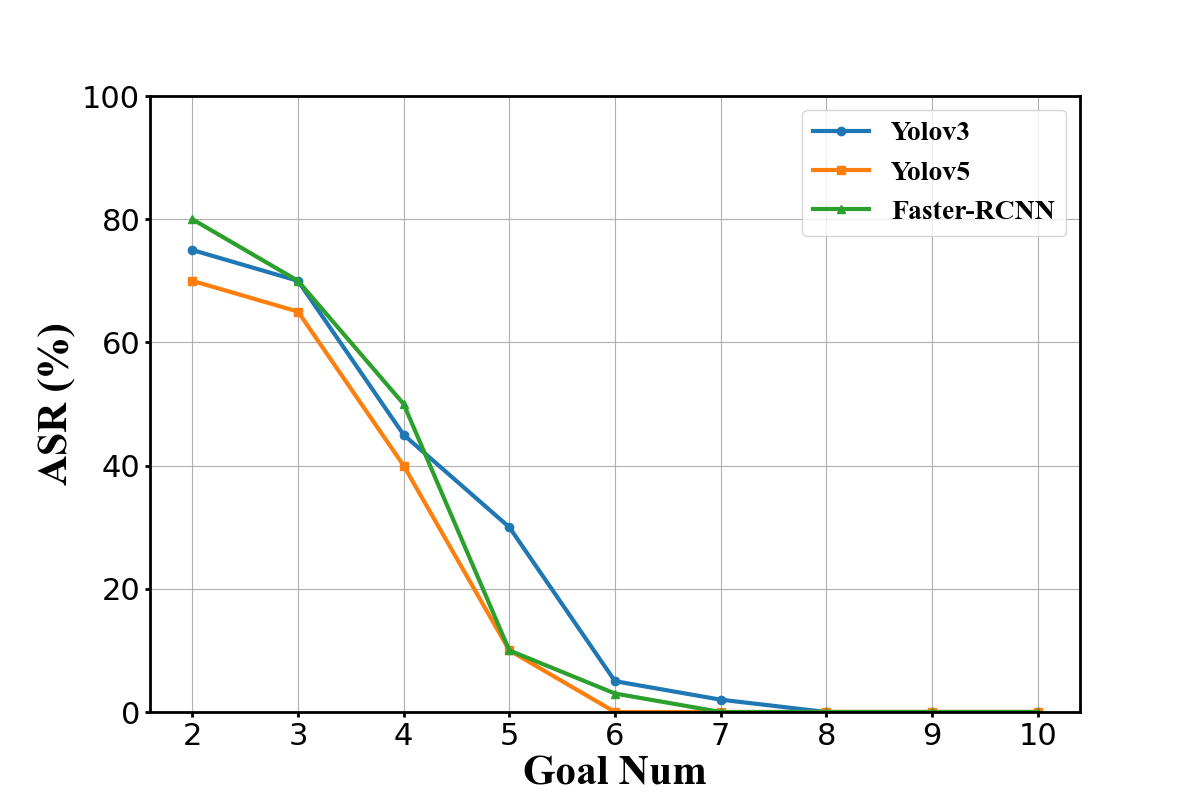}
\vspace{-5pt}
\caption{ASR with increasing attack goals on YOLOv3, YOLOv5, Faster R-CNN.}
\vspace{-15pt}
\label{fig:increasing_goals}
\end{figure}

\subsection{Formulations for Different Tasks}
\label{app:fdt}
\subsubsection{Separate Trigger}\hfill

Below we provide the specific formulations for different tasks.

\begin{itemize}[leftmargin=*,nosep,topsep=0pt,itemsep=2pt]

\item \textbf{Object Detection.} 
An object detection model outputs bounding boxes 
$\hat{y} = \{y_{loc}, y_{size}, C\}$, where $y_{loc}$ is 
the localization, $y_{size}$ is the box size, and $C$ is 
the category confidence. The adversarial perturbation 
$x_{adv}$ is generated by:
{\setlength{\abovedisplayskip}{3pt}
\setlength{\belowdisplayskip}{3pt}
\begin{equation}
\begin{split}
x_{adv} = \underset{x' \in \mathcal{X}_{\epsilon}}{\mathrm{argmin}} 
\Bigg( 
&\sum_{\substack{k=1 \\ k \neq h}}^{N} 
\big( l_{obj} + l_{cls} + l_{bbox} \big) \\
& \quad \quad \big( f(x' \oplus cl_k), \{y_{loc}, y_{size}, C_k\} \big) \\
&+ \mathcal{L}\big(f(x'), \{y_{loc}, y_{size}, C\}\big) \\
&- l_{obj}\big( f(x' \oplus cl_h), 
\{y_{loc}, y_{size}, C\} \big) 
\Bigg)
\end{split}
\label{eq:object_det}
\end{equation}}
\vspace{-5pt}

\noindent subject to:
{\setlength{\abovedisplayskip}{3pt}
\setlength{\belowdisplayskip}{3pt}
\begin{equation}
\left\{
\begin{aligned}
f(x) &= f(x_{adv}) = \{y_{loc}, y_{size}, C\}, \\
f(x_{adv} \oplus cl_h) &= \varnothing, \\
&\;\vdots \\
f(x_{adv} \oplus cl_N) &= \{y_{loc}, y_{size}, C_k\}, \\
k &\neq h,
\end{aligned}
\right.
\label{eq:det_constraint}
\end{equation}}
\vspace{-3pt}

\noindent where $h \in \{1,\ldots,N\}$ denotes the index of the HA 
goal, and $l_{obj}$, $l_{cls}$, $l_{bbox}$ are the 
objectness, classification, and bounding box regression 
losses, respectively.

\item \textbf{Semantic Segmentation.} 
A semantic segmentation model performs per-pixel 
classification for an input image $x$, and outputs 
a pixel-wise category map $Y$ with spatial size $H \times W$. 
The adversarial perturbation $x_{adv}$ is generated by:
{\setlength{\abovedisplayskip}{3pt}
\setlength{\belowdisplayskip}{3pt}
\begin{equation}
\begin{split}
x_{adv} = \underset{x' \in \mathcal{X}_{\epsilon}}{\mathrm{argmin}} 
\Bigg( 
&\sum_{k=1}^{N} \sum_{i,j}^{H,W} \mathcal{W}^{(i,j)} \cdot 
\Big( \mathcal{L}_1^{(i,j)}\big(f(x' \oplus cl_k), y_{1}^{(i,j)}\big) \\
&+ \lambda \cdot 
\mathcal{L}_2^{(i,j)}\big(f(x'), y_{2}^{(i,j)}\big) \Big) \\
&+ \mathcal{L}\big(f(x'), Y\big) + I_{\text{cross}}' 
\Bigg)
\end{split}
\label{eq:semantic_seg}
\end{equation}}
\vspace{-5pt}

\noindent subject to:
{\setlength{\abovedisplayskip}{3pt}
\setlength{\belowdisplayskip}{3pt}
\begin{equation}
\left\{
\begin{aligned}
f(x) &= f(x_{adv}) = Y, \\
f(x_{adv} \oplus cl_1) &= Y_1, \\
&\;\vdots \\
f(x_{adv} \oplus cl_N) &= Y_N,
\end{aligned}
\right.
\label{eq:seg_constraint}
\end{equation}}
\vspace{-3pt}

\noindent where $x' \in \mathcal{X}_{\epsilon}$ satisfies the $\ell_p$ 
norm constraint $\|x' - x\|_{p} \leq \epsilon$, $N$ is the 
number of attack goals, $\mathcal{W}^{(i,j)}$ is the 
pixel-wise weight mask (0.9 for the target region, 0.2 for 
the rest), $\lambda$ is the balance parameter, 
$I_{\text{cross}}'$ is the Shapley Interactive Index, 
and $Y_k$ is the malicious pixel-wise prediction map 
under the $k$-th trigger.

\item \textbf{Multi-Label Learning.} 
A multi-label learning model outputs a prediction vector 
$Y \in \{0,1\}^C$ for $C$ candidate categories, with three 
attack goals: Appearing Attack (AA), Hiding Attack (HA), 
and Misclassification Attack (MA). 
The adversarial perturbation $x_{adv}$ is generated by:
{\setlength{\abovedisplayskip}{3pt}
\setlength{\belowdisplayskip}{3pt}
\begin{equation}
\begin{split}
x_{adv} = \underset{x' \in \mathcal{X}_{\epsilon}}{\mathrm{argmin}} 
\Bigg( 
&\sum_{k=1}^{N} \sum_{i=1}^{C} w_{i}^{(k)} \cdot 
Y_{i}^{(k)} \log \left( \sigma \big( f(x' \oplus cl_k)_i 
\big) \right) \\
&+ \mathcal{L}\big(f(x'), Y\big) + I_{\text{cross}}' 
\Bigg)
\end{split}
\label{eq:multi_label}
\end{equation}}
\vspace{-3pt}

\noindent subject to the same constraint form as 
Eq.\,(\ref{eq:seg_constraint}), where $w_{i}^{(k)}$ is 
the weight for the $i$-th category under the $k$-th attack 
goal, $f(\cdot)_i$ is the logit of the $i$-th category, 
and $\sigma(\cdot)$ is the sigmoid activation. 
For HA, we simply set $Y_i = 0$ for the target category.

\end{itemize}

\subsubsection{Overlapping Trigger}\hfill

Below we provide the specific formulations for different tasks. 

\begin{packeditemize}

\item\textbf{Object Classification.} An object classification model predicts the categories of a given image $x$ as $y$. The procedure of generating the PAP $x_{adv}$ (\Name) can be formulated as:
\begin{equation}
\begin{split}
x_{adv} = \underset{x^{'} \in X_{\epsilon}}{\mathrm{argmin}}(\sum_{ k=1 }^{N} \mathcal{L}(f(x^{'} + cl_k), y_k) + \mathcal{L}(f(x^{'}) , y))\\
     \mathrm{s.t.}\quad 
     \left\{ 
     \begin{array}{lr}
     f(x) = f(x_{adv})= y \\
     f(x_{adv} + cl_1)= y_1 \\
     \dots \\
     f(x_{adv} + cl_N)= y_N
     \end{array}
     \right. 
    \label{eq:overall_1}
\end{split}
\end{equation}
where $x^{'}$ belongs to a set of images $X_{\epsilon}$ that satisfy an $L_p$ norm perturbation constraint (i.e, $|| x^{'} - x ||_{p} \le \epsilon$). $N$ is the number of goals that an attacker can achieve. 

\item\textbf{Object Detection.} An object detection model $f$ extracts features from an image $x$ and outputs its bounding box $y = \{y_{loc},y_{size}, C\}$ with its localization $y_{loc}$, size $y_{size}$ and the confidence score $C$ of the categories. Attacking this model can be formulated as:
\begin{equation}
\begin{split}
x_{adv} = \underset{x^{'} \in X_{\epsilon}}{\mathrm{argmin}}( \sum_{ k=1 }^{N} \mathcal{L}(f(x^{'} + cl_{k}), \{y_{\text{loc}},y_{\text{size}}, C_{k} \}) \\ + \mathcal{L}(f(x^{'}), \{y_{\text{loc}},y_{\text{size}}, C\}) \\ - \mathcal{L}(f(x^{'} + cl_{h}), \{y^{*}_{\text{loc}},y_{\text{size}}, C\})  \\
     \mathrm{s.t.}\quad 
     \left\{ 
     \begin{array}{lr}
     f(x) = f(x_{adv})= \{y_{\text{loc}},y_{\text{size}}, C \} \\
     f(x_{adv} + cl_h)= \varnothing \\
     \dots \\
     f(x_{adv} + cl_N)= \{y_{\text{loc}},y_{\text{size}}, C_{k} \} \\
     k \neq h\\
     % k \neq j
     \end{array}
     \right.
\end{split}
\label{eq:eq2}
\end{equation}

where $h \in \{1,..., N\}$ denotes the $h_{th}$ attack goal as HA. Eq~\ref{eq:eq2} can be segmented into three parts:  the first term is designed to achieve $k$ attack goals for misclassifications; the second term is associated with achieving the benign effect in the absence of attacker-controlled conditions; the third term induces the hiding attack goal when the $h_{th}$ condition is applied.

\item\textbf{Depth Estimation.} A depth estimation model $f$ predicts a depth map $D$ for the given RGB image $x$, where the depth map represents the depth information of each pixel in the input. Generating $x_{adv}$ for depth estimation can be formally expressed as: 
\begin{equation}
\begin{split}
x_{adv} = \underset{x' \in X_{\epsilon}}{\mathrm{argmin}} \left( \sum_{k=1}^{N} \mathcal{L}\big(f(x' + cl_k), D_k\big) + \mathcal{L}\big(f(x'), D\big) \right)\\
     \mathrm{s.t.}\quad 
     \left\{ 
     \begin{array}{lr}
f(x) &= f(x_{adv}) = D, \\
f(x_{adv} + cl_1) &= D_1, \\
&\dots \\
f(x_{adv} + cl_N) &= D_N.\\
     \end{array}
     \right.
\label{eq:eq3}
\end{split}
\end{equation}
where $D_k$ represents the maliciously perturbed depth map under the $k_{th}$ condition.

\end{packeditemize}

\subsection{Details of Robustness Enhancement, Stealthiness Enhancement, and Physical-World Adaptation}
\label{app:dresepwa}

\subsubsection{Robustness Enhancement}\hfill

\label{sec:robustness_enhancement}
To preserve the high attack effectiveness in the physical world, it is ideal that \Name can continuously realize all the target goals or stay benign under different environmental conditions. However, it is challenging to directly apply \Name generated in the digital domain to the physical world, due to the influence of unpredictable environmental conditions. 

To address this challenge, we adopt the Expectation over Transformation (EoT) technique, which augments the optimization of \Name with random transformations to overcome the environmental factors in the real world. Specifically, we augment \Name in the following dimensions: translation, rotation, resizing, color shifting, and variations in colored light intensity. Translation, rotation, resizing and color shifting are strategies utilized in previous works~\cite{zhu2023tpatch, chen2019shapeshifter, lovisotto2021slap} to enhance the patch robustness against distance effects and variations in environmental lighting. Additionally, we introduce variations in colored light intensity to complement the patch's use of colored light projections.  Below we give some details of adopting these transformations. 

% \noindent\textbf{Color shifting and Resizing.} 

% \noindent\textbf{Resizing.} This technique is commonly used in previous studies, where transformations are uniformly sampled within a specified range~\cite{zhu2023tpatch,chen2019shapeshifter, lovisotto2021slap}. However, recent findings by Wang et al.~\cite{wang2023does} challenge this uniform assumption for resizing operations in physical models, indicating that the majority of pixel sizes typically fall within a smaller range, suggesting that the vehicle is often farther from the object.  Consequently, we follow such a distribution in EoT and adapt to this reality. For instance, in our traffic sign detection experiments, with a traffic sign visible from a maximum distance of 20 meters and an input image size of 416x416, we set the selection likelihood as 90\% for a range of 0 to 200 pixels, and 10\% for the other pixel sizes. 

\noindent\textbf{Colored light intensity.} We apply translation, rotation, resizing, and color shifting with a uniform distribution to ensure the degree of their randomness. For colored light intensity, the effectiveness of colored light from a flashlight can vary significantly with ambient light conditions. For instance, colored light appears more visible at night, while less visible during the day under strong sunlight. To accurately simulate these varying conditions in EoT, we integrate colored light intensity variations that reflect different levels of ambient brightness.

We simulate $cl_k$ using a $k_{th}$ mask of different colors. For example, for blue light, the default setting is [0, 0, 255], representing the brightest state under low ambient light conditions. To simulate situations with higher ambient light and lower colored light intensity, we reduce this value to [0, 0, 127], effectively halving the perceived brightness. In our EoT process, we apply a uniform distribution to these RGB values to randomize the degree of light intensity transformation. This ensures that \Name can adapt to a broad range of real-world lighting conditions, thereby enhancing its practical effectiveness and robustness in physical environments.

\subsubsection{Stealthiness Enhancement}\hfill
Our optimization process has two designs for improving the stealthiness of \Name. First, we solve Eq. 10
% ~\ref{eq:objective_design} 
by using Project Gradient Descent (PGD) with the $L_{\infty}$ distance constraint during the gradient update step, which ensures that the per-dimension moving distance for each pixel in $x$ is smaller than $\epsilon$. We can use $\epsilon$ to control how similar \Name looks compared to the benign $x$: a smaller $\epsilon$ indicates stealthier \Name. We set $\epsilon$ as 0.4 by default. 

We also introduce the following three loss items for stealthiness enhancement: content loss, smoothness loss, and photorealism regularization loss. 
Formally, let $H$ denote a pre-trained CNN model used for feature extraction, $I_s$ and $I_d$ represent a source image and a designated style image, respectively. In this paper, $I_s$ denotes \Name which is initialized with $I_d$. 

\noindent\textbf{Content loss $\mathcal{L}_c$.} Proposed by style transfer works~\cite{gatys2016image}, this term can regularize \Name, encouraging the patch to learn the content and spatial structure of the target image rather than the details. 

The content loss is also defined based on the extracted features by $H$:
\begin{equation}
    \mathcal{L}_c = \sum_{l\in K} \|H_l(I_s)-H_l(I_d)\|^2_2.
\end{equation}
Different from $L_s$, $L_c$ is calculated based on the Euclidean distance between the feature maps of $I_s$ and $I_d$.

\noindent\textbf{Smoothness loss $\mathcal{L}_{sm}$.} This loss item encourages a locally smooth image, which can improve the stealthiness while also increasing the patch robustness~\cite{sharif2016accessorize}. 
It is defined as:
\begin{equation}
\begin{split}
    \mathcal{L}_{sm} = \sum_{i,j} ((I_s[i,j+1]-I_s[i,j])^2 \\ +(I_s[i+1,j]-I_s[i,j])^2 )^{\frac{1}{2}},
\end{split}
\end{equation}
where $I_s[i,j]$ denotes a pixel corresponding to the coordinate $(i,j)$.

\noindent\textbf{Photorealism regularization loss $\mathcal{L}_r$.} This loss is proposed in~\cite{luan2017deep} for imposing certain constraints on color transfer, thereby preventing color distortions. It is defined as follows:
\begin{equation}
    \mathcal{L}_r = \sum_{c\in\{R,G,B\}}V_c(I_s)^{\top}\mathcal{M}(I_s)V_c(I_s),
\end{equation}
where $c$ denotes one channel of RGB, $V_c$ reshapes its input into a shape of $N\times 1$ ($N$ represents the number of pixels in $I_s$), $\mathcal{M}(I_s)\in \mathbb{R}^{N\times N}$ represents a standard linear system that
can minimize a least-square penalty function described in~\cite{levin2007closed}.

Finally, the loss function of generating \Name is:
\begin{equation}
\begin{split}
    \underset{\Name}{\mathrm{argmin}}  \mathbb{E}_{x\sim X, t\sim       
    T} (\mathcal{L}_{no} + \sum_{ k=1 }^{N} w_k * \mathcal{L}_{cl}^{k}) + \\ \alpha * \mathcal{L}_c + \beta * \mathcal{L}_{sm} + \gamma * \mathcal{L}_r
\end{split} 
\label{eq:objective_design_overall}
\end{equation}
where $w_k$ is the weight for adjusting the loss of different attack goals. In subsequent experiments, unless otherwise stated, different goals share the same $w_k$ value. $\alpha$, $\beta$ and $\gamma$ are the weights to balance different loss components, which are set as 1e-2, 3e-6 and 3e-6, respectively. In addition, $t$ is the random transformation with its corresponding distribution $T$, which is designed to improve the robustness of \Name in the physical world.

\subsubsection{Adaption to the Physical World}\hfill

\label{subsec:Enh}

Besides the digital domain, our \Name can function in the physical world. We adopt the following techniques to bridge the semantic gap between two worlds, and enhance the robustness of \Name. 

\noindent \textbf{Context Augmentation}. When crafting \Name, the attacker expects the perturbation not to overfit certain pixel combinations of the triggering context, while \Name can be easily triggered. To this end, we propose to augment not only the perturbation itself but also the triggering context during the optimization. Specifically, the augmentation includes randomly moving and rotating the context object, zooming in and out, and adjustments of its brightness and contrast. This can enhance the sensitivity of \Name and ensure the triggering context can change the attack goals as expected. 

\noindent \textbf{Expectation over Transformation (EoT)}~\cite{athalye2018synthesizing}. Following many previous works~\cite{zhu2023tpatch, chen2019shapeshifter, lovisotto2021slap}, we also apply the EoT technique to enhance the perturbation. We also apply EoT to the background. This is because \Name is very sensitive to the background compared to conventional AEs, due to its nature of co-working with some contextual information. Therefore, we exploit EoT to simulate the potential shifting in the background.

\noindent \textbf{Total Variation Loss}. We add the total variation loss~\cite{sharif2016accessorize} into the final loss to smooth \Name and avoid overfitting to the digital field. This is also a common technique to generate physical patches. By minimizing the total variation loss, the color changes between the adjacent pixels are reduced to improve the image quality. This loss term is formulated as below: 
\begin{equation}
    \mathcal{L}_{TV} = \sum_{i,j}\sqrt{(x_{i+1,j}-x_{i,j})^2+(x_{i,j+1}-x_{i,j})^2}.
    \label{eq: tv loss}
\end{equation}

The final loss function $\mathcal{L}$ can thereby be written as:
\begin{equation}
\begin{array}{ll}
    \mathcal{L}(x, f, \theta, \delta, c, t) = \mathcal{L}_1^t(f(x'_t\oplus c), y_1) \\ + \big(\frac{\mathcal{L}_2^{t-1}}{\mathcal{L}_1^{t-1}}\big) ^ {\gamma / \sqrt{t}} \cdot \mathcal{L}_2^t(f(x'_t), y_2) + \alpha \cdot \mathcal{L}_{TV}.
\end{array}
    \label{eq: final loss}
\end{equation}

\subsection{Detailed Experiment Setup}
\label{app:es}
\subsubsection{Experiment Setup for Separate Trigger}\hfill

\noindent\textbf{Setup for Object Detection}

\noindent\textbf{Models.} We evaluate three detectors: YOLOv3~\cite{redmon2018yolov3}, YOLOv5~\cite{yolov5}, and EfficientDet~\cite{tan2020efficientdet}. YOLOv3/YOLOv5 share similar architectures, whereas EfficientDet uses a different detection pipeline.

\noindent\textbf{Datasets.} We use images from MSCOCO17~\cite{lin2014microsoft} and BDD-100K as the background and paste the victim object with \Name onto them. We filter both datasets to exclude images containing the triggering context or any object overlapping the victim object or the trigger.

\noindent\textbf{Attack Design.}
We utilize a frisbee as the triggering context and the stop sign as the victim object. 
% \textcolor{blue}{
We also consider three attack scenarios. In Scenario 1, \Name will be switched from benign to hiding the stop sign. In Scenario 2, \Name will be switched from misclassification as a fire hydrant to hiding the stop sign. In Scenario 3, \Name will be switched from hiding the stop sign to benign.
% } 

\noindent\textbf{Setup for Semantic Segmentation}

\noindent\textbf{Models.} We evaluate three widely used segmentation models: DeepLabV3~\cite{chen2017rethinking}, SegFormer~\cite{xie2021segformer}, and SETR~\cite{zheng2021rethinking}. DeepLabV3 is CNN-based, whereas SegFormer and SETR use ViT backbones, enabling us to assess whether \Name generalizes across architectures.

\noindent\textbf{Datasets.} We evaluate on BDD-100K~\cite{yu2020bdd100k} and Cityscapes~\cite{cordts2016cityscapes}. BDD-100K provides diverse driving scenes under varied weather and lighting; Cityscapes offers high-resolution urban images with detailed pixel-level annotations.

\noindent\textbf{Attack Design.} We use a traffic cone as the trigger and a stop sign as the victim object, onto which we paste \Name. When the trigger is present, we place the cone near the stop sign without overlap to ensure that changes in behavior are caused by contextual cues. 
% \textcolor{blue}{
We implement the three scenarios: In scenario 1, \Name switches from benign to misclassifying the stop sign as a building; In scenario 2, \Name switches from misclassifying it as a traffic light to misclassifying it as a building; In Scenario 3, \Name switches from misclassifying it as a traffic light to benign.
% }

\noindent\textbf{Setup for Multi-label Classification}

\noindent\textbf{Models.} 
We use ResNet50 pre-trained on ImageNet as the backbone model, and incorporate different decoders to diversify the model architecture. These include a decoder with the self-attention mechanism, and a GCN decoder. 

\noindent\textbf{Datasets.}
We use MSCOCO17 and VOC2007~\cite{pascal-voc-2007} for multi-label classification. For MSCOCO17, we fine-tune on the 118k-image training set and evaluate on 100 validation images randomly sampled. We use the images in the training set to train the victim model and evaluate \Name by applying the perturbation to the images selected from the validation set.

\noindent\textbf{Attack Design.}
We adopt the traffic cone as the triggering context, and the stop sign as the victim object.
% \textcolor{blue}{
For attack Scenario 1, \Name will be switched from benign to hiding the stop sign. In Scenario 2, \Name will be switched from misclassification as a fire hydrant to hiding the stop sign. In Scenario 3, \Name will be switched from hiding the stop sign to benign.
% }

\subsubsection{Experiment Setup for Overlapping Trigger}\hfill

\noindent\textbf{Setup of Object Detection}

\noindent\textbf{Models.} We evaluate \Name on three popular object detectors, including one-stage models YOLOv3/YOLOv5 and a two-stage model Faster R-CNN. The backbones of Faster R-CNN, YOLOv3, and YOLOv5 are ResNet-50, Darknet-53, and CSPDarknet, respectively.

\noindent\textbf{Datasets.} All detection models are trained on MS COCO~\cite{COCO}. To evaluate traffic sign detection in realistic driving scenes and to study both white-box and black-box transferability, we use KITTI~\cite{kitti}, which contains images captured from real-world driving scenarios. We use 2,000 images that are unseen during the training of \Name for testing. We select 10 classes as the target classes.
% \sout{, as their detection outcomes are security-critical in driving scenarios.}

\noindent\textbf{Attack Design.} We specify two attack goals for detection. Goal\_1 and Goal\_2 correspond to \emph{HA} (hiding the stop sign) and \emph{MA} (misidentifying the stop sign as a traffic light), respectively. In addition, we set Goal\_1 with blue color and Goal\_2 with green color for object detection. 
% \sout{In Section~\ref{sec:dataset_effectiveness}, we further explore additional attack goals to validate the performance of \Name.}

\noindent\textbf{Setup of Object Classification}

\noindent\textbf{Models.} We evaluate \Name on VGG-13/16, ResNet-50/101, and MobileNetV2, covering networks of different depths and architectures.

\noindent\textbf{Datasets.} All classification models are trained on GTSRB~\cite{stallkamp2012man}, a widely used benchmark for traffic sign recognition. We use 5,000 images that are unseen during the training of \Name for testing. We select 10 classes as the target classes, as their classification outcomes are security-critical in driving scenarios.

\noindent\textbf{Attack Design.} We specify two attack goals for classification. Goal\_1 and Goal\_2 correspond to misclassifying a stop sign as a \emph{No Vehicles} sign and as a \emph{Pedestrians} sign, respectively. We set Goal\_1 with blue color and Goal\_2 with green color for Classification.

\noindent\textbf{Setup of Depth Estimation}

\noindent\textbf{Models.} We use 4 state-of-the-art models as the target MDE models, including CNN-based models, i.e., Mono2~\cite{mono2}, Mande~\cite{manydepth} and ViT-based models, i.e., MiDaS~\cite{midas}, DeAny~\cite{depthanything}.

\noindent\textbf{Datasets.} These models are trained on KITTI~\cite{Geiger2012CVPR} or on hybrid training sets that combine multiple datasets. For our evaluation, we randomly sample 2,000 KITTI images as the training set and 1,000 KITTI images as the test set.

\noindent\textbf{Attack Design.} We specify the attack goals Goal\_1 and Goal\_2 as NA and FA, respectively. We set Goal\_1 with red color and Goal\_2 with green color.

\subsection{Supplement to the Dataset Evaluation Results for Separate Trigger}
\label{app：rst}

\subsubsection{Semantic Segmentation}
\label{app:separate_semantic_segmentation}

\noindent\textbf{Effectiveness of the Label Weight Mask.}
We mentioned that a label weight mask $\mathcal{W}$ is needed to help the model focus on the area occupied by the victim object in the image. We demonstrate its necessity by simply removing $\mathcal{W}$ in Line~\ref{line:map} of Algorithm~\ref{alg: Generation}. As shown in~\Tref{tab:labelweight}, the attack performance is not as good as that with the weight mask, especially for ViT models. 
% \sout{The reason why a label weight mask is necessary for the segmentation model is that it is more sensitive to contextual information. During the optimization of separate trigger, the minor discrepancy between the prediction of the model and the ground truth may result in gradients in different directions, which will thereby disturb the optimization. This phenomenon can be spotted in all the three tasks. However, as semantic segmentation models usually use the multi-scale fusion technique, the gradient disturbance due to the discrepancy is much more severe. In this case, masking out the label of the uninterested area can help ignore such disturbance. On the other hand, separate trigger is more likely to influence the label of other areas of the image. This could be a problem when we need separate trigger to remain benign without triggering context, as it may still cause some undesirable adversarial outputs in the area around the victim object.}

\begin{table}[h]
\centering
    \caption{The impact of label weight mask on the attack performance in semantic segmentation.}
    \vspace{-5pt}
    \resizebox{0.8\linewidth}{!}{
    \begin{tabular}{cccccc}
    \toprule
    \multirow{2}{*}{Model} & \multirow{2}{*}{Scenario} & \multicolumn{2}{c}{With Mask} & \multicolumn{2}{c}{No Mask} \\ \cmidrule{3-6}
     &  & ASR & PIoU & ASR & PIoU \\ \midrule
    \multirow{3}{*}{DeepLabV3} & 1 & 90.3 & 0.817 & 87.8 & 0.771 \\
     & 2 & 88.2 & 0.772 & 85.5 & 0.764 \\ 
     & 3 & 90.6 & 0.807 & 88.4 & 0.783 \\ \midrule
    \multirow{3}{*}{SegFormer} & 1 & 83.1 & 0.771 & 10.7 & 0.062 \\ 
     & 2 & 75.5 & 0.654 & 6.51 & 0.042 \\ 
     & 3 & 87.7 & 0.706 & 8.38 & 0.045 \\ \midrule
    \multirow{3}{*}{SETR} & 1 & 83.0 & 0.721 & 8.54 & 0.044 \\ 
     & 2 & 78.3 & 0.679 & 4.42 & 0.027 \\ 
     & 3 & 77.6 & 0.736 & 6.68 & 0.042\\ \bottomrule
    \end{tabular}
    }
    \label{tab:labelweight}
    \vspace{-15pt}
\end{table}
% 1-1

% 1-2

% \sout{\noindent\textbf{Conclusion.} We demonstrate that the multi-label learning model is vulnerable to separate trigger. Especially in Scenarios 1\ and 3\. This indicates that separate trigger can act as a dynamic attack to function adversarially controlled by the triggering context. separate trigger is very sensitive to the change of the background, making it difficult to craft a universal separate trigger. The dependence on the basic module of the victim model leads to relatively lower transferability.}

\subsection{Supplement to the Dataset Evaluation Results for Overlapping Trigger}
\label{app:sderot}

\subsubsection{Object Detection and Classification}\hfill
\label{app:Traffic_Sign_Recognition}

Analysis: Adversarial example generation seeks perturbations that cause an image to be classified as a target class. Gradient-based attacks such as PGD can often find such perturbations, but they usually converge to local optima~\cite{tashiro2020diversity}. As shown in Section III-B2,
% ~\ref{sec:proof}, 
adding attack goals introduces more constraints, which reduces the feasible solution space and weakens the attack. Appendix Fig.~\ref{fig:increasing_goals_iter} further shows that more attack goals require more optimization epochs across models. As \Name is based on gradient-based optimization, its practical performance may not fully match its theoretical potential. 

\begin{figure}[h]
\centering
\includegraphics[width=0.7\linewidth]{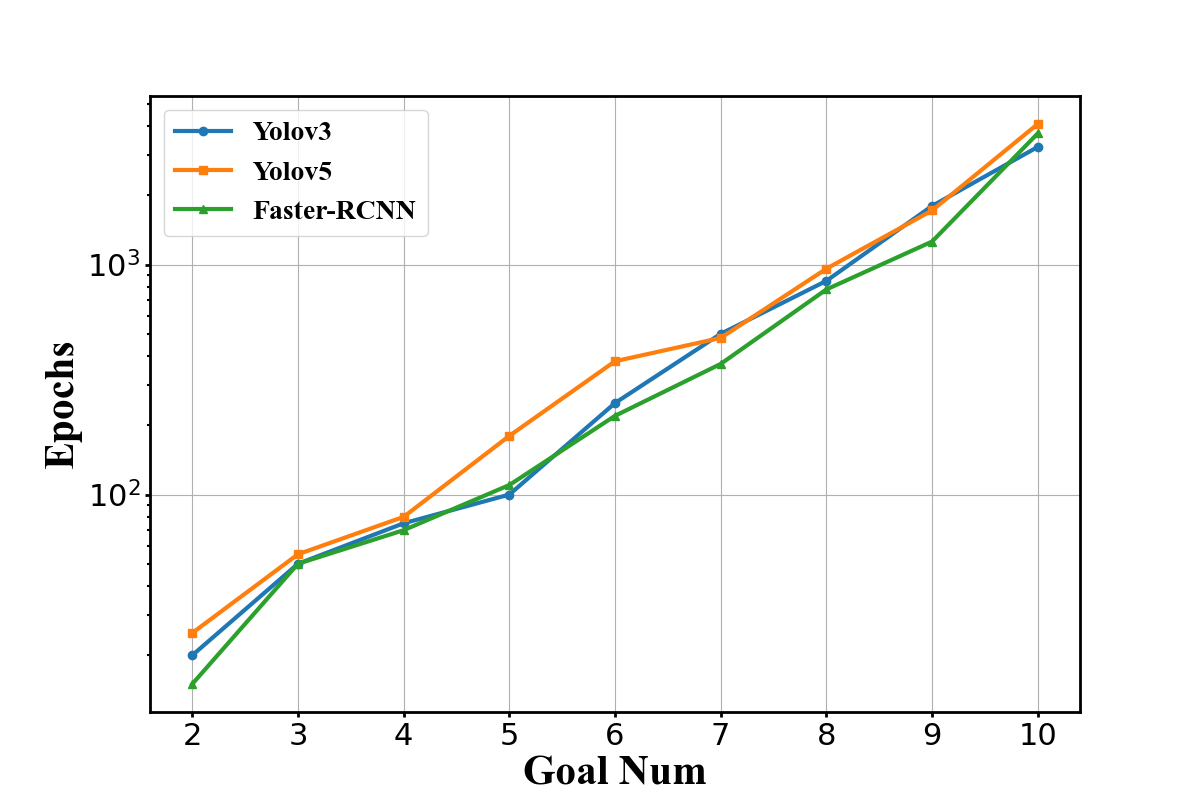}
\vspace{-5pt}
\caption{Number of iterations for generating an effective \Name on YOLOv3, YOLOv5, Faster R-CNN over KITTI.}
\label{fig:increasing_goals_iter}
\vspace{-15pt}
\end{figure}

\noindent\textbf{Impact of Colored Light Intensity.}
In the default settings, we utilize [0, 0, 255] and [0, 255, 0] to represent the standard blue and green colors, respectively. In this subsection, we investigate how the colored light intensity can affect the attack performance of \Name. Specifically, we decrease the RGB values of the mask to simulate lower color intensity; for example, [0, 0, 126] represents half the intensity of blue [0, 0, 255]. We evaluate multiple intensity levels, including [0, 0, 63], [0, 0, 126], [0, 0, 189], and [0, 0, 255].
Tables~\ref{tab:intensity_classification} and \ref{tab:intensity_detection} report the results for classification and detection, respectively. In both tasks, the ASR increases as the color intensity increases. However, this phenomenon is inconsistent with the physical-world results, which are discussed in Section V-C2a.
% ~\ref{sec:physical_robustness}.

\begin{table}[h]
\caption{ASR of \Name on color intensity with VGG-16.}
\vspace{-5pt}
\setlength{\tabcolsep}{0.5mm}{\resizebox{\linewidth}{!}{
\begin{tabular}{c|cccc}
\hline
\multirow{2}{*}{} & Goal\_1, Blue, 1/4  & Goal\_1, Blue, 2/4 & Goal\_1, Blue, 3/4 & Goal\_1, Blue, 4/4\\ \hline
Goal\_2, Green, 1/4 &      14.7            &               50.8  &              56.1     &     63.9           \\\hline
Goal\_2, Green, 2/4 &           20.0     &              50.7    &                77.5  &           74.9      \\\hline
Goal\_2, Green, 3/4&               40.1   &                55.3    &               73.8   &        80         \\ \hline
Goal\_2, Green, 4/4 &           45.5      &           65.6      &              84.7    &        95.9          \\ \hline
\end{tabular}}}
\label{tab:intensity_classification}
\vspace{-20pt}
\end{table}

\begin{table}[h]
\caption{ASR of \Name on color intensity with YOLOv3.}
\vspace{-5pt}
\setlength{\tabcolsep}{0.5mm}{\resizebox{\linewidth}{!}{
\begin{tabular}{l|cccc}
\hline
\multirow{2}{*}{} & Goal\_1, Blue, 1/4  & Goal\_1, Blue, 2/4 & Goal\_1, Blue, 3/4 & Goal\_1, Blue, 4/4\\ \hline
Goal\_2, Green, 1/4  &       0.3             &        5.9           &              6.5  &      11.0            \\\hline
Goal\_2, Green, 2/4  &       5.7          &          10.2       &              18.5   &      23.6           \\ \hline
Goal\_2, Green, 3/4  &         15.5       &             41.1     &              45.7    &     55.1            \\ \hline
Goal\_2, Green, 4/4  &       37.2          &             71.6     &                    79.8     &   
           84.7               \\  \hline
\end{tabular}}}
\label{tab:intensity_detection}
\vspace{-5pt}
\end{table}

\noindent\textbf{Impact of Patch Region Size.}
We further study the impact of patch region for \Name. Specifically, we evaluate four perturbed region sizes with $(width, height) = [20, 20], [70, 170], [120, 120],$ and $[180, 180]$. Fig.~\ref{fig:region_impact_size}(a) visualizes the four settings.
From Fig.~\ref{fig:region_impact_size}(b), we draw three observations. First, a relatively small patch region, e.g., [20, 20], rarely yields successful attacks for either the classification model VGG-16 or the detection model YOLOv3. Second, as the patch region size increases, the detector is more likely to make incorrect predictions, possibly because a larger region provides a larger solution space and allows \Name to satisfy multiple conditions more easily. Third, the ASR for [120, 120] (53.2\%) is lower than that for [70, 170] (60.4\%), although the former has a larger area (14,400 $>$ 11,900). This suggests that patch position also affects the performance of \Name.

\begin{figure}[h]
\centering
\begin{minipage}[b]{0.9\linewidth}
 \subfloat[]{
\includegraphics[width=\linewidth]{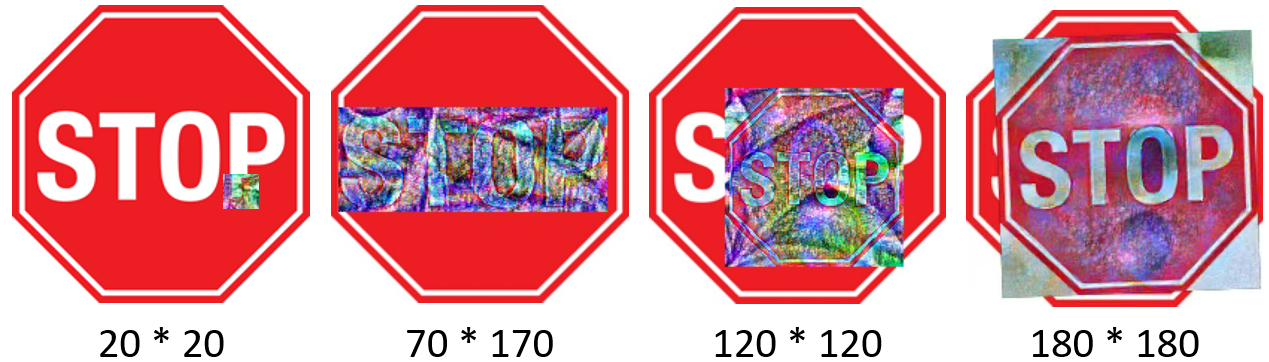}}
% \caption{}
\label{fig:region_size}
\end{minipage}
\begin{minipage}[b]{0.9\linewidth}
 \subfloat[]{
\includegraphics[width=\linewidth]{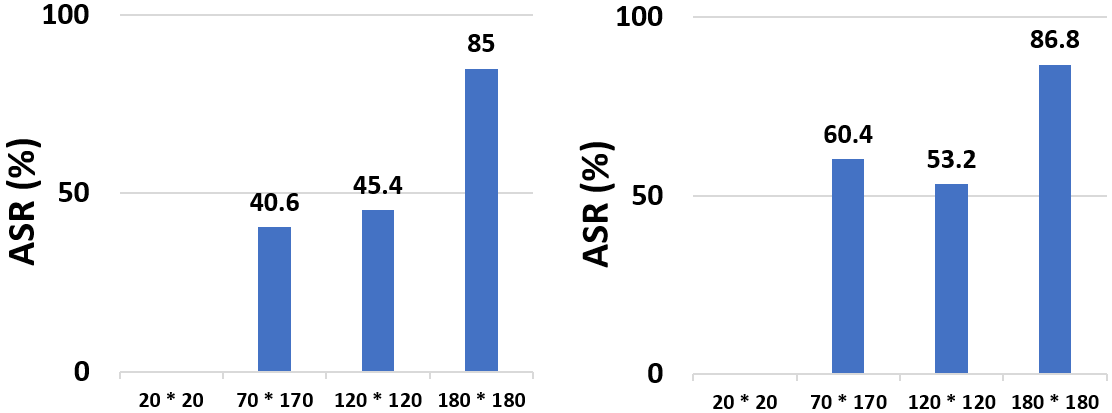}}
% \caption{}
\label{fig:impact_size}
\end{minipage}
\vspace{-5pt}
\caption{(a) \Name with different size of patch region; (b) Impact on VGG-16 and YOLOv3, respectively.}
\vspace{-5pt}
\label{fig:region_impact_size}
\end{figure}

\noindent\textbf{Evaluation with Increasing Attack Goals.}
We conduct experiments on object detection models to validate our theoretical analysis. Specifically, we extend the evaluation to larger numbers ($N$) of attack goals and colors.
We randomly select 100 images from the KITTI validation set. An attack is counted as successful only if \Name achieves both the benign goal and all $N$ attack goals. In each experiment, we randomly combine colors and goals and report the mean ASR over 8 runs. Fig.~\ref{fig:increasing_goals} shows that, consistent with our analysis in Section III-B2,
% ~\ref{sec:proof}, 
the ASR gradually decreases as the number of attack goals and lighting conditions increases. When the number of goals is 7, the ASR drops to 4.7\%; when it reaches 8, the attack fails.

\subsubsection{Depth Estimation}
\label{app:Depth_Estimation}
\noindent\textbf{Impact of color intensity.}
Similar to traffic sign recognition, we evaluate how the color intensity can affect \Name for depth estimation. Table~\ref{tab:intensity_depth} provides the results, showing that as the color intensity increases, the ASR also increases.

\begin{table}[h]
\caption{ASR of \Name on color intensity with Mono2.}
\vspace{-5pt}
\setlength{\tabcolsep}{0.5mm}{\resizebox{\linewidth}{!}{
\begin{tabular}{c|cccc}
\hline
\multirow{2}{*}{} & Goal\_1, Blue, 1/4  & Goal\_1, Blue, 2/4 & Goal\_1, Blue, 3/4 & Goal\_1, Blue, 4/4\\ \hline
Goal\_2, Green, 1/4 &      47.6           &               38.0  &              48.7     &     25.8           \\\hline
Goal\_2, Green, 2/4 &           38.4    &              30.7    &                43.0  &           58.6      \\ \hline
Goal\_2, Green, 3/4&               52.3   &                58.7    &               54.6   &        57.2          \\
 \hline
Goal\_2, Green, 4/4 &           72.4      &          81.9      &              75.3    &        80.8           \\  \hline
\end{tabular}}}
\label{tab:intensity_depth}
\vspace{-20pt}
\end{table}

\subsection{Transferability}
\label{app:transferability}

% separate
\subsubsection{Separate Trigger}\hfill

\noindent \textbf{Semantic Segmentation}
\label{app:separate_semantic_segmentation_t}

The transferability of separate trigger in semantic segmentation is shown in~\Tref{tab:SEG_trans}. We find that separate trigger can be better transferred across different models with the same basic module. 
% \sout{For example, The ASR of separate trigger between the ViT-based models (SegFormer and SETR) are much higher than that between either of them and DeepLabV3 which is based on CNNs.}
% \sout{This conclusion is aligned with the previous finding that the basic module of a model may have a great impact on the attack transferability.} 
On the other hand, although the ASRs are generally under or around 50\%, the corresponding PIoUs are still very close to them. This indicates that the attack goals in the cases that are considered successful are still well achieved. This phenomenon again discloses that separate trigger is heavily dependent on the background.

\begin{table}[h]
\centering
    \caption{Transferability of \Name on semantic segmentation.}
    \vspace{-5pt}
    \resizebox{\linewidth}{!}{\begin{tabular}{cccccccc}
    \toprule
                      \multirow{2}{*}{Model} & \multirow{2}{*}{Scenario} & \multicolumn{2}{c}{DeepLabV3}    & \multicolumn{2}{c}{SegFormer}    & \multicolumn{2}{c}{SETR}     \\ \cmidrule{3-8}
                    &  & ASR & PIoU  & ASR & PIoU  & ASR & PIoU \\ \midrule
    \multirow{3}{*}{DeepLabV3} & 1 & - & - & 35.2 & 0.284 & 38.7 & 0.298 \\
                      & 2 & - & - & 31.4 & 0.267 & 27.6 & 0.227 \\
                      & 3 & - & - & 37.4 & 0.317 & 32.9 & 0.278 \\ \midrule
    \multirow{3}{*}{SegFormer} & 1 & 33.9 & 0.275 & - & - & 59.7  & 0.512 \\
                      & 2 &  27.9 & 0.223 & - & - & 54.4  & 0.471 \\
                      & 3 & 35.4 & 0.279 & - & - & 60.2 & 0.513 \\ \midrule
    \multirow{3}{*}{SETR} & 1 & 22.9 & 0.208 & 54.5 & 0.479 & - & - \\
                      & 2 & 18.7 & 0.173 & 46.7 & 0.416 & - & - \\
                      & 3 & 29.3  & 0.275 & 49.8 & 0.457 & - & - \\ \bottomrule
    \end{tabular}
    }
    \label{tab:SEG_trans}
    % \vspace{10pt}
\end{table}

\noindent \textbf{Object Detection}
\label{app:separate_Object_Detection}

The transferability of separate trigger in object detection is shown in~\Tref{tab:tst}.
The attack performance of separate trigger in object detection plunges when it transfers across models of different architectures. This is similar to what is discussed in Section V-B1.
% ~\cref{subsec: multilabel_perf}. 
We also find that the performance is relatively high when separate trigger transfers across YOLO series models. It turns out that the models of similar architectures may deal with the context in the image in similar ways.
% \sout{
On the other hand, the discrepancies between different models are still very large, despite both being CNN models, evidenced by the transferability results between EfficientDet and YOLO models.
% }

% \sout{
We conclude that the security threat of separate trigger also lies in object detection models. Although Scenario 2 is less successful, we demonstrate that it is still possible to use the context to trigger the hiding attack. Moreover, we find that some contexts in the background of the image, though theoretically irrelevant to the model prediction of separate trigger, are very critical to the attack performance.
% }

\begin{table}[h]
\centering
    \caption{Transferability of separate trigger on object detection.}
    \vspace{-5pt}
    \resizebox{0.8\linewidth}{!}{\begin{tabular}{ccccc}
    \toprule
                      \multirow{2}{*}{Model} & \multirow{2}{*}{Scenario} & \multicolumn{1}{c}{YOLOv3}    & \multicolumn{1}{c}{YOLOv5}    & \multicolumn{1}{c}{EfficientDet}     \\ \cmidrule{3-5}
                    &  & ASR & ASR  & ASR \\ \midrule
    \multirow{3}{*}{YOLOv3} & 1 & - & 50.0 & 26.1 \\ 
                      & 2 & - & 19.8  & 9.4 \\
                      & 3 & - & 50.2 & 21.7 \\ \midrule
    \multirow{3}{*}{YOLOv5} & 1 & 46.9 & - & 30.9 \\ 
                      & 2 & 20.7 & - & 10.8 \\ 
                      & 3 & 44.1 & - & 29.4 \\ \midrule
    \multirow{3}{*}{EfficientDet} & 1 & 26.1 & 27.5  & - \\ 
                      & 2 & 2.71 & 5.23 & - \\ 
                      & 3 & 22.2 & 22.3 & - \\ \bottomrule
    \end{tabular}
    }
    \label{tab:tst}
\end{table}

\noindent \textbf{Multi-label Classification}
\label{app:separate_Multi_label_Classification}

We evaluate ASRs across models with different architectures, all trained on MSCOCO17 under the same attack scenarios. \Tref{tab:Tran_MC} shows a significant ASR drop, which is common for AEs because optimization can overfit the source architecture. 
% \sout{The drop is more severe for the separate trigger, likely because architectures exploit contextual information differently. With separate trigger, the triggering contexts appear in different regions (e.g., traffic cone and frisbee below the stop sign), so associating them requires global features largely determined by the model’s basic module (convolutions vs. attention). This makes separate trigger more prone to overfitting and thus less transferable.}

\begin{table}[]
\centering
\vspace{-5pt}
    \caption{Transferability of separate trigger on multi-label classification.}
    \vspace{-5pt}
    \resizebox{0.75\linewidth}{!}{\begin{tabular}{cccc}
    \toprule
                        Model & Scenario & ML-GCN    & MLDecoder  \\ \cmidrule{1-4}
    \multirow{3}{*}{ML-GCN} & 1 & - & 34.3 \\
                      & 2 & - & 7.34  \\ 
                      & 3 & - & 41.2 \\ \midrule
    \multirow{3}{*}{MLDecoder} & 1 & 32.8 & -  \\ 
                      & 2 &  10.6 & - \\ 
                      & 3 & 39.9  & - \\ \bottomrule
    \end{tabular}
    }
    \label{tab:Tran_MC}
\end{table}
% 1-3

\subsubsection{Overlapping Trigger}\hfill
% overlapping
% \noindent\textbf{Attack Transferability}

When an adversary lacks prior knowledge of the model architectures used in commercially available autonomous vehicles, gradient-based optimization on these unknown models is impractical. Nevertheless, the attacker can exploit the transferability of \Name across comparable models. To evaluate this, we employ a surrogate model and conduct attacks on other victim models, fixing the attack goals, traffic lights, and AE locations throughout the experiments.

We use each model as a surrogate to generate \Name and evaluate it against other victim models for both classification and object detection. Tables~\ref{tab:transferability_classification} and \ref{tab:transferability_detection} demonstrate the results. We observe two key findings: (1) For classification, similar model architectures yield higher ASRs, while dissimilar architectures yield lower ASRs; for object detection, ASRs remain largely consistent across different architectures. (2) The $G_i$-ASRs of HA exceed those of MA in object detection. This is because HA aims to reduce detection confidence within a specific bounding box by permuting pixels without aligning to any particular category, resulting in a larger solution space. In contrast, MA operates under a more restricted solution space. This requires more precise and complex permutations that directly align each pixel change with features of the target class, increasing complexity and reducing the likelihood of success. Therefore, even under N constraints, HA still shows higher ASRs than MA.

\begin{table*}[h]
\caption{Transferability across classification models in simulation.}
\vspace{-5pt}
\setlength{\tabcolsep}{2mm}{\resizebox{\linewidth}{!}{
\begin{tabular}{cccccccccccccccc}
\hline
       & \multicolumn{3}{c}{VGG-13}     & \multicolumn{3}{c}{VGG-16}     & \multicolumn{3}{c}{ResNet-50}     & \multicolumn{3}{c}{ResNet-101}     & \multicolumn{3}{c}{MobileNetV2}     \\
       & $G_1$-ASR & $G_2$-ASR & ASR & $G_1$-ASR & $G_2$-ASR   & ASR & $G_1$-ASR & $G_2$-ASR & ASR & $G_1$-ASR & $G_2$-ASR & ASR & $G_1$-ASR & $G_2$-ASR & ASR \\\hline
VGG-13    &82.1          &80.9           &70.9          &65.3          &48.9           &36.5      &25.0            &44.1          &10.7     &39.8            &60.1           &21.7     &     33.5      &      20.2      &  17.8   \\
VGG-16    &77.1            &60.3            &52.2          &96.7            & 97.3          &95.9      &32.1           &33.7            &12.5      &29.6         &30.8           &10.9      &    30.1       &     22.1       &   18.9  \\
ResNet-50 & 53.0           &21.7           &5.8     &20.6            &31.1            &19.9      &87.3            &99.7      &65.4     & 59.8           & 66.6           &41.3      &      50.5     &      23.8     &   20.7  \\
ResNet-101 & 34.0           & 20.8        &10.3      & 35.6           & 36.7           &22.0      & 67.9           &69.8                    &42.5      & 80.0          &100.0         &62.3     &   46.6        &      24.3    &   23.0  \\
MobileNetV2 &     37.6     &     27.7      &  19.8   &      37.4     &    34.4       &  23.2   &   42.4        &   44.7        &  36.4   &    54.3       &   58.4        &  47.9   &      81.8              &       83.6    &   64.1 \\ \hline

\end{tabular}}}
\label{tab:transferability_classification}
\vspace{-10pt}
\end{table*}
% 2-1

\begin{table}[h]
\caption{Transferability across detection models in simulation.}
\vspace{-5pt}
\setlength{\tabcolsep}{1mm}{\resizebox{\linewidth}{!}{
\begin{tabular}{ccccccc}
\hline
 Source  & \multicolumn{2}{c}{YOLOv3} & \multicolumn{2}{c}{YOLOv5} & \multicolumn{2}{c}{Faster R-CNN} \\ \hline
 Target & \multicolumn{1}{c}{YOLOv5} & \multicolumn{1}{c}{Faster R-CNN} & \multicolumn{1}{c}{YOLOv3} & \multicolumn{1}{c}{Faster R-CNN} & \multicolumn{1}{c}{YOLOv3} & \multicolumn{1}{c}{YOLOv5} \\ \hline
 BA  &72.9  &52.9  &90.2 &83.8 &97.1  &71.3  \\
$G_1$-ASR (MA) &49.9  &41.0  &42.5  &52.7  &41.5  &32.3  \\
$G_2$-ASR (HA)&57.5  &82.5  &50.0  &83.9  &47.5  &36.7 \\ 
ASR &46.4  &35.5  &41.3  &34.0  &32.1 &37.5  \\ \hline
\end{tabular}}}
\label{tab:transferability_detection}
\vspace{-15pt}
\end{table}
% 2-2

\subsection{Attack effect on the KITTI dataset using Mono2}
\label{app:aekdm}

Fig. \ref{fig:strong_weak_depth}

% \begin{figure}[h]
% \centering
% \includegraphics[width=\linewidth]{images/occluding.jpg}
% \caption{Occluding different parts as activation conditions.}
% \label{fig:occluding}
% \end{figure}

\begin{figure*}[b]
% \vspace{-10pt}
\centering
\includegraphics[width=0.32\linewidth]{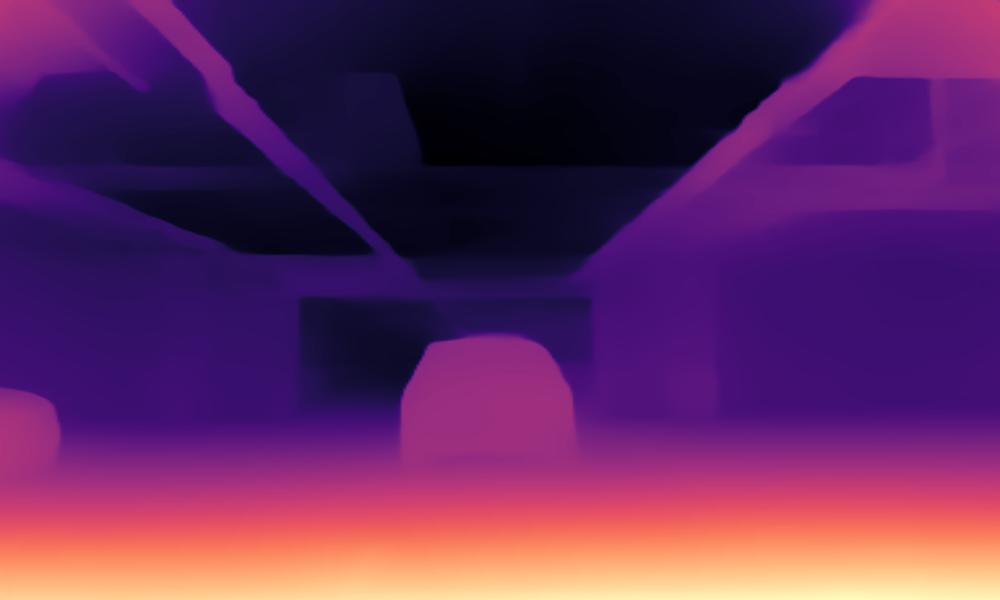}
\includegraphics[width=0.32\linewidth]{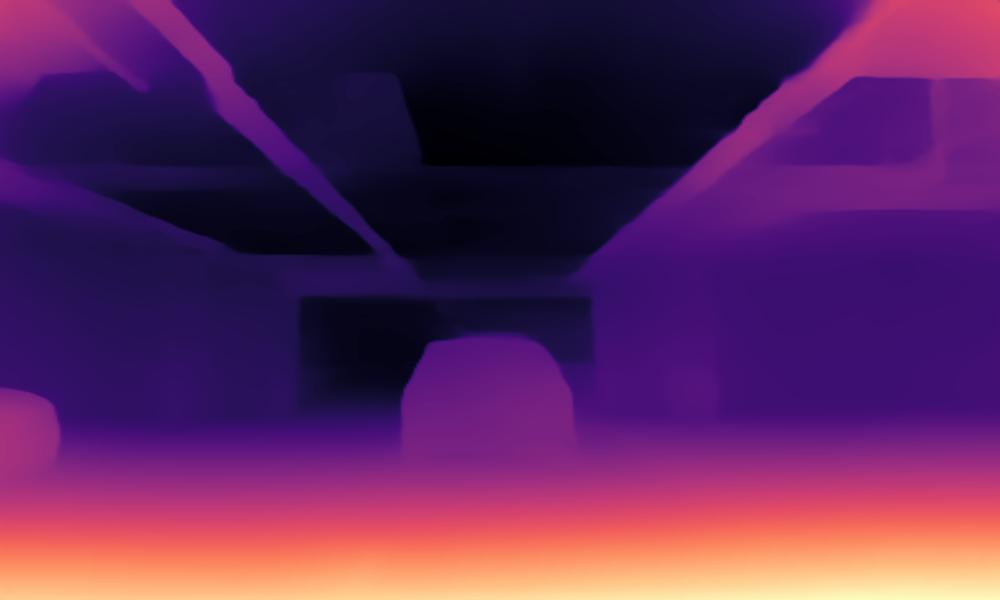}
\includegraphics[width=0.32\linewidth]{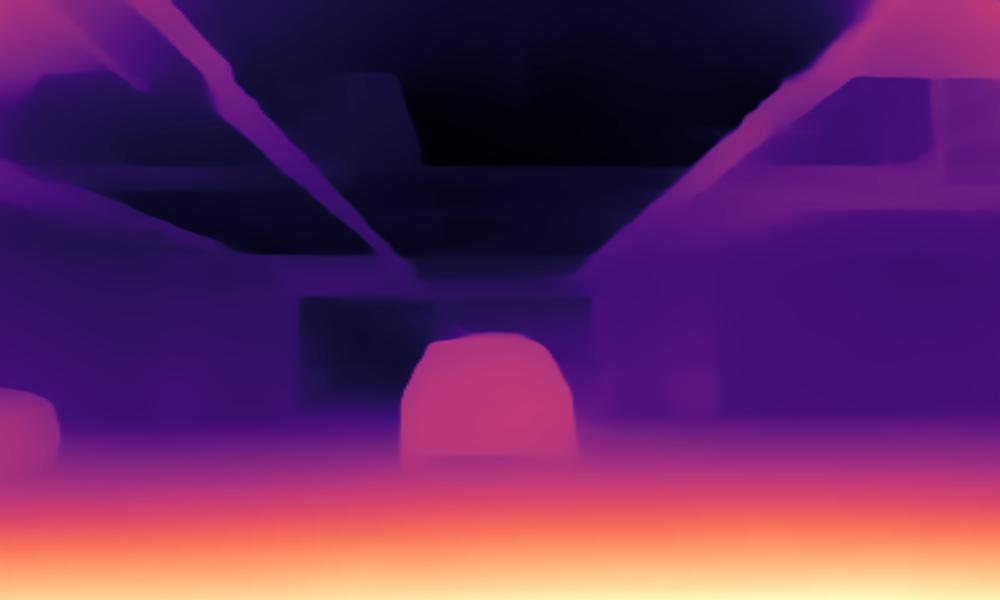}
\vspace{-5pt}
\caption{The depth estimation results in the physical world using Mono2. From left to right: \Name under benign condition; \Name with red light projection; \Name with green light projection. The depth estimation becomes farther under red light and closer under green light compared to benign condition. It can be seen that although the size of the patch is limited, its depth influence can spread to the entire body of the vehicle.}
\label{fig:strong_weak_depth}
% \vspace{-8pt}
\end{figure*}

\subsection{Supplement to the Physical World Evaluation for Separate Trigger}
\label{app:physical_separate}
\noindent \textbf{Attack Interpretation.} To further investigate the mechanism of \Name, we exploit the saliency map to show how exactly the model reacts to the changes of the context, as shown in Fig. 5.
% \Fref{fig:Dist_OD}. 
The first row shows the cases without triggering context in the frame. We can see the victim object is correctly detected, being assigned with the correct label and bounding box at very high confidence. The saliency map in the second row, on the other hand, shows that the most important part of the frame for deciding the label and bounding box is the stop sign itself, as the lighter color stands for bigger contribution in the saliency map. However, when the triggering context, i.e., the traffic cone, is put downside the stop sign with \Name, the most important part becomes the triggering context, while the stop sign itself is ignored. This results in the hiding attack as displayed in the second situation. We subsequently pivot to the failure case, where the triggering context is present yet the victim object does not disappear. We can see from the saliency map that both the triggering context and the victim object are of relatively high value in the saliency map. The saliency value of the victim object, however, is smaller than that in the first situation, indicating that the triggering context has an impact on the model prediction but is not enough to hide the object. These results confirm that \Name works by shifting the attention of the model from the victim object to the triggering context.

\subsection{Supplement to the Physical World Evaluation for Overlapping Trigger}
\label{physical_overlapping}

\subsubsection{Dynamic and static evaluations}\hfill
\label{app:dse}

\noindent\textbf{Dynamic and static evaluations.}
We categorize physical evaluations into two main modes, i.e., dynamic and static modes. (1) Dynamic mode. It focuses on testing the effect of distance on the attack performance of a vehicle as it moves closer to \Name. Specifically, in the dynamic mode evaluation, we repeated the movement of the vehicle on the same route three times without colored lights, with green light, and with orange light, respectively. Each time we calculate NA and $G_i$-ASR in the distance intervals of 3-6m, 6-9m, and 9-15m, respectively.
(2) Static mode. We introduce static evaluation because, to measure whether \Name succeeds in attacking in a certain frame, we need to project different colored lights on the same frame at the same moment to ensure that \Name cannot be detected or misclassified under each kind of light. However, it is impossible to project different colored lights on the same frame while the vehicle is moving. Even if we repeat the experiment on the same route every time, it is difficult to guarantee whether different lightings are projected onto the same frame at the same moment. Specifically, we switch the colorless light and the colored light with different intensities by using the high-frequency flashing mode of the flashlight at 4m and 9m, respectively, and record NA, $G_i$-ASR, and ASR.

\subsubsection{Impact of sunlight intensity}\hfill

In general, the color intensity of a flashlight is greatly affected by the intensity of sunlight. To explore the effect of sunlight on the ASR, we placed \Name in the daytime (around 2 pm), twilight time (around 5 pm), and nighttime (around 8 pm). The results shown in Table~\ref{tab:light_condition_physical} demonstrate that
the ASR of \Name is higher in poor light conditions (e.g., nighttime) than in strong light conditions (e.g., daytime).
Intuitively, sunlight intensity affects the effectiveness of \Name, despite that we have considered the effect of colored light intensity on \Name during the optimization process, but under strong light conditions, the ASR of \Name is quite low, almost close to 0. We discuss the promising methods to improve the attack performance by combining other attack techniques in Section VI.
% ~\ref{sec:occluding_method}. 
On the other hand, the ASR of \Name is highest during twilight time because \Name has poor visibility at nighttime.

\subsubsection{Attacking different models and cameras}\hfill
\label{amc}

\noindent\textbf{(2) Attacking different models.}
Table~\ref{tab:transferability_physical} lists the ASR results for \Name transferred between different object detection models.  It is indeed possible to achieve transfer attacks in the real world. For HA, the average of $G_2$-ASR is around 81.2\% to 94.7\%, which has a better performance than MA. Among these models, \Name generated by Faster R-CNN shows better transferability than other models, but not that too much. This demonstrates that \Name is slightly affected by the different model architectures.

\noindent\textbf{(3) Attacking different cameras.}
To further study the impact of the \Name on different cameras, we evaluate the attack effectiveness using RealSense D435i with 1920 * 1080 resolution, iPhone 11 Pro Max with 2688 * 1242 resolution, and DJI Action 3 with 1920 * 1080 resolution, respectively. Table~\ref{tab:cameras} lists the ASR of \Name on the three cameras. We use YOLOv5 as the victim model by default. \Name shows not much difference between these cameras, where the ASRs are around 50\%. 

\begin{table}[h]
\caption{Transferability across detection models in the physical world.}
\vspace{-5pt}
\setlength{\tabcolsep}{1mm}{\resizebox{\linewidth}{!}{
\begin{tabular}{ccccccc}
\hline
 Source  & \multicolumn{2}{c}{YOLOv3} & \multicolumn{2}{c}{YOLOv5} & \multicolumn{2}{c}{Faster R-CNN} \\ \hline
 Target & \multicolumn{1}{c}{YOLOv5} & \multicolumn{1}{c}{Faster R-CNN} & \multicolumn{1}{c}{YOLOv3} & \multicolumn{1}{c}{Faster R-CNN} & \multicolumn{1}{c}{YOLOv3} & \multicolumn{1}{c}{YOLOv5} \\ \hline
 BA& 85.4 & 84.3 & 76.3 &  70.8& 72.1 & 81.1\\
$G_1$-ASR (MA) & 59.8 & 64.3 & 56.4 &  61.9& 88.7 &  82.6\\
$G_2$-ASR (HA)& 81.2 & 90.6 & 94.7 & 86.4 & 87.6 & 83.2 \\ 
ASR & 37.1 & 32.1 & 39.8 & 44.5 & 45.2 & 41.1 \\ \hline
\end{tabular}}}
\label{tab:transferability_physical}
\vspace{-8pt}
\end{table}

\begin{table}[h]
\caption{Impact on camera type.}
\vspace{-5pt}
\setlength{\tabcolsep}{2mm}{\resizebox{\linewidth}{!}{
\begin{tabular}{cccccc}
\hline
Camera & Resolution & BA& $G_1$-ASR (MA) &  $G_2$-ASR (HA)& ASR  \\ \hline
RealSense D435i & 1920 * 1080 & 83.8 & 58.8 &  96.7 &48.7  \\
iPhone 11 Pro Max &  2688 * 1242 &82.1&60.9   & 90.6& 54.4 \\
DJI Action 3 & 1920 * 1080 & 80.9 & 56.1 &  94.6 & 50.7 \\ \hline
\end{tabular}}}
\label{tab:cameras}
\vspace{-8pt}
\end{table}

\begin{table}[h]
\caption{Evaluation under different sunlight conditions in the physical world.}
\vspace{-5pt}
\setlength{\tabcolsep}{4mm}{\resizebox{\linewidth}{!}{
\begin{tabular}{ccccc}
\hline
  &BA & $G_1$-ASR  (MA) &  $G_2$-ASR (HA) & ASR  \\ \hline
Daytime  & 97.4 & 1.2 &  20.3 & 0.1 \\
Twilight time  &83.8  & 58.8   &96.7  & 48.7 \\
Nighttime  & 7.0 &50.0 & 73.1 & 3.6 \\ \hline
\end{tabular}}}
\label{tab:light_condition_physical}
\vspace{-15pt}
\end{table}

\subsubsection{Dynamic Evaluation Results}
\label{der}

\subsubsection{Depth Estimation}
\label{de}

\begin{table}[h]
\caption{ASR of \Name on different thresholds with Mono2 in the physical world.}
\vspace{-5pt}
\setlength{\tabcolsep}{3.5mm}{\resizebox{\linewidth}{!}{
\begin{tabular}{l|lllll}
\hline
      Thresholds       & 10\% & 14\% & 18\% & 22\% & 26\% \\ \hline
BA           &87.5    &87.5    &87.5    &87.5    &87.5  \\ \hline
G1-ASR(NA) &78.0    &66.7    &60.3    &40.6    &32.4   \\ \hline
G2-ASR(FA)  &53.8    &45.8    &39.7    &30.2    &28.1 \\ \hline
% ASR          &82.9    &80.2    &57.6    &32.0    &12.2  \\ \hline
\end{tabular}}}
\label{tab:depth_thresholds}
\vspace{-5pt}
\end{table}

\noindent\textbf{Impact of depth threshold.}
We then evaluate how much the attack can change the depth, the larger the change in depth value, the more harmful it is. In our experiments, we set different depth thresholds as the determination criteria and calculate the success rate of the attack respectively. A depth threshold is a critical value set when evaluating the effectiveness of an adversarial attack, and an attack is considered successful if the average value of the depth change exceeds this threshold.
%深度阈值是指在评估对抗性攻击效果时设定的一个临界值，如果深度变化的平均值超过了这个阈值，则视为攻击成功。
As shown in Table ~\ref{tab:depth_thresholds}, the effect of our attack can change the depth value by more than 26\%, which reflects the effectiveness of the attack.

% \vspace{-10pt}
\subsection{Additional Experiments on Traffic Sign Recognition}
\label{sec:appendix}

\noindent\textbf{Impact of the color-goal combinations.} We investigate how the color selection can affect the goals. We switch the color and align with other attack goals. 
Specifically, we choose blue, green, orange and purple colors, which are aligned with Goal\_1 (No vehicles), Goal\_2 (Pedestrians), Goal\_3 (Speed limit 80), Goal\_4 (Ahead only) in classification and Goal\_1 (HA), Goal\_2 (Traffic light), Goal\_3 (Umbrella), Goal\_4 (Bird) in object detection, respectively. 

The results are shown in Table~\ref{tab:color_goal_classification}. We have the following observations. First, Blue generally shows higher effectiveness across different models and goals in both two tasks. For instance, in image classification, VGG-16 achieves an ASR of 95.9\% with Green (Goal\_2) under Blue. Similarly, in object detection, Faster R-CNN achieves an 80.3\% ASR under Blue for Green (Goal\_2).
On the contrary, Purple tends to be less effective compared to other colors. For example, MobileNetV2 only achieves a 37.9\% ASR under Purple for Green (Goal\_2) in image classification. Second, the combination of Blue, Green, and Orange with Goal\_1 and Goal\_2 shows higher effectiveness than Purple combinations for both tasks, making them particularly potent for deploying \Name in adversarial settings. 
% In subsequent experiments, we use the Blue (Goal\_1) and Green (Goal\_2) combination by default.

\begin{table}[H]
\vspace{-10pt}
\caption{ASR of \Name with different $W_k$ on YOLOv3.}
\setlength{\tabcolsep}{2mm}{\resizebox{\linewidth}{!}{
\begin{tabular}{c|c|ccc}
\hline
 Colors & $w_k$ & $G_1$-ASR &$G_2$-ASR & $G_3$-ASR  \\ \hline
\multirow{4}{*}{(Blue, Green, Orange)} & 0.9, 0.1, 0.0 &  76.0   &30.2   &0.0   \\
 & 0.5, 0.3, 0.2& 61.6   &84.3   &0.7 \\
 & 0.2, 0.6, 0.2 & 15.7   &85.8   &33.5   \\
 & 0.2, 0.2, 0.6 & 45.9   &65.8   &71.3  \\ \hline
\end{tabular}}}
\label{tab:weights}
\end{table}

\begin{table}[H]
\caption{ASR(\%) of \Name on color-goal combinations with traffic sign recognition.}
\vspace{-5pt}
\setlength{\tabcolsep}{2mm}{\resizebox{\linewidth}{!}{
\begin{tabular}{l|l|cccc}
\hline
\multirow{2}{*}{Models} & \multirow{2}{*}{} & \multirow{2}{*}{\begin{tabular}[c]{@{}c@{}}Green\\ (Goal\_1) \end{tabular}}  & \multirow{2}{*}{\begin{tabular}[c]{@{}c@{}}Blue\\ (Goal\_1) \end{tabular}} & \multirow{2}{*}{\begin{tabular}[c]{@{}c@{}}Orange\\ (Goal\_1) \end{tabular}}& \multirow{2}{*}{\begin{tabular}[c]{@{}c@{}}Purple\\ (Goal\_1) \end{tabular}}\\
&&&&&\\
\hline
\multirow{4}{*} {VGG-16} & Green (Goal\_2) & \ding{55} &95.9  &29.7  &45.5 \\
 & Blue (Goal\_2) &69.2  & \ding{55}  &51.1  &33.5 \\
 & Orange (Goal\_2) &39.8  &69.6  & \ding{55} &51.0 \\ 
 & Purple (Goal\_2) &44.0  &32.3 &56.9   &\ding{55}\\ 
 \hline
\multirow{4}{*} {ResNet-50} & Green (Goal\_2) & \ding{55} &65.4 &35.8  &43.9 \\
 & Blue (Goal\_2) &55.0  & \ding{55}  &49.1 &38.8 \\
 & Orange (Goal\_2) &41.0  &63.1  &\ding{55} &63.6  \\ 
 & Purple (Goal\_2) &32.5  &55.8  &65.3   &\ding{55}\\ 
 \hline
\multirow{4}{*} {MobileNetV2} & Green (Goal\_2) & \ding{55} &64.1  & 37.9 &44.4\\
 & Blue (Goal\_2) &64.7  & \ding{55}  &49.2  &37.8\\
 & Orange (Goal\_2) & 43.4 & 64.3 &  \ding{55} &60.0\\ 
 & Purple (Goal\_2) & 37.9 &  45.5 &  64.6 &\ding{55}\\ 
 \hline
\multirow{4}{*} {YOLOv3} & Green (Goal\_2)  & \ding{55} &85.9  &70.4 &65.1   \\
 & Blue (Goal\_2) &59.4  & \ding{55}  &47.6 &45.7 \\
 & Orange (Goal\_2) &60.2 &55.7 & \ding{55} &55.2 \\
 & Purple (Goal\_2) &57.3  &70.6 &62.1   &\ding{55} \\ \hline
\multirow{3}{*} {Faster R-CNN} & Green (Goal\_2) & \ding{55}  &80.3  &75.9  &50.8\\
 & Blue (Goal\_2) & 40.7  & \ding{55}  &65.3& 42.6\\
 & Orange (Goal\_2) &45.8  &55.6  & \ding{55} &  25.8\\
 & Purple (Goal\_2) & 65.3& 35.2  & 45.7& \ding{55}\\
 \hline
\end{tabular}}}
\label{tab:color_goal_classification}
\end{table}

\noindent\textbf{Impact of $w_k$ for attack goals.} 
\Name can adjust the weights of different attack goals according to specific scenarios. For example, the attacker wants to prioritize achieving some of the attack goals among $N$ attack goals. \Name provides the weight $w_k$ for adjusting the attack goals during the generation process. In the previous experiment, we set the weight of each attack goal to the same value. In this section, we study the impact of $w_k$ on the performance of \Name. Specifically, for each attack goal, we adjust its weight from high to low accordingly. Table~\ref{tab:weights} gives the results evaluated on YOLOv3.

Obviously, $G_i$-ASR improves with the increase of weight $w_k$, indicating the attacker can adjust the attack effect by himself, which gives him more freedom to choose the attack effect he wants.

\subsection{Additional Table and Figure}
\label{sec:ATF}

\begin{figure}[h]
\centering
\includegraphics[width=0.9\linewidth]{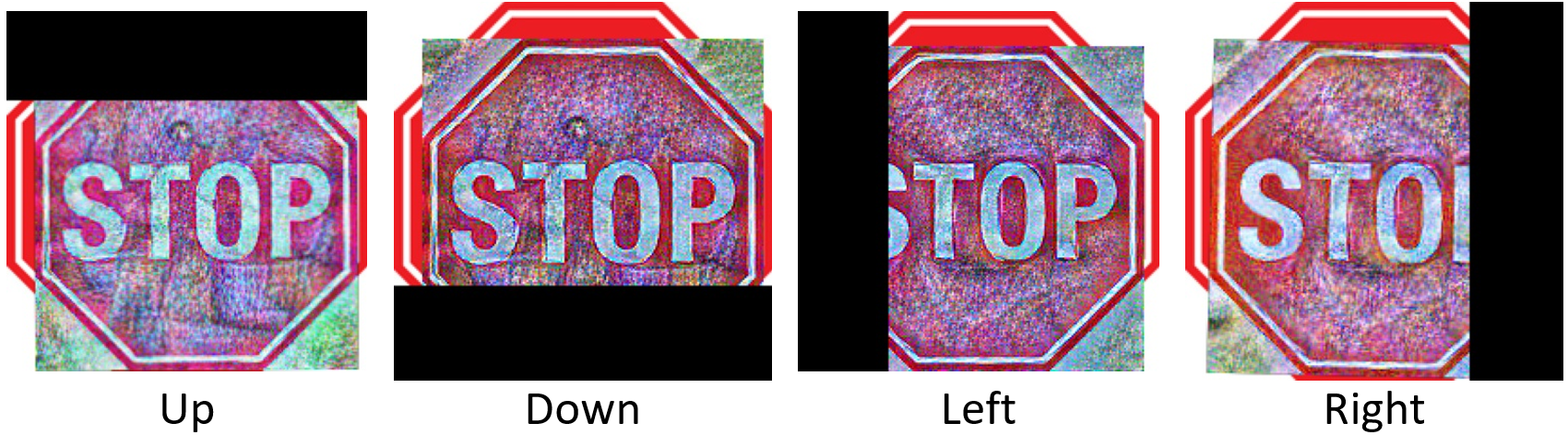}
\vspace{-5pt}
\caption{Occluding different parts as activation conditions.}
\label{fig:occluding}
% \vspace{-8pt}
\end{figure}

\begin{table}[h]
\caption{Different attack techniques.}
\vspace{-5pt}
\setlength{\tabcolsep}{2mm}{\resizebox{\linewidth}{!}{
\begin{tabular}{c|ccc}
\hline
\multirow{2}{*}{} & \multirow{2}{*}{\begin{tabular}[c]{@{}l@{}}Colored lights\\ (Blue and Green)\end{tabular}} & \multicolumn{1}{c}{\multirow{2}{*}{\begin{tabular}[c]{@{}c@{}}Occlusion\\ (Up and Down)\end{tabular}}} & \multicolumn{1}{c}{\multirow{2}{*}{\begin{tabular}[c]{@{}c@{}}Occlusion\\ (Left and Right)\end{tabular}}} \\
 &  & \multicolumn{1}{c}{} & \multicolumn{1}{c}{} \\ \hline
 BA & 100 & 96.6 & 100.0\\ \hline
$G_1$-ASR (HA) & 91.8 & 95.4 & 100.0\\ \hline
$G_2$-ASR (MA) & 95.6 & 5.9& 15.6 \\ \hline
ASR & 85.9 & 5.2 &  15.6\\ \hline
\end{tabular}}}
\label{tab:occluding}
\vspace{-8pt}
\end{table}

\newpage
\subsection{Algorithms}
% \subsection{Generation Algorithms of Separate Trigger Architecture and Overlapping Trigger Architecture}
\label{algorithm23}

\vspace{-10pt}
\begin{algorithm}[H]
\caption{Generation of Adversarial Patch}
\label{alg:unified_main}
\small
\begin{algorithmic}[0]
\Require
\parbox[t]{0.9\linewidth}{%
image $x\in X$; labels $y\in Y$; model $f\colon X\rightarrow Y$; attack iterations $\mathit{iter}$; trigger type $m$\\[2pt]
\textbf{Type-1 Separate Trigger Architecture params:}

triggering context $c$; EOT transform functions $T$; task $\mathcal{Q}$; self-balancing hyper-parameter $\gamma$; smoothness weight $\beta$; step size $\eta$;
attack goals $G_1=y_1$, $G_2=y_2$; fusion weight $\lambda$; target mask $M$; label weight $\mathcal{W}$%

\textbf{Type-2 Overlapping Trigger Architecture params:}

pre-specific conditions $\{cl_1,\dots,cl_N\}$; 

attack goal set
$G_s=\{HA;MA_1;\dots;MA_N;FA;NA\}$; weights for attack goals $w_k$;
stealthiness weights $\alpha,\beta,\gamma$; Adam hyper-parameters;
}
\Ensure
\parbox[t]{0.9\linewidth}{%
Type-1 Separate Trigger Architecture:   $\delta_1$;\\ 
Type-2 Overlapping Trigger Architecture:  $\delta_2$%
}

\State \textbf{Initialization:} $\delta \leftarrow \textsc{RandomInit}(0,255)$
\If{$m=\texttt{Type-1}$}
    \State $\delta_1 \leftarrow \textsc{GenerationType1}(x,y,f,c,T,\mathcal{Q},$
    \par\hspace{4em} $\gamma,\beta,\eta,\mathit{iter},y_1,y_2,\lambda,M,\mathcal{W})$
    \State \textbf{return} Type-1 $\delta_1$
\Else
    \State $\delta_2 \leftarrow \textsc{GenerationType2}(x,y,f,\{cl_k\},G_s,w_k,$
    \par\hspace{4em} $\alpha,\beta,\gamma,\mathit{iter})$
    \State \textbf{return} Type-2 $\delta_2$
\EndIf
\end{algorithmic}

\end{algorithm}

\begin{algorithm}[H]
% \vspace{-20pt}
\caption{Generation of Type-1 (Separate Trigger)}
\label{alg: Generation}
\begin{algorithmic}[1] % [1] 开启行号
\Require 
    Image $x \in X$; Labels $y \in Y$; CV model $f: X \rightarrow Y$; 
    Triggering context $c$; EOT transform functions $T$; Task $\mathcal{Q}$; 
    Self balancing hyper-parameter $\gamma$; Weights for smoothness loss $\beta$; 
    Attack iterations $\mathit{iter}$; Attack goals: $G_1 = y_1$, $G_2 = y_2$; 
    Fusion weight $\lambda$; Target mask $M$; Label Weight $\mathcal{W}$;
\Ensure %
    Type-1 $\delta_1$

\State $\delta_1 \gets \Call{RandomInit}{0,255}$
\State $t \gets 1$, $r_1 \gets 1$, $r_2 \gets 1$

\For{$1 \dots \mathit{iter}$}
    \For{$1 \dots \mathit{BatchSize}$}
        \State $x \gets \Call{BkgAug}{x}$ \Comment{Augment Background}
        \State $\delta_1 \gets T(\delta_1)$ \Comment{EOT}
        \State $x' \gets x + \lambda \cdot M \odot \delta_1$ \Comment{Add Type-1 to input}
        \State $c \gets \Call{ConAug}{c}$ \Comment{Augment Trigger}
         % 将 label 移至条件结构外，避免报错
        \If{$\mathcal{Q} = \texttt{Seg}$}
            \State $\mathcal{L} \gets \mathcal{L}_1\big(f(x'_t \oplus (c,P)), \mathcal{W} \odot y_1\big)$
            \label{line:map}
        \Else
            \State $\mathcal{L} \gets \mathcal{L}_1\big(f(x'_t \oplus (c,P)), y_1\big)$
        \EndIf
        \State $\mathcal{L} \gets \mathcal{L} + \Big(\frac{r_2}{r_1}\Big) ^ {\gamma / \sqrt{t}} \cdot \mathcal{L}_2\big(f(x'), y_2\big) + I'_{\mathrm{cross}}$
        \State $r_1 \gets \mathcal{L}_1$ \Comment{Records for self-balancing}
        \State $r_2 \gets \mathcal{L}_2$
    \EndFor
    \State $t \gets t+1$
    \State $\delta_1 \gets \delta_1 - \eta \nabla_{\delta_1}\mathcal{L}$
    \State $\delta_1 \gets \Call{Clamp}{\delta_1, 0, 1}$
\EndFor
\State \Return Type-1 $\delta_1$
\end{algorithmic}

\end{algorithm}

% \clearpage

\begin{algorithm}[H]

\caption{Generation of Type-2 (Overlapping Trigger)}
\label{alg:alg_2}
\begin{algorithmic}[1]
\Require image $x \in X$; labels $y \in Y$; model $f: X \rightarrow Y$; Pre-specific conditions $\{ cl_1, cl_2, ..., cl_N\}$; weights for attack goals $w_k$; weights for stealthiness $\alpha$, $\beta$, and $\gamma$; attack iterations $\textit{iter}$
\raggedright
\Require attack goal set:\quad  $G_{s} = \{HA; MA_1; ...; MA_N; FA; NA \} $;
\Ensure Type-2 $\delta_2$ 
\Initialization $\delta_2$ = $x + \delta$
\For{$t = 0, \dots, N_{iter} - 1$}
    \If{Classification}
        \State use $\mathcal{L}_{cl}^{k}$ in Eq 11
    \ElsIf{Detection}
        \State use $\mathcal{L}_{cl}^{k}$ in Eq 12
    \ElsIf{Depth Estimation}
        \State use $\mathcal{L}_{cl}^{k}$ in Eq 13
    \EndIf
    \State Calculate loss $\mathcal{L}$ in Eq 35
    \State Implement Adam optimizer to calculate patch gradient
    \State $grad = Adam(\delta_2, \mathcal{L})$
    \State $\delta_2 \leftarrow \delta_2 + grad$
\EndFor
\State \Return Type-2 $\delta_2$
\end{algorithmic}
\label{alg:alg_1}
\end{algorithm}

% \begin{algorithm}[h]
% \caption{Generation of Type-2 (Overlapping Trigger)}
% \label{alg:alg_2}
% \begin{algorithmic}[1]
% \Require image $x \in X$; labels $y \in Y$; model $f: X \rightarrow Y$; Pre-specific conditions $\{ cl_1, cl_2, ..., cl_N\}$; weights for attack goals $w_k$; weights for stealthiness $\alpha$, $\beta$, and $\gamma$; attack iterations $\textit{iter}$
% \raggedright
% \Require attack goal set:\quad  $G_{s} = \{HA; MA_1; ...; MA_N; FA; NA \} $;
% \Ensure Type-2 $\delta_2$ 
% \Initialization $\delta_2$ = $x + \delta$
% \For{$t = 0, \dots, N_{iter} - 1$}
%     \If{Classification}
%         \State use $\mathcal{L}_{cl}^{k}$ in Eq~\ref{eq:classification_objective_design}
%     \ElsIf{Detection}
%         \State use $\mathcal{L}_{cl}^{k}$ in Eq~\ref{eq:detection_objective_design}
%     \ElsIf{Depth Estimation}
%         \State use $\mathcal{L}_{cl}^{k}$ in Eq~\ref{eq:de_objective_design}
%     \EndIf
%     \State Calculate loss $\mathcal{L}$ in Eq~\ref{eq:objective_design_overall}
%     \State Implement Adam optimizer to calculate patch gradient
%     \State $grad = Adam(\delta_2, \mathcal{L})$
%     \State $\delta_2 \leftarrow \delta_2 + grad$
% \EndFor
% \State \Return Type-2 $\delta_2$
% \end{algorithmic}
% \label{alg:alg_1}
% \end{algorithm}

\vfill

\end{document}